\documentclass[11pt,a4paper]{article}
\pdfoutput=1

\usepackage{setspace,graphicx,amssymb,amsmath,mathtools,latexsym,amsfonts,amscd,amsthm,multirow,mathdots,caption,array,dsfont,wrapfig}
\usepackage{fancyhdr,tabularx,cite,mathrsfs,enumitem,graphicx}
\usepackage{stmaryrd}
\usepackage{rotating}
\usepackage{authblk}
\usepackage{hyperref}

\usepackage{comment,enumitem}
\usepackage{geometry}
\usepackage{arydshln}

\usepackage{hyperref}
\usepackage[font=small,labelfont=bf]{caption}

\usepackage{lscape}
\usepackage{tikz}
\usepackage{ctable}
\usetikzlibrary{calc,trees,positioning,arrows,chains,shapes.geometric,%
    decorations.pathreplacing,decorations.pathmorphing,shapes,%
    matrix,shapes.symbols}

\tikzset{
>=stealth',
  punktchain/.style={
    rectangle,
    draw=black, very thick,
    text width=15em,
    minimum height=2em,
    text centered,
    on chain},
  line/.style={draw, thick, <-},
  element/.style={
    tape,
    top color=white,
    bottom color=blue!50!black!60!,
    minimum width=8em,
    draw=blue!40!black!90, very thick,
    text width=10em,
    minimum height=3.5em,
    text centered,
    on chain},
  every join/.style={->, thick,shorten >=1pt},
  decoration={brace},
  tuborg/.style={decorate},
  tubnode/.style={midway, right=7pt},
}

\usepackage[all,cmtip]{xy}

\newcommand{\CC}{\mathbb{C}}
\newcommand{\RR}{\mathbb{R}}
\newcommand{\QQ}{\mathbb{Q}}

\newcommand{\HH}{\mathbb{H}}
\newcommand{\ZZ}{\mathbb{Z}}

\newcommand{\M}{\mathcal{M}}
\newcommand{\N}{\mathcal{N}}

\newcommand{\xmod}{{\rm \;mod\;}}

\def\a{\alpha}

\def\c{\gamma}

\def\f{\phi}

\def\im{\mathrm{Im}}

\def\inf{\infty}

\def\l{\lambda}
\def\m{\mu}
\def\n{\nu}

\def\p{\pi}
\def\pa{\partial}

\def\s{\sigma}
\def\t{\tau}

\def\th{\theta}
\def\til{\tilde}

\def\D{\Delta}

\def\O{\Omega}

\def\Tr{\tr}

\newcommand{\resi}{\operatorname{\rm{Res}}}
\newcommand{\SL}{\operatorname{\textsl{SL}}}      %SL group
\newcommand{\GL}{\operatorname{\textsl{GL}}}      %GL group
\newcommand{\ex}{\operatorname{e}} %Number theory exp

\newcommand{\tr}{\operatorname{tr}}

\newcommand{\Ex}{\operatorname{Ex}}

\usepackage{pifont}% http://ctan.org/pkg/pifont
\newtheorem{theorem}{Theorem}

\newtheorem{lemma}[theorem]{Lemma}
\theoremstyle{definition}
\newtheorem{defn}[theorem]{Definition}

\newcommand{\bea}{\begin{eqnarray}}
\newcommand{\eea}{\end{eqnarray}}
\newcommand{\bee}{\begin{eqnarray*}}
\newcommand{\eee}{\end{eqnarray*}}
\newcommand{\al}{\begin{align*}}
\newcommand{\eal}{\end{align*}}
\newcommand{\be}{\begin{equation}}
\newcommand{\ee}{\end{equation}}
\newcommand{\eq}[1]{(\ref{#1})}
\newcommand{\bem}{\begin{pmatrix}}
\newcommand{\eem}{\end{pmatrix}}

\numberwithin{equation}{section}

\begin{document}

\setstretch{1.2}

\centerline{\huge{3d Modularity}}
\bigskip
\bigskip
\centerline{\large{Miranda C. N. Cheng$^{a,b}$, Sungbong Chun$^c$, Francesca Ferrari$^{b,d}$, Sergei Gukov$^c$, Sarah M.~Harrison$^e$}}
\bigskip
\bigskip
\centerline{$^a$Institute of Physics and Korteweg-de Vries Institute for Mathematics}
\centerline{University of Amsterdam, Amsterdam, the Netherlands}
\medskip
\centerline{$^b$Institute of Physics, University of Amsterdam, Amsterdam, the Netherlands}
\medskip
\centerline{$^c$Walter Burke Institute for Theoretical Physics, California Institute of Technology,
}
\centerline{Pasadena, CA 91125, USA}
\medskip
\centerline{$^d$ International School for Advanced Studies (SISSA), Trieste, Italy
}
\medskip
\centerline{$^e$ Department of Mathematics and Statistics and Department of Physics, }
\centerline{McGill University, Montreal, QC, Canada}
\medskip

\bigskip
\begin{abstract}
We find and propose an explanation for a large variety of modularity-related symmetries
in problems of 3-manifold topology and physics of 3d $\N=2$ theories
where such structures {\it a priori} are not manifest.
These modular structures include: mock modular forms, $\SL(2,\ZZ)$ Weil representations,
quantum modular forms, non-semisimple modular tensor categories,
and chiral algebras of logarithmic CFTs.
\end{abstract}

\newpage
\tableofcontents

\newpage

\section{Introduction and summary}
\label{sec:intro}

This work relies on the interplay between different fields of research, including topology, physics and  number theory. As shown in Figure \ref{fig:bigpic}, each of the three fields asks different questions, brings in different results, and employs different techniques, which all turn out to be related and in fact crucial for one and other. The central object is a certain family of infinite $q$-series ``$\widehat Z_b(q)$'', which plays the role of supersymmetric indices, topological invariants, and quantum modular forms in physics, topology, and number theory respectively.
We hope the results can be of interest to the corresponding communities.
To facilitate this, the introduction is written from three points of view. That said, the readers are encouraged to read all of them to get a complete picture.

\begin{figure}[h]
\begin{center}
\includegraphics[width=5.5in]{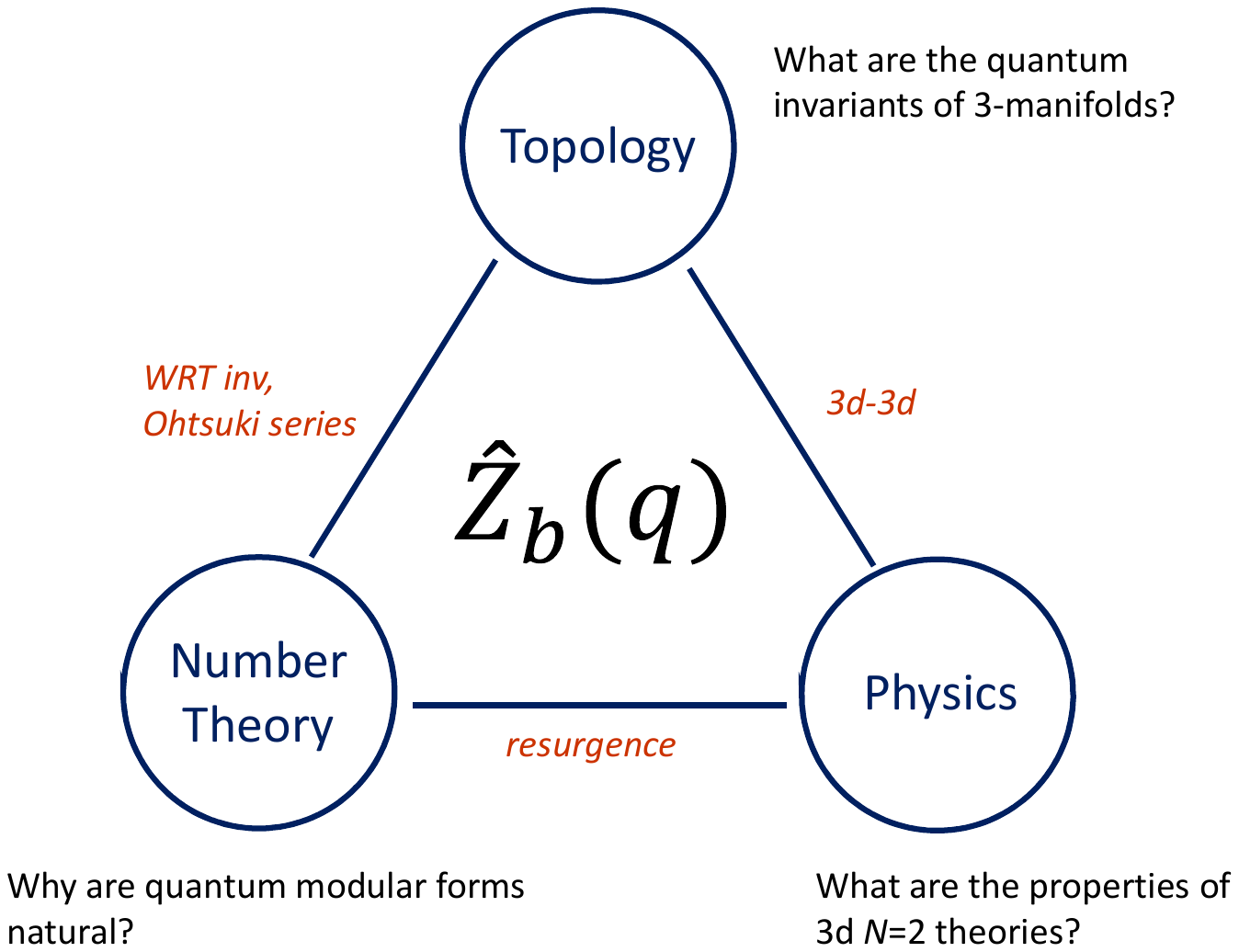}
\caption{The different topics involved in this paper.}
\label{fig:bigpic}
\end{center}
\end{figure}

\subsection{... for physicists}

In the past two decades, tremendous progress has been made in understanding strongly coupled
supersymmetric (SUSY) quantum field theories, even to the extent that insights coming from SUSY theories
motivate progress in non-supersymmetric theories.
In part, this progress is based on rapid development of
localization techniques in supersymmetric theories~\cite{Pestun:2016zxk},
which can be used to compute various partition functions and indices --- including the ones of interest in this paper --- exactly.

{}From the physics point of view, this paper is about a certain 3d analogue of
the famous elliptic genus \cite{Witten:1986bf}.
More precisely, we study the combined index of a 3d $\N=2$ supersymmetric
theory with a half-BPS boundary condition, originally introduced in \cite{Gadde:2013wq}.
While the elliptic genus of 2d $\N=(0,2)$ theories is known to be related to
the traditional theory of modular forms, the combined 3d-2d index
(sometimes called half-index or $D^2\times_q S^1$ partition function)
will be shown to exhibit more subtle and interesting types of modular behavior.
Specifically, in \S\ref{sec:3dphysics} we will discuss three types of modular-like behavior that can be displayed by the half-indices, with an increasing degree of subtlety as the bulk 3d theory becomes more and more non-trivial.

In the process, we also find a new and unexpected way in which 2d logarithmic conformal
field theories (log-CFTs) can arise from supersymmetric quantum field theories,
in fact, from three-dimensional theories!

\subsection{... for topologists}

{}From the point of view of topology, the present paper aims to make progress
on the following long-standing problem :
How can one extend $G_k$ quantum group (Witten--Reshetikhin--Tureav, or WRT in short) invariants of 3-manifolds
away from roots of unity, to the interior of the unit disk $|q|<1$?

Surprisingly, recent physics developments \cite{Gukov:2016gkn,Gukov:2016njj},
brought about by studying M5 branes wrapped on 3-manifolds,
predict that a solution to this problem involves not just one function $Z(q)$,
but rather a collection of functions labeled by elements of the finite set
\be
\label{eqn:flatconn}
\p_0 \, {\cal M}^{\text{ab}}_{\text{flat}}(M_3, \SL(2,\CC)) \quad \cong \quad {\text{Tor}} \, H_1(M_3,\ZZ)/\ZZ_2 \,,
\ee
written here for $G=SU(2)$.
Specifically, it was conjectured in \cite{GPPV} that there exist
new 3-manifold invariants $\widehat{Z}_b (q) \in q^{\Delta_b} \ZZ [[q]]$,
which in practice can be computed for a large class of 3-manifolds,
such that
\begin{equation}
\text{WRT}(M_3,k) \; = \; \sum_a e^{2\pi ik \text{CS} (a)} \left( \lim_{q \, \to \, e^{2\pi i /k}} \sum_{b} S^{(A)}_{ab} \, \widehat{Z}_b (q) \, \right),
\label{WRTviaSZ}
\end{equation}
where the sum runs over the connected components
of the moduli space of flat connections~\eqref{eqn:flatconn}.
Another form of this relation, with a few extra details, will appear below,
in \S\ref{sec:mod},
where the role of the $S$-matrix $S^{(A)}$ will also be clarified.
It has the following explicit form
\be
\label{def:Sabelian}
S^{(A)}_{ab} \; = \; { \sum_{a' \in \{\ZZ_2\text{-orbit of }a\}} \ex(2\lambda(a',b))  \over  \sqrt{|{\text{Tor}}H_1(M_3,\ZZ)|}},
\ee
and only depends on basic topological invariants of the 3-manifold,
such as $H_1(M_3,\ZZ)$ with its inner inner product $\lambda$,
on which the Weyl group $\ZZ_2$ acts by $a\mapsto -a$.

One of our main results in this paper is that $q$-series invariants $\widehat{Z}_a (M_3)$
have a ``hidden structure,''
namely the structure of a projective $\SL(2,\mathbb{Z})$ representation,
distinct from the role(s) modular group played in this context so far \cite{Gukov:2016gkn,GPPV}.
This new structure leads to powerful predictions:

\begin{itemize}

\item
The hidden modular structure helps to determine $\widehat{Z}_a (M_3)$
when $\widehat{Z}_a (-M_3)$ is known (\S\ref{sec:otherside}).
For example, it leads to the following new prediction:
\begin{multline}
\widehat{Z}_1 \left( -M(-2;\frac{1}{2},\frac{1}{3},\frac{1}{2}) \right) \; = \;
2 q^{\frac{5}{12}} - q^{\frac{9}{24}} + q^{\frac{9}{24}} \sum_{n\geq 1}{(-1)^nq^{n} \over (-q;q)_n} = \\
= 2 q^{\frac{5}{12}} - q^{\frac{9}{24}} \left( 1 + q - 2 q^2 + 3 q^3 + \ldots \right)
\label{proposal621}
\end{multline}
which so far was not accessible by any other methods.

\item
It also provides a clear picture of what happens --- at the level of $q$-series $\widehat{Z}_a (q)$
and at the level of the underlying representation theory --- when $q$ approaches a root of unity,
{\it cf.} Figure~\ref{fig:radial}.
In particular, it clarifies when and why one should expect ``corrections'' at the roots
of unity (\S\ref{sec:MTC} and \S\ref{sec:otherside}).

\item
It suggests why and explains in what ways the underlying algebraic structure is
more delicate and interesting in the case of hyperbolic $M_3$ (\S\ref{sec:logCFT}).

\item
Finally, it provides a very simple ``non-topological'' way to determine pretty much everything
one wants to know about flat connections
on a 3-manifold $M_3$ (\S\ref{sec:modularity} and \S\ref{sec:examples}):
the complete taxonomy,
including the type, stabilizer group, values of the Chern-Simons invariant,
transseries coefficients, explicit computations of the Ohtsuki series
and asymptotic expansions around non-trivial flat connections, {\it etc.}

\end{itemize}

\noindent
The text contains various other advances, developed independently of modularity.
For example, computation of $\widehat{Z}_a (M_3)$ for a large class of indefinite
plumbings is developed in \S\ref{sec:indef}.

\bigskip

All values of the Chern-Simons functional in this paper are defined modulo 1.

\subsection{... for number theorists}

To number theorists, the problems discussed in this paper could serve as a ``factory''
that produces infinitely many $q$-series of increasing complexity and
potentially interesting subtle modularity properties (see \S\ref{sec:3dphysics}).
In particular, one rich family of examples which can be handled explicitly is labeled by decorated graphs
(graphs whose vertices are decorated by integer numbers).
Turning it around, the results from number theory find the following important applications in topology and physics:
\begin{itemize}\vspace{-5pt}
\item making predictions on perturbative and non-perturbative three-manifold topological invariants (\S\ref{sec:resurgence}); \vspace{-5pt}
\item shedding light on the resurgence property of half-indices (\S\ref{sec:resurgence});\vspace{-5pt}
\item helping to determine the unknown $\widehat Z_b(q)$ whose computation is not accessible by other methods at present (\S\ref{sec:otherside}).
\end{itemize}

When complexity is moderate, the resulting $q$-series expressions produced by
our physical/topological ``factory''   turn out to be false theta functions, and their relevance to our problems lies in their  quantum modular structure. In this paper we mainly focus on this situation.
By scrutinizing these properties,
we advocate the important role played by the ``false--mock'' pair in the $\widehat Z_b(q)$ story.
Physically and topologically, the crucial requirements for the relevance of such a pair is
\begin{enumerate}\vspace{-6pt}
\item they are related by a $q\leftrightarrow q^{-1}$ transformation in the appropriate sense; \vspace{-4pt}
\item they have the same transseries expression near $q\to 1$ (or $\t\to 0$), in order to be consistent with requirements coming from Ohtsuki series/perturbative Chern--Simons.
\end{enumerate}

Interestingly, in his famous last letter to Hardy in which he introduced the notion of mock theta functions, Ramanujan wrote \cite{MR947735}
\begin{quote}{\em ``I discovered very interesting functions recently which I call ``Mock'' theta functions.
Unlike the ``False'' theta functions  they enter into mathematics as beautifully as ordinary theta functions.''},
\end{quote}
and went on to investigate their behaviour when $q$ approaches roots of unity, a property that is pertinent to the 2nd requirement above.
At the same time, it is precisely these two specific properties of the mock theta functions that he investigated -- the $q$-hypergeometric series expressions (\S\ref{subsec:qhyper}) and the radial limits, that led us to propose that false and mock theta functions in fact form a pair playing a starring role in the problems outlined in Figure \ref{fig:bigpic}.
To connect mock with false,
based on earlier works \cite{CLR,Cheng2011,Rhoades} we demonstrate
that mock modular forms and the corresponding Eichler integral give rise to the same asymptotic series near a cusp (Lemma \ref{lem:asymp}), and show that a Rademacher sum defines a function well-defined  in both the upper and lower half of the plane and equal to the two objects in question respectively (Theorem \ref{defined_bothupdown}).  The relation among different modular objects we discuss in this paper is summarised in Figure \ref{diagram_mockquantumfalse}.

The topological/physical ``factory'' also produces objects with higher complexity. At present we do not have a complete picture of exactly which types of modular behaviour they display. However we believe it should be a fruitful endeavor which could shed new light on the novel modular objects or even lead to the discovery of new natural modular-like structures.

Throughout the paper we use the standard notations:
$$
\HH := \{\tau \in \CC\lvert \im \tau >0\}
$$
for the upper-half plane, and
$$
\HH^- := \{\tau \in \CC\lvert \im \tau < 0\}
$$
for the lower-half plane, as shown in Figure \ref{eich_quantum}.
Cusps refer to the natural boundary $\QQ\cup\{i\infty\}$ of $\HH$ and similarly for $\HH^-$.
By mock modular forms we have in mind the modern definition  which defines them in terms of their non-holomorphic modular corrections (Definition \ref{def:mock}), and by mock theta functions we mean the $q$-series that are mock modular forms with theta function shadows up to the multiplication by some rational power of $q$  \cite{Zagiermock}.

\section{From mock to modular, via 3d $\mathcal{N}=2$ theories}
\label{sec:3dphysics}

Topological phases of matter have been actively studied in recent years,
especially in 2+1 dimensions where many interesting examples have been explored
(quantum Hall effect, topological insulators and superconductors, just to name a few).
A prototypical example of such a phase in 2+1 dimensions is a 3d system with a mass gap
which is nevertheless non-trivial and leads to gapless 2d excitations in the presence of boundaries.

A quantum field theory description of such topological phases often can be phrased
in terms of anomalies, which require 2d degrees of freedom to be present on the boundary
in order to compensate the anomaly of a 3d bulk theory.
In turn, the anomalies as well as the vacuum structure of a 3d gapped phase can be
conveniently described by a topological quantum field theory (TQFT) that encodes
the effects of the topological order and long-range entanglement.

A familiar example is the Chern-Simons gauge theory, which has no physical degrees of freedom
and can arise as a low-energy TQFT in a 2+1 dimensional physical system with a mass gap.
In the presence of a boundary, though, it requires 2d massless degrees of freedom charged under
the gauge group --- the so-called ``edge modes'' --- in order to make the combined 2d-3d system non-anomalous.

\begin{figure}[ht]
\centering
\includegraphics[width=3.0in]{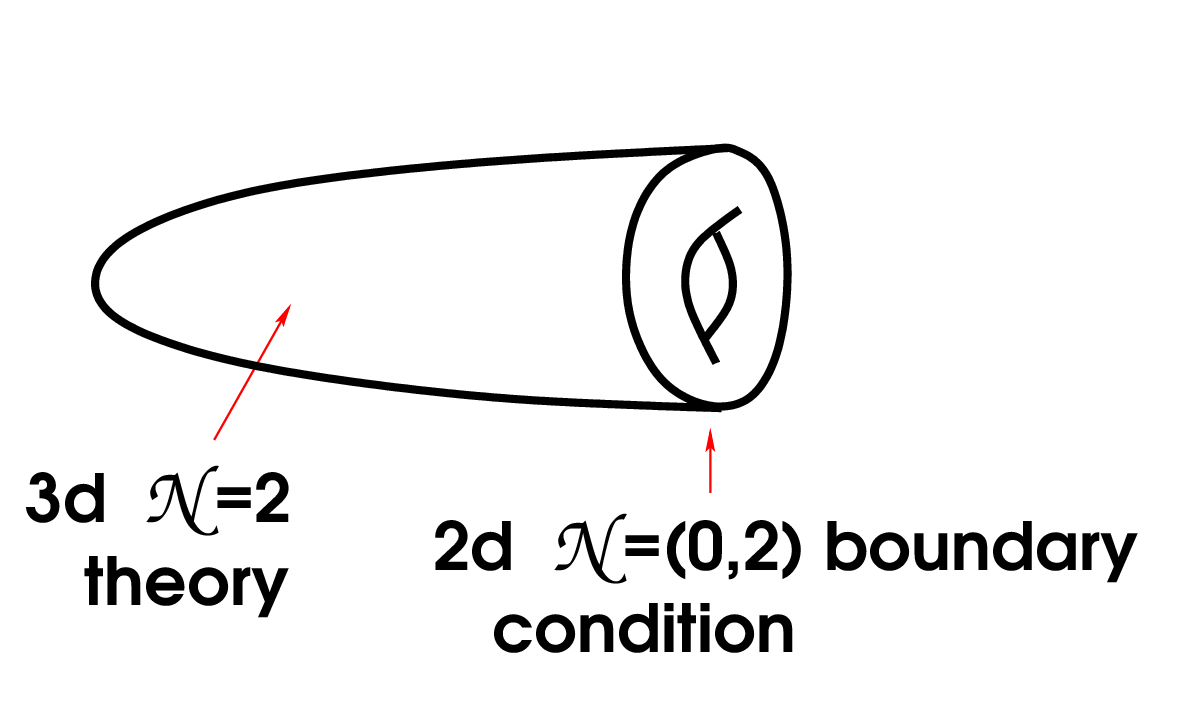}
\caption{A 3d $\mathcal{N}=2$ theory with a 2d $\mathcal{N}=(0,2)$ boundary condition $\mathcal{B}_a$.}
\label{fig:halfindex}
\end{figure}

In the present paper we will be interested in a supersymmetric version of this phenomenon,
where 3d theory with a mass gap has $\mathcal{N}=2$ supersymmetry and 2d ``edge modes'' on the boundary
preserve 2d $\mathcal{N}=(0,2)$ supersymmetry.
The advantage of supersymmetry is that it allows to study the dynamics of such combined 2d-3d system
through the quantities protected by supersymmetry. In the case of 2d $\mathcal{N}=(0,2)$ system,
the elliptic genus is a famous example of such a SUSY-protected quantity and will be our main tool \cite{Witten:1986bf}.
In our problem, illustrated in Figure~\ref{fig:halfindex}, the 2d elliptic genus has a natural
extension \cite{Gadde:2013wq} to the supersymmetric index of the entire 3d theory with
a 2d supersymmetric boundary condition $\mathcal{B}_a$ indexed by a label $a$,
\begin{equation}
\widehat{Z}_a (q) \; = \; Z(D^2\times_q S^1; \mathcal{B}_a) \,.
\label{BlockD2S1}
\end{equation}
A random choice of the 2d $\mathcal{N}=(0,2)$ boundary condition $\mathcal{B}_a$ does not lead to
a $q$-series \eqref{BlockD2S1} with integer powers of $q$ and integer coefficients.
But for a particular choice of boundary conditions --- which correspond to degenerate critical points
of the twisted superpotential $\widetilde{\mathcal{W}}$ when the theory is put on a circle --- the
half-index \eqref{BlockD2S1} does exhibit non-trivial integrality properties:
\be
\widehat{Z}_a (q) \; = \; q^{\Delta_a} \sum_n a_n q^n \,,\qquad\qquad a_n \in \mathbb{Z}
\label{Zan}
\ee
In the context of 3d-3d correspondence, that is for 3d $\mathcal{N}=2$ theories associated with 3-manifolds,
such expressions are sometimes called {\it homological blocks} since the integer coefficients $a_n$ are graded
Euler characteristics of certain homology groups. One of our goals in this paper is to study the modular properties of \eqref{Zan}.

Not only supersymmetry allows to define a protected quantity, it also helps to compute it,
via localization techniques in the regime of weak coupling. This leads to an expression
for the half-index in terms of the contour integral (in the complexified Cartan of the gauge group):
\begin{equation}
\widehat Z_a \; = \; \int \frac{dx}{2\pi i x} \, F_{3d} (x) \, \Theta^{(a)}_{2d} (x)
\label{Zalocalization}
\end{equation}
where the two factors in the integrand, $F_{3d} (x)$ and $\Theta^{(a)}_{2d} (x)$, correspond
to the contributions of 3d theory and 2d boundary degrees of freedom, respectively.

\subsection{The half-index and three-manifolds}

Now let us take a closer look at the definition and the structure of the vortex partition function / half-index \eqref{BlockD2S1},
especially for those boundary conditions $\mathcal{B}_a$ which lead to a power series in $q$ with integer powers and integer coefficients.
Such special boundary conditions were classified in \cite{GPPV} and the corresponding half-index of the combined 2d-3d system
in such cases is known as the homological block (for its relation to homological invariants).

Indeed, when $\widehat{Z}_a (q)$ has a $q$-series expansion \eqref{Zan} we can interpret it as a trace over
the Hilbert space $\mathcal{H}_a$ of the combined 2d-3d system on $\mathbb{R}^2 \cong (\text{cigar})$ (times the ``time circle'')
with boundary condition $\mathcal{B}_a$, as illustrated in Figure~\ref{fig:cigar1}.
Note, the integrality of the coefficients in \eqref{Zan} is crucial for this interpretation.

\begin{figure}[ht]
\centering
\includegraphics[width=3.0in]{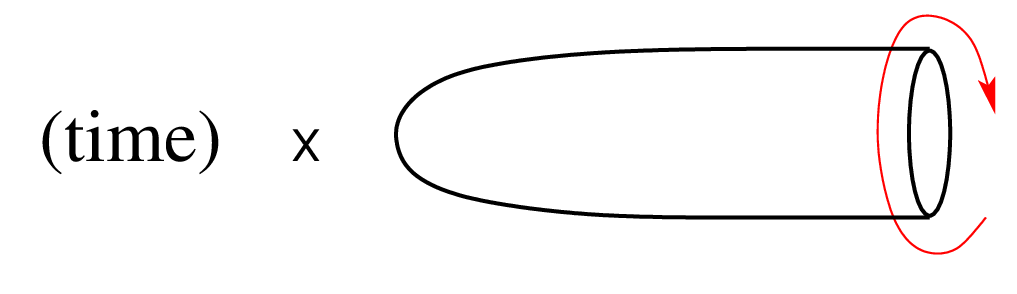}
\caption{A homological block (a.k.a. half-index) counts BPS states of 3d $\mathcal{N}=2$ theory on $(\text{time}) \times (\text{cigar})$.}
\label{fig:cigar1}
\end{figure}

This interpretation of the supersymmetric partition function $\widehat{Z}_a (q)$ is completely analogous
to a similar interpretation of the 3d $\mathcal{N}=2$ superconformal index which, likewise, can be formulated
as a supersymmetric partition function {\it \`a la} \eqref{BlockD2S1} where the 3d space-time $D^2\times_q S^1$
is replaced by $S^1 \times S^2$:
\begin{equation}
\mathcal{I} (q) \; := \; \Tr_{\mathcal{H}_{S^2}}(-1)^F q^{R/2+J_3} \; = \; Z (S^2\times_q S^1)
\end{equation}
In fact, these two supersymmetric indices / partition functions are closely related.
Conjecturally,
\begin{equation}
\mathcal{I} (q) \; = \; \sum_{a}\, |\mathcal{W}_a| \, \widehat{Z}_a (q) \widehat{Z}_a (q^{-1})
\qquad \in\mathbb{Z} [[q]]
\label{IndexFactorization}
\end{equation}
where $|\mathcal{W}_a|$ are certain symmetry factors \cite{GPPV}
and $\widehat{Z}_a (q^{-1})$ is an appropriate extension of $\widehat{Z}_a (q)$ to the region $|q|>1$
(or, equivalently, to $\text{Im} (\tau) < 0$).
Mathematically, the existence of such extension across the border $\text{Im} (\tau) = 0$
is completely non-obvious, but from the physics perspective can be understood as a result of orientation reversal
(parity) transformation on one of the hemispheres $D^2$ that upon gluing produce a 2-sphere $S^2$:
\begin{equation}
\mathcal{I} (q) \; = {\,\raisebox{-.5cm}{\includegraphics[width=6.0cm]{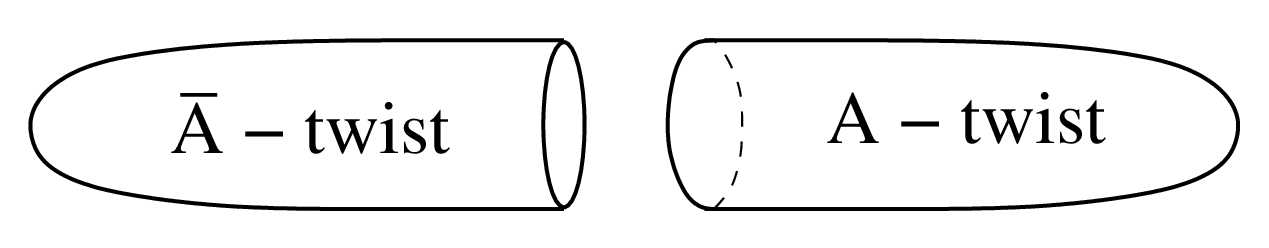}}\,}
\end{equation}

As we shall see in this paper, the question about extending $\widehat{Z}_a (q)$ across the border $\text{Im} (\tau) = 0$
and the search for $\widehat{Z}_a (q^{-1})$ is deeply inter-related to the (quantum) modular properties of the original $q$-series \eqref{Zan}.
The latter, in turn, are determined by the physical properties of the combined 2d-3d system.
There are roughly three qualitatively distinct cases one might consider, which correspond to progressively more
delicate modular properties:
\begin{itemize}

\item
3d ``bulk'' theory is completely gapped and its contribution to the half-index \eqref{BlockD2S1} is trivial, $F_{3d} (x) = 1$.
In this case, \eqref{Zalocalization} basically computes the elliptic genus of the 2d $\mathcal{N}=(0,2)$ boundary theory $\mathcal{B}_a$
and has the standard modular properties of a 2d elliptic genus. In particular, it involves the ordinary modular forms familiar from textbooks.
Examples of this type abound; any (non-relative) 2d $\N=(0,2)$ or $\N=(2,2)$ theory, together with the trivial 3d theory, is an example.

\item
The next case is when 3d $\mathcal{N}=2$ theory is gapped but nevertheless is in a non-trivial topological phase,
as described in the introduction. In this case, the ``dominant'' contribution to the half-index \eqref{BlockD2S1}
still comes from 2d massless degrees of freedom, but the nice modular behavior of the 2d elliptic genus is ``spoiled''
by the non-trivial contribution $F_{3d} (x) \ne 1$ of the 3d $\mathcal{N}=2$ theory.
This case of intermediate complexity in its modular properties is the main subject of the present paper;
in this case, the relevant modular objects are false theta functions and mock modular forms, as well as their close generalisations.
All examples in this paper apart from those presented in \S\ref{sec:4sing} are of this type.
In particular, in \S\ref{sec:examples} and \S\ref{sec:optimal} we present many explicit examples of half-indices
$\widehat{Z}_a$ for 3d $\N=2$ theories $T[M_3]$ that correspond to simple 3-manifolds.

\item
Finally, the most general case that one can consider is when both 2d boundary condition $\mathcal{B}_a$
and 3d $\mathcal{N}=2$ theory have massless degrees of freedom ({\it i.e.} no mass gap). In this case,
the standard modular properties of the 2d elliptic genus of $\mathcal{B}_a$ are considerably distorted by
non-modular behavior of the 3d bulk theory.
Although we expect the objects to be significantly more complicated in this case,
presumably they still exhibit the structure of quantum modular forms.
Clearly, the two previous cases are special instances of this more general behavior.
A typical example of this behavior is 3d $\N=2$ theory $T[S^1 \times \Sigma_g]$,
a.k.a. 3d $\N=2$ adjoint SQCD whose half-index is given by \eqref{Zalocalization}
with the integrand \eqref{Sungnitpick} for general $g>1$.
A slight modification gives a 3d $\N=2$ theory
\be
\begin{array}{c@{\;}|@{\;}c@{\;}|@{\;}c@{\;}|@{\;}c}
& \; ~SU(2)_{\text{gauge}} \; & \; ~\text{R-charge} \; & \; ~\text{boundary condition} \; \\\hline
\text{chiral} & \text{adj} & 2 & \text{Neumann} \\
N_f~\text{chirals} & \Box & 0 & \text{Neumann} \\
N_f~\text{chirals} & \Box & 0 & \text{Dirichlet}
\end{array}
\ee
whose half-index is almost identical, {\it cf.} \cite{GPPV},
\be
\widehat{Z} (q) \; =\; \frac{1}{2 (q;q)_{\infty}} \int \frac{dz}{2\pi i z}
\frac{(1-z^{2}) (1-z^{-2})}{(1-z)^{N_f} (1-z^{-1})^{N_f}} \sum_{n \in \ZZ} q^{n^2} z^{2n}
\label{ZSQCD}
\ee
but which does not arise, to the best of our knowledge, from any 3-manifold via 3d-3d correspondence.\footnote{For various other examples and applications of half-indices see {\it e.g.}~\cite{Gadde:2013wq,Gadde:2013sca,Gukov:2016gkn,GPPV,Dimofte:2017tpi,Costello:2018fnz,Feigin:2018bkf,Dedushenko:2018aox,Dedushenko:2018tgx,Jockers:2018sfl}.}

\end{itemize}

\noindent
This classification, of course, is only qualitative; its main purpose is to provide
an intuitive explanation of the deviation from traditional types of modularity.
In particular, the borderlines between different types of behavior are not sharp
and some examples may fall right in the middle, or one might find sub-classes in each type of behavior.

Although our considerations apply to arbitrary 3d $\mathcal{N}=2$ theories, a particularly large
class of examples comes from 3-manifolds via the so-called 3d-3d correspondence or, equivalently,
compactification of 6d $(0,2)$ fivebrane theory on a 3-manifold $M_3$.
The resulting 3d $\mathcal{N}=2$ theory, usually denoted $T[M_3]$, can therefore be a proxy for
a more general 3d $\mathcal{N}=2$ theory.

For 3d $\mathcal{N}=2$ theories $T[M_3]$, the BPS Hilbert space $\mathcal{H}_a[M_3]$
is a homological invariant of 3-manifolds, and
\begin{equation}
a \; \in \; \pi_0 \, \mathcal{M}_\text{flat}^{\text{ab}} (M_3, G_{\mathbb{C}})
\label{aHW}
\end{equation}
labels the connected components of the moduli space of {\it abelian} flat connections on $M_3$ ({\it c.f.}, \eqref{eqn:flatconn}).
Here, the requirement for the boundary condition $\mathcal{B}_a$ to represent abelian flat connections
is intimately related to the integrality of the resulting $q$-series \eqref{Zan}. As explained in \cite{Gukov:2016njj},
the information about non-abelian flat connections is not lost, but repackaged in the $q$-series $\widehat{Z}_a (q)$
and in its categorification $\mathcal{H}_a[M_3]$.

Aside from its applications in low-dimensional topology, the advantage of working with this class of
3d $\mathcal{N}=2$ theories $T[M_3]$ is that the homological blocks \eqref{BlockD2S1} can be explicitly
computed for many non-trivial examples.
The answer is often expressed as a contour integral \eqref{Zalocalization}, where (up to an overall power of $q$):
\begin{equation}
F_{3d} (x) \; = \; (x-x^{-1})^{2-2g} \,, \qquad\qquad \text{for degree-}p~S^1~\text{fibration over}~\Sigma_g
\label{Sungnitpick}
\end{equation}
\begin{equation}
F_{3d} (x) \; = \;
\prod_{v\;\in\; \text{Vertices} (\Gamma)} \left({x_v-1/x_v}\right)^{2-\text{deg}(v)} \,, \qquad\qquad \text{for plumbing } \Gamma
\label{Fxexamples}
\end{equation}
Note, that $F_{3d} (x)$ does not depend on the choice of 2d $\mathcal{N}=(0,2)$ boundary condition $\mathcal{B}_a$;
this dependence comes through the factor $\Theta^{(a)}_{2d} (x)$ which is basically the elliptic genus of $\mathcal{B}_a$.
For instance, in the above examples \eqref{Fxexamples}:
\begin{equation}\label{eqn:gentheta}
\Theta^{(a)}_{2d} (x) \; = \;
\sum_{\ell \in 2M\mathbb{Z}^L+a}q^{-\frac{(\ell,M^{-1}\ell)}{4}} \prod_{v\in \text{Vertices($\Gamma$)}} x_v^{\ell_v}
\end{equation}
is the theta function for the lattice determined by the linking form $M$ of $\Gamma$,
such that $H_1 (M_3,\mathbb{Z}) \cong \mathbb{Z}^L/M\mathbb{Z}^L$ \cite{GS}.
For a degree-$p$ $S^1$ fibration over $\Sigma_g$, we simply have
$\Theta^{(a)}_{2d} (x) = \sum_{n\in p\mathbb{Z}+a}q^{n^2/p}\,x^{2n}$.

\section{Three encounters of modularity}
\label{sec:mod}

In our journey we encounter three different $S$-matrices and three
corresponding ``$\SL(2,\ZZ)$ representations'':

\begin{itemize}

\item
One set of modular $S$ and $T$ matrices encodes the information about all {\it twisted indices}
of 3d $\N=2$ theories on \cite{Gukov:2016gkn}. In fact, this numerical data is a part of much
richer structure, namely the modular tensor category whose Grothendieck group is
the space of supersymmetric states of a 3d $\N=2$ theory on $\RR \times T^2$.
When combined with 3d-3d correspondence, it associates a modular tensor category (MTC for short), $\text{MTC} [M_3]$, to every closed 3-manifold $M_3$.

\item
A different, much simpler $S$-matrix $S^{(A)}$ already appeared in \eqref{def:Sabelian}.
It is trying to make an ``$\SL(2,\ZZ)$ representation'' out of
the set \eqref{eqn:flatconn} of {\it abelian} flat connections on $M_3$ and,
in the basic case $H_1(M_3,\ZZ) = \ZZ_p$, takes a simple form
\be
S_{ab} \; = \; \frac{\cos 2\pi \frac{ab}{p}}{1 + \delta_{a,0}}
\ee
This peculiar ``$\cos$" representation of $\SL(2,\ZZ)$ is suggestive of
a non-semisimple MTC common in logarithmic conformal field theory.
It appears to be related to another connection with logarithmic CFTs
which enters our story again in \S\ref{sec:logCFT}.

\item
The last but not least --- in fact, the most important to us here --- is the projective $\SL(2,\ZZ)$ representation
that describes modular properties of the $q$-series $\widehat Z_a (q)$.

\end{itemize}

\noindent
The main goal of this section is to describe each of these (close cousins of) $\SL(2,\ZZ)$ representations, and we devote each one a subsection.

\subsection{Twisted indices of 3d $\N=2$ theories}

Three-dimensional $\N=2$ theories, with or without a Lagrangian description, do not have sufficient supersymmetry
to admit a full topological twist on an arbitrary 3-manifold.
However, when $U(1)_R$ R-symmetry is unbroken, they can be twisted on $S^1 \times \Sigma_g$ or,
more generally, on a degree-$p$ circle bundle over a genus-$g$ surface $\Sigma_g$.
Such partition functions are sometimes called {\it twisted indices} of 3d $\N=2$ theories
and, for general $g$ and $p$, their entire structure is captured by a modular tensor category (MTC)
that can be assigned to a 3d $\N=2$ theory \cite{Gukov:2016gkn}.

Among other things, this rich structure involves modular $S$ and $T$ matrices,
whose values $S_{0 \alpha}$ and $T_{\alpha \alpha}$
allow to write a general formula for twisted indices in a succinct form:
\be
Z_{\text{twisted}} \; = \; \sum_{\alpha} (S_{0 \alpha})^{2-2g} (T_{\alpha \alpha})^p.
\ee
When 3d $\N=2$ theory in question admits a Lagrangian description,
the sum over $\alpha$ can be interpreted as a sum over solutions
to the Bethe ansatz equations using the standard localization technique,
whereas $S_{0 \alpha}$ and $T_{\alpha \alpha}$
can be identified with what sometimes are called
handle-gluing and twist/fibering operators:
\be
\begin{aligned}
S_{0 \alpha} &= \text{``handle-gluing operator''}
 \\
T_{\alpha \alpha} &= \text{``twist/fibering operator''}.
\end{aligned}
\ee

In the context of 3d-3d correspondence, {\it i.e.} for 3d $\N=2$ theories $T[M_3]$,
this modular tensor category is effectively assigned to a 3-manifold $M_3$
(plus a choice of the root system) and was dubbed $\text{MTC} [M_3]$ in \cite{Gukov:2016gkn}.
Correspondingly, the $S$ and $T$ matrices then admit interpretation in terms of
the topological data of the 3-manifold $M_3$. 
For instance, 
\be
T_{\alpha \beta} \; = \; \delta_{\alpha \beta} \, e^{2 \pi i \text{CS} (\alpha)}
\label{TviaCS}
\ee
where $\text{CS} (\alpha)$ is the Chern-Simons invariants of
a flat connection $\alpha : \pi_1 (M_3) \to G_{\CC}$,
defined modulo 1.
Similarly, $S_{0 \alpha}$ is related to the Reidemeister torsion of $M_3$
twisted by $\alpha$; this relation easily follows  from {\it e.g.} \cite{Gukov:2016njj},
where it also appears as the constant ($\hbar$-independent)
coefficient of the transseries $Z_{\text{pert}}^{(\alpha)}$.

If $\text{MTC} [M_3]$ is a representation category of some conformal field
theory (or, equivalently, vertex algebra) with diagonalizable $T$-matrix, we can use the standard relation
in conformal field theory, $T_{\alpha \alpha} = e^{2 \pi i (\Delta_{\alpha} - \frac{c}{24})}$,
to write \eqref{TviaCS} in terms of the conformal dimensions $\Delta_{\alpha}$:
\be
\text{CS} (\alpha) \; = \; \Delta_{\alpha} - \frac{c}{24}
\label{CSconfdim}
\ee
This relation plays an important role in various gluing formulae
of 4-manifold invariants \cite{Feigin:2018bkf}.

Note, the $S$ and $T$ matrices of $\text{MTC} [M_3]$ described here
have elements, $S_{\alpha \beta}$ and $T_{\alpha \beta}$, labeled by
$\alpha$ and $\beta$ which run over {\it all} flat connections on $M_3$,
abelian and non-abelian, reducible and irreducible:
\be
\alpha, \beta \; \in \; \pi_0 \, \M_{\text{flat}} (G_{\CC}, M_3)
\ee
This is in stark contrast with ``modular'' matrices that enter the relation \eqref{WRTviaSZ}
between $\widehat Z_a$ and Witten-Reshetikhin-Turaev invariants of $M_3$:
\be
\label{relation_CS_blocks}
\text{WRT}(M_3,k) = \frac{q^{\D}}{2^{c} \, (\sqrt{k})^{1-b_1(M_3)}}  \sum_{a,b} \ex(k\lambda(a,a)) S^{(A)}_{ab} \widehat{Z}_b(q) \lvert_{q\to \ex({1\over k})}
\ee
The peculiar $S$-matrix $S^{(A)}$ that appears here will be the subject
of the next subsection; $c$ and $\D$ are certain rational numbers, and
the sum runs over connected components of the moduli space of flat connections \eqref{eqn:flatconn}
equipped with a bilinear form $\lambda$ given by
the linking pairing on the torsion part of $H_1(M_3,\ZZ)$.

\begin{figure}[ht]
\centering
\includegraphics[width=2.6in]{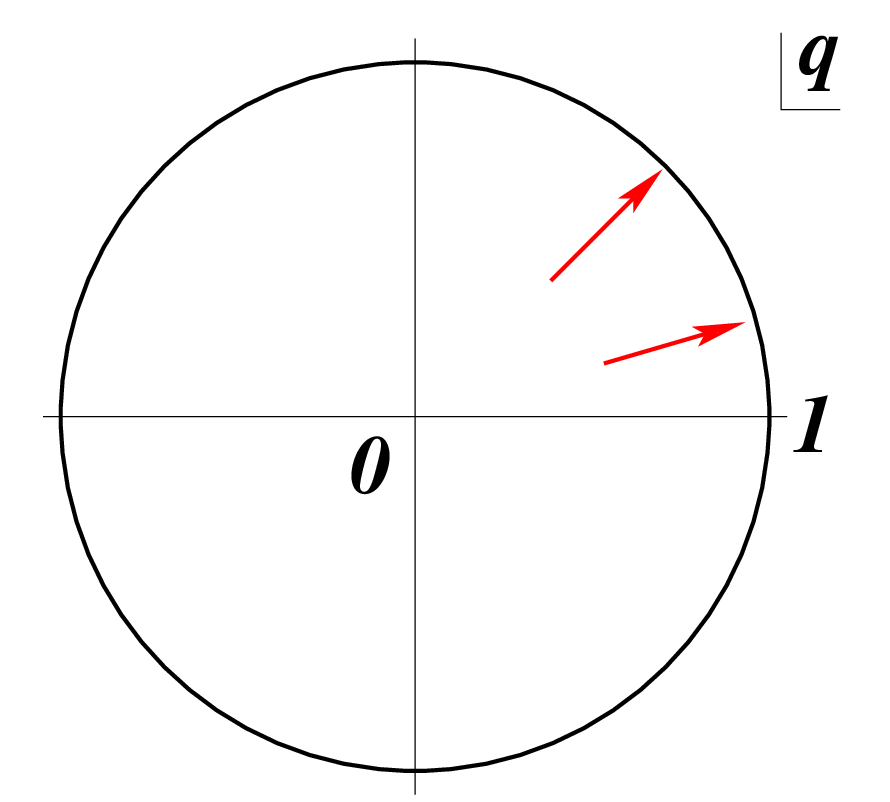}
\caption{The limit $q \to e^{2\pi i /k}$, with $k \in \ZZ$,
enters many aspects of our story:
the Kazhdan-Lusztig correspondence,
the relation between $\widehat{Z}_a (M_3)$ and WRT invariants,
the relation between mock modular forms and false thetas, {\it etc.}}
\label{fig:radial}
\end{figure}

\subsection{$\widehat{Z}_a$ and non-semisimple MTCs}
\label{sec:MTC}

Many examples of modular tensor categories arise as representation
categories of vertex operator algebras (VOAs).
In particular, rational VOAs give rise to semisimple representation categories,
whereas more esoteric logarithmic VOAs lead to non-semisimple MTCs,
in the sense of Lyubashenko~\cite{Lyubashenko}.

As the name suggests, a key feature of logarithmic CFTs (equivalently, the corresponding VOAs)
is that some correlation functions exhibit logarithmic behavior.
This happens when the Hamiltonian $L_0$ is non-diagonalizable (has non-trivial Jordan blocks)
and, therefore, necessarily requires representations which are {\it reducible},
but not {\it decomposable}.\footnote{{\it Indecomposable} means that
a representation can not be written as a direct sum of other non-trivial representations.
A good example to keep in mind is that of a finite-dimensional non-semisimple
algebra $\mathcal{A}$ with finitely many irreducible (simple) modules $M_i$:
\be
\mathcal{A} \; = \; \bigoplus_{i=1}^n (\dim M_i) \mathcal{P}_i
\label{AdecompviaMP}
\ee
where $\mathcal{P}_i$ denotes the indecomposable projective cover of $M_i$,
such that $\dim \text{Hom} (\mathcal{P}_i, M_j) = \delta_{ij}$.}
The converse also appears to be true, and another key feature of logarithmic CFTs is the presence of
irreducible representations which have non-trivial extensions among themselves. 

Perhaps the simplest and most well-known examples of logarithmic CFTs
with such properties are the so-called $(1,p)$ triplet models.
They have central charge\footnote{In the special case $p=2$, we have
$$
\mathcal{W}_2 \; \cong \; SF_1^+
$$
where $SF^+_d$ denotes the even part of the symplectic fermions $SF_d$,
another popular family of logarithmic vertex superalgebras,
with the central charge $c = -2d$.}
\be
c \; = \; 1 - 6 \frac{(1-p)^2}{p}
\; = \; 1 - 3 \alpha_0^2
\label{1pccharge}
\ee
where we use the standard CFT notations
\be
\alpha_0 = \alpha_+ + \alpha_-
\,, \qquad \qquad
\alpha_+ = \sqrt{2p}
\,, \qquad \qquad
\alpha_- = - \sqrt{\frac{2}{p}}
\label{aaa}
\ee
The name ``triplet'' comes from the fact that the corresponding vertex algebra,
usually denoted either $\mathcal{W}_p$ or $\mathcal{W} (2,(2p-1)^{\otimes 3})$,
is an extension of the Virasoro algebra by the $\frak{sl} (2)$ triplet of
the Virasoro primary fields $W^{\pm , 0} (z)$
of conformal dimension $2p-1$ \cite{Fuchs:2003yu}:
$$
W^- (z) \; = \; e^{- \alpha_+ \varphi} (z)
\,, \qquad
W^0 (z) \; = \; [S_+ , W^- (z)]
\,, \qquad
W^+ (z) \; = \; [S_+ , W^0 (z)]
$$
Here, $S_+$ is the ``long'' screening operator \eqref{ScreeningS}
that will be useful to us later.

The triplet algebra $\mathcal{W}_p$
has $2p$ irreducible representations $\mathcal{X}^{\pm}_s$, $s=1, \ldots, p$,
with conformal dimensions
\be
\Delta (\mathcal{X}^+_s) \; = \; \frac{(p-s)^2}{4p} + \frac{c-1}{24}
\ee
\be
\Delta (\mathcal{X}^-_s) \; = \; \frac{(2p-s)^2}{4p} + \frac{c-1}{24}
\ee
Unlike the familiar case of a rational CFT,
the characters of the irreducible representations $\mathcal{X}^{\pm}_s$,
\begin{multline}
\chi_{s}^+ (q) \; := \; \text{Tr}_{\mathcal{X}_{s}^+} q^{L_0 - \frac{c}{24}}
\; = \; \frac{q^{-1/24}}{\prod_{n=1}^{\infty} (1-q^n)}
\sum_{n \in \ZZ} (2n+1) q^{p (n + \frac{p-s}{2p})^2} \; = \;  \\
\; = \; \frac{1}{\eta (q)} \left( \frac{s}{p} \theta_{p-s} (q) + 2 \theta_{p-s}' (q) \right)
\end{multline}
\begin{multline}
\chi_{s}^- (q) \; := \; \text{Tr}_{\mathcal{X}_{s}^-} q^{L_0 - \frac{c}{24}}
\; = \; \frac{q^{-1/24}}{\prod_{n=1}^{\infty} (1-q^n)}
\sum_{n \in \ZZ} 2n q^{p (-n + \frac{s}{2p})^2} \; = \;  \\
\; = \; \frac{1}{\eta (q)} \left( \frac{s}{p} \theta_{s} (q) - 2 \theta_{s}' (q) \right)
\end{multline}
do not close under the action of the modular group $\SL(2,\ZZ)$.
This is a general feature of logarithmic CFTs.
Indeed, just like correlation functions, (modular transformations of) characters
in logarithmic theories involve logarithms and naively take values in $\ZZ [[q]] [\log q]$.
Then, formal manipulations that re-express $\log q$ terms
as power series in $q$ often lead to expressions
which are not modular in the traditional sense ({\it e.g.} they can be mock modular)
and also contain both positive and negative coefficients in the $q$-expansion.
This formal way of rewriting $\log q$ terms via $q$-series is precisely
what one encounters in the analytic continuation of WRT invariants
away from roots of unity~\cite{Pei:2015jsa,GPPV}.
This parallel between $\widehat{Z}_a (q)$ and (pseudo-)characters of log-CFTs
will be developed further in \S\ref{sec:logCFT}.

The modular properties of the characters can be restored by augmenting
them with a set of ``extended'' characters (or, ``pseudo-characters'').
In the case of the logarithmic $(1,p)$ triplet model, this means that,
in addition to the $2p$ characters $\chi_{s}^{\pm} (q)$,
one needs to introduce $p-1$ pseudo-characters,
which then altogether form a $(3p-1)$-dimensional
projective representation ${\frak Z}$ of $\SL(2,\ZZ)$.
This representation can be identified with the endomorphisms
of the identity functor in the category of VOA modules
and has the structure~\cite{Feigin:2005zx}:
\be
\frak{Z} \; = \;
\mathcal{R}_{p+1} \; \oplus \; \CC^2 \otimes \mathcal{R}_{p-1}
\label{QGcenter}
\ee
where $\mathcal{R}_{p-1}$ is the $(p-1)$-dimensional
``$\sin \frac{\pi rs}{p}$" representation of $\SL(2,\ZZ)$
on the unitary $\widehat{\frak{sl}(2)}_{p-2}$ characters,
and $\CC^2$ is the defining two-dimensional representation of $\SL(2,\ZZ)$.
Of most interest to us here is a non-unitary $(p+1)$-dimensional
``$\cos \frac{\pi rs}{p}$" representation $\mathcal{R}_{p+1}$
of $\SL(2,\ZZ)$ that does not come from any familiar rational CFT.
In particular, it has a non-diagonalizable $T$-matrix.

Much like $\CC [\mathcal{M} (G_{\CC}, M_3)]$ is isomorphic (as a set)
to the Grothendieck ring of $\text{MTC} [M_3]$ described in the previous subsection,
$\frak{Z}$ in \eqref{QGcenter} is related to the Grothendieck ring of
a non-semisimple MTC.\footnote{Here we find a connection
to non-semisimple MTCs based on non-perturbative arguments
and modular properties of the partition functions. These arguments are consistent
with braiding properties of Wilson lines in complex Chern-Simons theory and quantization
of the moduli space of flat $G_{\CC}$-connections~\cite{DEGVunpublished}.}
Its structure is most easily understood via the Kazhdan-Lusztig correspondence
which we describe next.
In particular, the Kazhdan-Lusztig correspondence helps to see the structure
of indecomposable modules which, as advertised earlier, are responsible for
the logarithmic nature of the CFT,\footnote{They are also responsible for
the additional mysterious pseudo-characters, which can be viewed as
modified traces $\tilde \chi_V (q) = \text{Tr}_{V} \, g \, q^{L_0 - \frac{c}{24}}$,
twisted by $g \in \text{End} (V)$.}
and which can be constructed as (iterative) extensions of irreducible (simple) modules.
In particular, in the end of this process one finds projective modules
with the following structure ($1 \le s \le p-1$):
\be
\label{PmodXXXX}
\xymatrix{
\\
\mathcal{P}_s^{\pm} : \\
}
\qquad
\xymatrixcolsep{4pc}
\xymatrix{
& \mathcal{X}_s^{\pm} \ar[ld] \ar[rd] & \\
\mathcal{X}_{p-s}^{\mp} \ar[rd] & & \mathcal{X}_{p-s}^{\mp} \ar[ld] \\
& \mathcal{X}_s^{\pm} & }
\ee
where, following \cite{Feigin:2005zx,Feigin:2005xs}, we denote extensions by
\be
\overset{\displaystyle{\mathcal{X}^{\pm}_s}}{\bullet}
\xrightarrow[~~~~]{~~~~}
\overset{\displaystyle{\mathcal{X}^{\mp}_{p-s}}}{\bullet}
\ee
so that arrow always points from the irreducible subquotient to the irreducible submodule.
A reflection of the diamond diagram \eqref{PmodXXXX} is a simple example
of an endomorphism of the projective module $\mathcal{P}_s^{\pm}$.

\subsubsection{Kazhdan-Lusztig correspondence}

It is a relatively well known and widely used fact that fusion rules of a WZW model
are related to representation theory of a quantum group at a primitive root of unity.
Much less appreciated, however, is the key aspect of this relation which involves {\it semisimplification}.
Namely, the semisimple MTC which describes the semisimple fusion in rational CFT is only a quotient of
the representation category of a quantum group by the ideal of indecomposable tilting modules.

Curiously, this correspondence --- called the Kazhdan-Lusztig correspondence \cite{KL12,KL3,KL4} --- between
fusion algebra of a CFT and the Grothendieck ring of the corresponding quantum group is actually more direct in the case of logarithmic CFTs.
While surprising at first, there is a simple reason for it:
the MTC associated to a logarithmic VOA is not
semisimple and, therefore, the corresponding category on the quantum-group side requires no
semisimplification.

The semisimplification is only necessary if we wish to make an additional step
and relate quantum groups at roots of unity (or logarithmic CFTs) to rational WZW models.
Its implication for 3-manifolds is that $q$-series invariants $\widehat{Z}_a (M_3)$
--- which, as we explain below, are naturally related to (characters of) logarithmic CFTs ---
may require certain corrections at roots of unity, when comparing to WRT invariants of $M_3$,
{\it cf.} Table~\ref{tab:dict}.

\begin{table}[htb]
\centering
\renewcommand{\arraystretch}{1.3}
\begin{tabular}{|@{\quad}c@{\quad}|@{\quad}c@{\quad}| }
\hline  {\bf 3-manifolds} & ~~~~{\bf Logarithmic CFTs}~~
\\
\hline
\hline flat connections & modules \\
\hline invariants $\widehat{Z}_a (q)$ & characters $\chi (q)$ \\
\hline ``mock side'' & KL ``positive zone'' \\
\hline ``false side'' & KL ``negative zone'' \\
\hline ``corrections'' at roots of unity & semisimplification \\
\hline
\end{tabular}
\caption{Mysterious duality between 3-manifolds and logarithmic CFTs.}
\label{tab:dict}
\end{table}

The restricted (a.k.a. ``baby'') quantum group $\overline{{\mathcal U}}_q (\frak{sl}_2)$
at the primitive $2p$-th root of unity $q = e^{\frac{i \pi}{p}}$ is defined by supplementing the usual relations\footnote{See \cite{Chun:2015gda} for a friedly introduction and a physical realization
of the Lusztig quantum groups in the setup of \cite{Gukov:2016gkn,GPPV}
that leads to $\widehat{Z}_a (M_3)$.
In a two-dimensional description of this setup, $E$ and $F$ generators of the quantum group
correspond to half-BPS interfaces of a 2d $\N=(2,2)$ CFT (Kazama-Suzuki model),
so that the quantum group emerges as an algebra of interfaces.}
\be
[E,F] \; = \; \frac{K - K^{-1}}{q - q^{-1}}
\,, \qquad
KEK^{-1} \; = \; q^2 E
\,, \qquad
KFK^{-1} \; = \; q^{-2} F
\ee
with
\be
E^p \; = \; 0 \; = \; F^p
\qquad , \qquad K^{2p} \; = \; {\bf 1}
\ee
The resulting quotient of the (perhaps) more familiar quantum group ${\mathcal U}_q (\frak{sl}_2)$ is,
in fact, finite-dimensional, namely $2p^3$-dimensional.\footnote{Its regular
representation has the structure \eqref{AdecompviaMP}:
$$
\text{Reg} \; = \;
\bigoplus_{s=1}^{p-1} s \mathcal{P}_s^+
\oplus \bigoplus_{s=1}^{p-1} s \mathcal{P}_s^-
\oplus p \mathcal{X}_p^+
\oplus p \mathcal{X}_p^-
$$
Dimensions of various pieces, then, add up as follows:
$2p \cdot \sum_{s=1}^{p-1} s + 2p \cdot \sum_{s=1}^{p-1} s + p \cdot p + p \cdot p = 2p^3$,
where we used \eqref{xmoddim}, \eqref{projmoddim} and \eqref{Steinberg}.}
It has $2p$ irreducible representations $\mathcal{X}^{\pm}_s$, $s=1, \ldots, p$,
with the highest weight $\pm q^{s-1}$:
\be
\dim \mathcal{X}_s^{\pm} \; = \; s \;,
\qquad \qquad
\text{h.w.} \left( \mathcal{X}_s^{\pm} \right) \; = \; \pm q^{s-1}
\label{xmoddim}
\ee
and a $(3p-1)$-dimensional center,
which carries a projective $\SL(2,\ZZ)$ representation~\cite{Feigin:2005zx,Feigin:2005xs}:
\begin{equation}
\dim {\frak Z} \; = \; 3p-1
\end{equation}
Under the Kazhdan-Lusztig correspondence,
$\mathcal{X}_s^{\pm}$ and ${\frak Z}$ are identified, respectively,
with the irreducible representations and the space \eqref{QGcenter}
of pseudo-characters of the triplet algebra $\mathcal{W}_p$,
denoted by the same letters.

According to the Kazhdan-Lusztig correspondence,
not only the projective $\SL(2,\ZZ)$ representations are supposed to match,
but the entire representation categories
of $\mathcal{W}_p$ and $\overline{{\mathcal U}}_q (\frak{sl}_2)$
should be equivalent as braided tensor categories.
In particular, apart from $2p$ irreducible modules $\mathcal{X}^{\pm}_s$
there are also $2p$ Verma modules $\mathcal{V}^{\pm}_s$, $1 \le s \le p$,
and $2p$ projective modules $\mathcal{P}^{\pm}_s$, $1 \le s \le p$,
of dimension
\be
\dim \mathcal{P}^{\pm}_s \; = \; 2p
\,,\qquad \qquad
\text{qdim} \mathcal{P}_s^{\pm} \; = \; 0
\qquad \qquad (1 \le s \le p-1)
\label{projmoddim}
\ee
For generic $s \ne p$, they are given by extensions
\be
0 \to \mathcal{X}^{\mp}_{p-s} \to \mathcal{V}^{\pm}_s \to \mathcal{X}^{\pm}_s \to 0
\ee
and
\be
0 \to \mathcal{V}^{\mp}_{p-s} \to \mathcal{P}^{\pm}_s \to \mathcal{V}^{\pm}_s \to 0
\ee
respectively.
This is precisely the structure depicted in \eqref{PmodXXXX}.
In the special case $s=p$, the two modules
\be
\mathcal{X}^{\pm}_p \; = \;
\mathcal{V}^{\pm}_p \; = \;
\mathcal{P}^{\pm}_p
\label{Steinberg}
\ee
are irreducible, Verma, and projective simultaneously.
They are called {\it Steinberg modules} by analogy with what happens
in the quantum group over $\mathbb{F}_p$.

As we already mentioned earlier, another statement of the Kazhdan-Lusztig correspondence is that fusion algebra
of the logarithmic CFT is supposed to match the Grothendieck ring of $\overline{{\mathcal U}}_q (\frak{sl}_2)$.
The latter is generated over $\mathbb{Z}$ by $x=\mathcal{X}^{+}_2$ \cite{Feigin:2005zx}:
\begin{equation}
\text{Gr} \; = \; \left. \mathbb{Z} [x] \right/ (x-2)(x+2) \prod_{j=1}^{p-1} \big( x - 2 \cos \tfrac{\pi j}{p} \big)^2
\end{equation}
Note, in the Grothendieck ring there is no difference between direct sums and non-trivial extensions,
so that $[\mathcal{P}_s^{\pm}] = 2 [\mathcal{X}_s^{\pm}] + 2 [\mathcal{X}_{p-s}^{\mp}]$, {\it etc.}

\begin{table}[h!]
\centering
\renewcommand{\arraystretch}{1.3}
\begin{tabular}{|@{\quad}c@{\quad}|@{\quad}c@{\quad}| }
\hline  {\bf 3d topology and 3d BPS states}  & {\bf Modularity}
\\
\hline
\hline {non--abelian $\SL(2,\CC)$ flat connections} & $S$-matrix condition \eq{rule_nonAb} \\
\hline {complex flat connections} & $S$-matrix condition \eq{rule_complex} \\
\hline {pole contributions in the Borel} & Weil representation \\
{resummation of $\widehat{Z}_a(M_3)$} & $S$-matrix ${\cal S}^{(B)}$ \\
\hline {Chern--Simons invariants} & $T$-matrix ${\cal T}^{(B)}$ \eq{rule_nonAb}\\
\hline homological blocks & false theta functions  \\
$\widehat{Z}_a(M_3)$ and $\widehat{Z}_a(- M_3)$ & ~~and mock theta functions ~ \\
\hline
\end{tabular}
\caption{The correspondence between modularity and topology/BPS states.}
\label{dic_table}
\end{table}

\subsection{The Weil representations}
\label{sec:WeilRep}

As mentioned in \S\ref{sec:intro}, via Chern--Simons theory a representation for (the metaplectic double cover of) $\SL(2,\ZZ)$ is attached to the 3-manifold $M_3$, and this representation plays an important role in the categorification of 3-manifold invariants. 
Its S-matrix $S^{(B)}$ captures the perturbative as well as non-perturbative data of the homological blocks $\widehat{Z}_a (M_3)$.
By relating the homological blocks to the Chern--Simons partition functions (or WRT invariants), we see that its $T$- and  $S$-matrices give sharp predictions for the data of non--abelian $\SL(2,\CC)$ flat connections, including their numbers, Chern--Simons invariants, and whether they are real or complex flat connections\footnote{We say that an $\SL(2,\CC)$ flat connection is real if it is conjugate to an $SU(2)$ flat connection and complex otherwise.}.
These relations are summarised in Table~\ref{dic_table}.
In this subsection we give explicit details of these representations. 

In the main classes of examples, which are various Seifert manifolds with three or four singular fibers, the relevant  representations are (based on) the so-called Weil representations. Given a positive-definite lattice, one can associate a Weil representation, which is a representation for  $\SL(2,\ZZ)$ when the rank is even and a representation for the metaplectic double cover $\widetilde {\SL(2,\ZZ)}$ of $\SL(2,\ZZ)$ when the rank is odd. Recall that $\widetilde {\SL(2,\ZZ)}$ consists of  elements which are the pairs $(\gamma,\upsilon)$, where $\gamma=\left(\begin{smallmatrix}a&b\\c&d\end{smallmatrix}\right)\in\SL(2,\ZZ)$, and $\upsilon:\HH\to\CC$ is a holomorphic function satisfying $\upsilon(\tau)^2=(c\tau+d)$. The multiplication is $(\gamma,\upsilon)(\gamma',\upsilon')=(\gamma\gamma',(\upsilon\circ\gamma')\upsilon')$. The elements $\widetilde{T}:=\left(\left(\begin{smallmatrix}1&1\\0&1\end{smallmatrix}\right),1\right)$ and $\widetilde{S}:=\left(\left(\begin{smallmatrix}0&-1\\1&0\end{smallmatrix}\right),\sqrt{\tau}\right)$ form a generating set.

In this article we focus on Weil representations associated to rank one lattices, labelled by a positive integer $m$ \cite{MR2512363}. Later this integer will be determined by the topological data \eq{eqn:mformula}. 
Concretely, consider the Weil representation $\varrho_m$ corresponding to the finite abelian group $\ZZ/2m$ equipped with a quadratic form $\ZZ/2m \to \QQ/\ZZ$ given by $x\mapsto {x^2\over 4m}$.
The unitary representation $\widetilde{\SL(2,\ZZ)}\to\GL_{2m}(\CC)$
generated by the assignments $\widetilde{S}\mapsto{\cal S}$ and $\widetilde{T}\mapsto{\cal T}$,
where
\begin{align}
\label{STmatricesWeil}
{\cal S}_{rr'} & := \frac{1}{\sqrt{2m}}\ex\left(-\frac{rr'}{2m}\right) \,,  \cr
{\cal T}_{rr'} & :=\ex\left(\frac{r^2}{4m}\right)\delta_{r,r'}.
\end{align}
Throughout the paper we set $\ex(x):=e^{2\pi i x}$, and $q:=\ex(\t)$, $y:=\ex(z)$.

The Weil representation $\varrho_m$ is realized by the familiar theta functions:
\be\label{def:usualtheta}
\theta_{m,r}(\tau,z):=\sum_{\ell=r\xmod 2m}q^{\ell^2/4m}y^\ell,
\ee
for $\tau\in \HH$ and $z\in \CC$.
When regarding $\theta_m:=(\theta_{m,r})_{r\xmod 2m}$ as a column vector, it transforms as
\begin{align}
\label{transf_theta}
\theta_m\left(-\frac1\tau,\frac{z}\tau\right)\frac{1}{\sqrt{\tau}}\ex\left(-\frac{mz^2}{\tau}\right)
& = {\cal S}\theta_m(\tau,z) \,,  \cr
\theta_{m}(\tau+1,z)
& ={\cal T}\theta_m(\tau,z)
\end{align}
under $\widetilde {\SL(2,\ZZ)}$, where ${\cal S}$ and 
${\cal T}$ are as in \eq{STmatricesWeil}.
As a result, Weil representations play an important role in the study of Jacobi forms (see \S5 of \cite{eichler_zagier}).

The above shows that $\theta_{m}$ spans a $2m$-dimensional representation of $\widetilde{\SL(2,\ZZ)}$, which we denote by $\Theta_m$.
This representation is reducible for all $m>1$. To see this,  note that the orthogonal group $O_m:=\{a\in \ZZ/2m ~\lvert~ a^2=1 \xmod(4m)\}$ has the natural action
\be \label{Om_action}\theta_{m,r}\cdot a :=\theta_{m,ra}\ee that commutes with $\widetilde {\SL(2,\ZZ)}$. As a result, one obtains a sub-representation by considering eigenspaces of $a\in O_m$.
In fact, for most examples we encounter in this paper, the relevant representations are irreducible!

To label the sub-representations we are interested in, it will be convenient to employ the isomorphism between $O_m$ and $\Ex_m$, the group of the exact divisors of $m$. Recall that a divisor $n$ of $m$ is said to be exact if $(n,{m\over n})=1$, and the groups operation is given by $n\ast n' = {nn'\over(n,n)'^2}$.
For $n\in \Ex_m$ write $a(n)$ for the unique $a\in O_m$ such that
\be\label{def:an}
a(n)=-1\xmod 2n, ~{\rm and }~a(n)=1\xmod 2m/n. \ee The assignment $n\mapsto a(n)$ defines an isomorphism of groups $\Ex_m\xrightarrow{\simeq} O_m$.  Explicitly, the isomorphism is implemented by  the {\em Omega matrix}, defined as
\begin{align}\label{def:OmegaMatrices}
\Omega_{m}(n)_{r,r'} := \begin{cases} 1 &\text{if $r=-r'\xmod 2 n$ and $r=r'\xmod {2m}/{n}$,} \\
0 &{\rm otherwise},~~ r,r'\in \ZZ/2m,
\end{cases}
\end{align}
which is familiar from the classification of modular invariant combinations of chiral and anti-chiral characters of the $SU(2)$ current algebra \cite{Cappelli:1987xt}.

In the main examples in this article (corresponding to Seifert manifolds with 3 singular fibers and involving weight 1/2 quantum modular forms), we always encounter representations which are the $-1$ eigenspaces of the operation \eq{Om_action} for $a(m)=-1$. As a result, we are interested in subrepresentations of $\Theta_m$, labelled by $K\subset \Ex_m$ with $m\not\in K$ (the so-called ``non-Fricke'' property), which is defined as the simultaneous eigenspace of $a(n), n\in K$ with eigenvalue 1, and of $-1=a(m)$ with eigenvalue $-1$. Only in \S\ref{sec:4sing} we will encounter the ``Fricke'' cases where $m\in K$.

In terms of notations, following a tradition initiated in \cite{conway_norton}, we denote the pair $(m,K)$  by $m+K = m+ n,n',\dots$ for $K = \{1,n,n',\dots\}$. Subsequently, we denote by $\Theta^{m+K}$ the corresponding  sub-representation defined above. Especially interesting choices of $K$ are those such that the above prescription renders a 
simultaneous eigenspace of all $O_m$. Concretely, this happens when $K$ is large enough such that 
$\Ex_m=K\cup (m\ast K)$. For such a choice of $K$, and when $m$ is not divisible by any square number\footnote{When $m$ is not square-free, the irreducible representation is given by taking the orthogonal complement of the images of operators $U_d: \Theta_{m} \to \Theta_{md^2}$ given by   $U_d(\f(\t,z)) = \f(\t,dz)$ with respect to the so-called Petersson metric in the space $\{\phi\in\Theta~\lvert~ \phi\cdot a = \alpha(a) \phi \}$ \cite{Sko_Thesis}.}, the resulting representation is {\em irreducible}. This will be the case in most of our examples.

Concretely, to implement the projection onto eigenspaces we introduce the projection operators, given by the matrices
\be
\label{def1:proj}
P_m^\pm(n)= ({\mathbb I}\pm\O_m(n))/2,
\ee
and
\be
\label{eqn:proj}
P^{m+K}  =\big(\prod_{n\in K} P_m^+(n) \big)P_m^-(m)
\ee
when $m$ is square-free. Extra care needs to be taken when $m$ is divisible by a square. For instance, when $m=p^2 m'$ where $m'$ is square-free and $p$ is prime, we have
\be
 P^{m+K}  =\big(\prod_{n\in K} P_m^+(n) \big)P_m^-(m) ({\mathbb I}- \O_m(p)/p).
\ee
Using the above projection operator, we define for $r\in \ZZ/2m$
\be
 \th^{m+K}_r =2^{|K|} \sum_{\ell \in \ZZ/2m} P^{m+K}_{r\ell}\theta_{m,\ell}.
\ee
Denote by $r\in\sigma^{m+K}$ the set of unequal (up to a sign) vectors $\th^{m+K}_r$.
A specific basis for $\Theta^{m+K}$ is then given by $\{\th^{m+K}_r , ~ r\in\sigma^{m+K}\}$.

Explicitly, the S-matrix of the  sub-representation $\Theta^{m+K}$ is given by
\be\label{Smatrix}
{\cal S}^{m+K}_{rr'} = \sum_{\ell \in \ZZ/2m} {{\cal S}_{r \ell}  P^{m+K}_{\ell r'}\over P^{m+K}_{r' r'}} , ~~ r,r'\in {\sigma}^{m+K},
\ee
which can be understood from the fact that, given an element $r$ in $\sigma^{m+K}$, the number of $\ell \in \ZZ/2m$ such that $ \th^{m+K}_\ell=\pm \th^{m+K}_r$ is precisely $1/P^{m+K}_{r r}$.
It is easy to check that  indeed $({\cal S}^{m+K})^2 = -{\rm Id}.$
As can be easily deduced  from \eq{STmatricesWeil}, the corresponding $T$ matrix is simply given by the diagonal matrix
\be\label{Tmatrix}
{\cal T}^{m+K}_{rr'} =\ex\left(\frac{r^2}{4m}\right)\delta_{r,r'} .
\ee

As an example, let us take $m=6$ and $K=\{1,3\}$. Since $\Ex_6=\{1,2,3,6\} = K \cup 6\ast K$, see that the resulting representation $\Theta^{6+3}$ is irreducible. Following the above discussion,
a simple calculation leads to $\sigma^{6+3}= \{1,3\}$ and the corresponding basis vectors are
\begin{gather}
\begin{split}
\theta^{6+3}_1 &=  \theta_{6,1}+\theta_{6,5}-\theta_{6,-1}-\theta_{6,-5}   \\
\theta^{6+3}_3 &=2\left( \theta_{6,3}-\theta_{6,-3}  \right),
\end{split}
\end{gather}
and the S-matrix is
\be
{\cal S}^{6+3} = {i\over \sqrt{3}} \bem -1 & -1 \\ -2 & 1\eem.
\ee

\section{Resurgence and modularity}
\label{sec:resurgence}

We will see in this section how the third type of modular representations discussed in the previous section -- the Weil representations -- are materialized in the form of false theta functions in our problem.
To see the connection to topology and physics, we discuss their origin as Eichler integrals, and
analyze their asymptotic expansions near the cusps, following \cite{lawrence1999modular}. Subsequently in \S\ref{subsec:resurgence_modular} we demonstrate the relation between  Eichler integrals and  resurgence analysis, and highlight the fact that the transseries coefficients are given by the S-matrix entries of the Weil representation.
Finally in \S\ref{sec:modularity} we use these properties of the false theta functions to deduce predictions for topological information on the $\SL(2,\CC)$ flat connections of the relevant three-manifolds.

\subsection{False theta functions and the asymptotic expansions}
\label{subsec:partial}

Associated to the Weil representations discussed earlier are also the weight $3/2$ unary theta functions, defined for $\tau \in \HH$ in the upper-half plane:
\be\label{unary theta}
\theta^1_{m,r}(\tau) =\sum_{\substack{\ell\in\ZZ\\\ell=r\xmod 2m}}\ell \, q^{\ell^2/4m} ,
\ee
related to the theta function by the operator $$\theta^1_{m,r}(\tau) : = {1\over 2\pi i} {\pa \over \pa z}\theta_{m,r}(\t,z)\lvert_{z=0}.$$

In the context of Seifert three-manifold one often encounters its Eichler integral.  The Eichler integral of a cusp form $g=\sum_{n>0} a_g(n) q^n$ of weight $w$, which can be either integer or half-integer, is defined as
\be\label{def:eichler}
\widetilde{g}(\tau) : =\sum_{n>0} n^{1-w} a_g(n) q^n .
\ee
Note that this is equal to the following integral for integral  $w$ \footnote{We choose the branch to be the principal branch $-\pi < {\rm arg}\, x \leq \pi$. }
\be\label{def:Eichler_int1}
\widetilde{g}(\tau) = C \int_{\tau}^{i\infty} g(z' ) (z'-\t)^{-2+w} dz',
\ee
where $C={(2\pi i)^{w-1}\over \Gamma(w-1)}$.
In our case of \eq{unary theta} we have $w=3/2$ and  the Eichler integral has the following Fourier expansion ({\it cf.} \S\ref{sec:WeilRep}):
\be
\label{def:partial1}
\Psi_{m,r}(\tau) :=\widetilde{ \theta^1_{m,r}}(\t) = 2 \sum_{n>0}
(P^-_m(m))_{r,n} \,q^{n^2/4m},
\ee
and is often referred to as a {\it false theta function}.
In the above we have written $(P^-_m(m))_{r,n}$ as the entry of the matrix \eq{def1:proj} corresponding to the  $r$ and $n \xmod {2m}$.  Explicitly,
we have
\be\label{def:coeff_partial1}
2(P^-_m(m))_{r,n} = \begin{cases}\pm 1 & n=\pm r\xmod{2m}\\ 0 &{\rm otherwise}\end{cases}~.
\ee
Note that $\theta^1_{m,r} = -\theta^1_{m,-r}$ and consequently $\Psi_{m,r} = - \Psi_{m,-r}$, and  this is the reason why in \S\ref{sec:WeilRep} we focus on sub-representations contained in the $-1$ eigenspace under the action \eq{Om_action} with $-1=a(m)$ (the ``non-Fricke'' type). Clearly, both $\theta_{m,r}^1(\tau)$ and $\Psi_{m,r}(\tau)$ have  Fourier expansions that converge in the unit disk $|q|<1$. The false theta function is not a modular form, 
but naturally leads to a quantum modular form as we will review later.

To explain the nomeclature,
note that the  functions defined in \eq{def:partial1} can also be expressed as
\be\label{def:partial2}\Psi_{m,r}(\tau) = \sum_{\substack{\ell\in\ZZ\\\ell=r\xmod 2m}} {\rm sgn}(\ell) \,q^{\ell^2/4m}.   \ee
In \cite{MR557539} Andrews defined a false theta function to be a function of the form $$\sum_{n\in \ZZ} (\pm 1)^n q^{kn^2+\ell n}. $$
Since without the sign factors these are just the usual theta functions $\theta_{m,r}(\t,0)$, they are also often referred to as false theta functions\footnote{As  is clear from \eq{def:partial1}, the functions $\Psi_{m,r}$ also have the property that they are just like  (linear combinations of) ordinary theta functions except for that the sum is performed only over part of the lattice. As a result, they are sometimes also referred to as partial theta functions \cite{BrFoMi,Creutzig:2013zza,BrFoRh}. }.

In what follows we will be interested in the Eichler integral of the basis vectors discussed in \S\ref{sec:WeilRep}, given by
\be\label{fouriercoeff_foldedpartial}
\Psi^{m+K}_r : =\widetilde{\th^{m+K,1}_r}= 2^{|K|}\sum_{n\geq 0} P^{m+K}_{r,n}\,q^{n^2/4m} .
\ee
where similarly to $\theta^{1}_{m,r}$, we have defined  $\th^{m+K,1}_r (\t) : = {1\over 2\pi i} {\pa \over \pa z}\theta_{r}^{m+K}(\t,z)\lvert_{z=0}.$
Later we will see how these false thetas come to life as $q$-series invariants $\widehat{Z}_a (M_3)$ attached to certain three-manifolds $M_3$ and how their transseries give non-trivial predictions about $\SL(2,\CC)$ flat connections on $M_3$.

The relation between these false theta functions and WRT invariants was first pointed out in \cite{lawrence1999modular}
and further developed in \cite{Hikami,hikami1,hikami2011decomposition} as stemming from the following two facts:
\begin{itemize}
\item The false theta functions give finite values in the radial limit $\t \to {c\over d}  \in \QQ$ from the upper-half plane, and when $\t\to1/k$ they reproduce the WRT invariants at level $k$.
\item  The asymptotic expansion of the false theta functions near $\tau = 0$ captures the perturbative expansion (Ohtsuki series, or $1/k$ expansion) around the trivial flat connection in Chern-Simons TQFT.
\end{itemize}
Next we briefly discuss the relevant number theoretic properties of the building block false theta $\Psi_{m,r}$ which are responsible for the above matching.
Note that the false theta functions defined in \eq{def:partial1} and \eq{def:partial2} have Fourier coefficients with certain periodicity property (see \eq{def:coeff_partial1}) which moreover have vanishing mean value. For such a function $C: \ZZ \to \CC$, it was shown \cite{lawrence1999modular} that the corresponding $L$-series  $L(s,C) = \sum_{n\geq 1} n^{-s} C(n)$, $\Re(s)>1$, can be holomorphically extended to all  $s\in \CC$ and the following two functions have the asymptotic expansions given by
\begin{align}
\label{asymp_Lfunc}
\sum_{n \geq 1} C(n) e^{-n t} & ~\sim~ \sum_{\ell \geq0}L(-\ell,C){(-t)^\ell \over \ell!}, \cr
\sum_{n \geq 1} C(n) e^{-n^2 t} & ~\sim~ \sum_{\ell \geq0}L(-2\ell,C){(-t)^\ell \over \ell!}
\end{align}
for $t>0$.
From the above, both the radial limit values at $\tau \to 1/k$ and the asymptotic series near 0, when approaching from the upper-half plane,   can be computed and compared to the known result on the WRT invariants of the corresponding three-manifold. In the former case, we take $(P^-_m(m))_{r,n} \, \ex({-r^2\over 4mk})$  to be $C(n)$
and we set it to be $C_{m,r}(n):=(P^-_m(m))_{r,n}$ in the latter case.
The result of the calculation yields the asymptotic series
\be\label{asym1}
\Psi_{r}^{m+K}({it\over 2\pi}) \sim  \sum_{\ell\geq 0} {L(-2\ell,C_{\ell}^{m+K})\over \ell!} \big({-t\over 4m}\big)^{\ell}.
\ee
where we have taken $C_{r}^{m+K}(n):=P^{m+K}_{r,n}$.

Moreover, the relevant $L$-values are conveniently captured by the ratios of the sinh functions :
\be\label{sinh_asymp}
{\sinh((m-r)z)\over \sinh(mz)} = \sum_{\ell \geq 0}  {L(-2\ell,C_{m,r}) \over (2\ell)!}\, z^{2\ell},
\ee
obtained from applying the  asymptotic expansion in \eq{asymp_Lfunc} to the identity
\be\label{sinh1_expand}
{(x^{m-r}- x^{-m+r}) \over (x^{m}- x^{-m})} = \sum_{n>0}
(P^-_m(m))_{r,n} \,x^{n}.
\ee

We will see that  the above relations to sinh functions play an interesting role
in the resurgence interpretation of the $q$-series invariants $\widehat{Z}_a (q)$.

\subsection{Resurgence and Eichler integrals}
\label{subsec:resurgence_modular}

Anticipating the role of the false theta functions $\Psi^{m+K}_{r}$ as homological blocks,
in this section we study the transseries expression of the false theta function \eq{def:partial1},
which admits a simple physical interpretation in the context of resurgence, as pointed out in \cite{Gukov:2016njj}.
Apart from the asymptotic series \eq{asym1} computed in the previous subsection, one can moreover compute the non-perturbative part of $\Psi_{m,r}(\tau=1/k)$ and obtain the whole transseries. 
The latter captures the important information regarding flat $\SL(2,\CC)$ connections on the 3-manifold $M_3$.
This was first done in \cite{lawrence1999modular}, where it was demonstrated that the false theta function is modular near rational points up to a smooth function.
The transseries calculation is closely related to the  quantum modular  structure of false theta functions, which we  will discuss  in  details  in \S\ref{sec:otherside}. In this subsection we focus on the resurgence point of view of  \cite{Gukov:2016njj}.
Moreover, we stress that the resurgence sheds light on the origin of the appearance of the Eichler integrals in our problem, as the structure of the Eichler integrals arises from the resurgence calculation quite naturally \eq{sihn_int_theta}.

Resurgence is a method to sum up the infinite perturbative series arising from perturbative quantum field theories, which are often asymptotic instead of convergent series, into a complete function by incorporating the non-perturbative contributions. It relies on the techniques of {\em Borel resummation}, which we now describe briefly. Given a non-convergent series
\be\label{eqn:borel1}
Z_{\rm pert}(k) = \sum_{n} {a_n \over k^n}
\ee
we consider its Borel transform
\be
BZ_{\rm pert}(z) = \sum_{n} {a_n\over \Gamma(n)}{z^{n-1}}
\ee
which typically defines a function that is analytic in a neighbourhood near the origin.
We are then interested in the Borel sum of  $Z^{\rm pert}$, given  by
\be\label{eqn:borel3}
\int e^{-\tau z}BZ_{\rm pert}(z) dz,
\ee
where we have on purpose left the contour of integration unspecified at this stage.

Due to the role of false theta functions as the half-indices $\widehat{Z}_b$, we are interested in applying the resurgence analysis to the building blocks $\Psi_{m,r}$ \cite{Gukov:2016njj}.
There is however an important subtlety: note that there is additional overall $k$ dependence in the Chern--Simons partition function not captured by the homological blocks \eq{relation_CS_blocks}.
As we have $b_1(M)=0$ for our three-manifolds $M$, there is an overall factor of $1/\sqrt{k}$ multiplying the false theta functions evaluated at $\tau \to -1/k$. From the modular point of view, this $1/\sqrt{k}$ factor stems from the fact that  $\Psi_{m,r}$ is a  weight  $1/2$ quantum modular form (see \S\ref{sec:QMF}).

Comparing \eq{asym1} and \eq{sinh_asymp}, we conclude that the corresponding Borel transform is
\be\label{exact_boreltrans}
B\left({1\over \sqrt{k}}\Psi_{m,r}(\tfrac{1}{k})\right)(z) = {1\over \sqrt{\pi z}} {\sin((m-r)\sqrt{2\pi z \over m}) \over \sin(m\sqrt{2\pi  z \over m})} .
\ee

On the other hand, by performing a Gaussian integral on both sides of \eq{sinh1_expand} we obtain the following identity
 \be\label{resurgence_equality}
 {1\over \sqrt{k}}\Psi_{m,r}(\tfrac{1}{k}) = {\sqrt{i}\over 2} \left(\int_{e^{i\delta}\RR_+}  +\int_{e^{-i\delta}\RR_+}  \right)  {dz\over \sqrt{\pi z}} {\sin((m-r)\sqrt{2\pi z \over m}) \over \sin(m\sqrt{2\pi  z \over m})}e^{-ikz},
 \ee
 which in light of \eq{exact_boreltrans} can be interpreted as an exact Borel resummation \cite{Gukov:2016njj}.
Note  that the integral has poles at $z= 2\p {n^2\over 4m}$, and the residue is given by
 \be
\underset{z=2\p {n^2\over 4m}}{\resi} \left({\sqrt{i}\over 2} {1\over \sqrt{\pi z} } {\sin((m-r)\sqrt{2\pi z \over m}) \over \sin(m\sqrt{2\pi  z \over m})}e^{-ikz}\right)
= -{  \sqrt{i} \over \pi \sqrt{2m} }
 \sin({r\pi n\over m}) \ex(-k{n^2\over 4m})
 \ee
 for $n\in \ZZ_{>0}$.
 Note that the right-hand side is, up to an overall constant, the S-matrix \eq{Smatrix} corresponding to the sub-representation of the Weil representation $\Theta_m$ specified by eigenvalue $-1$ of the action \eq{Om_action} of $-1=a(m)$  (equivalently, this is the S-matrix of the unary theta function $\theta^1_{m,r}$ in \eq{unary theta}, corresponding to $K=\{1\}$ in the notation of \eq{Smatrix}):
\be
 {\cal S}^{m}_{r,n} = ({\cal S}P_{m}^- (m))_{r,n} = {-i\over \sqrt{2m}}\sin\left( {rn  \pi\over m }\right)
 \ee
 To sum up the contribution from the infinitely many poles lying on the upper half of the imaginary axis, we use the regularization in \eq{asymp_Lfunc} and \eq{sinh_asymp}
 \be
  \sum_{n\geq 0} (P^-_m(m))_{r,n} =\lim_{t\to 0^+}  \sum_{n\geq 0} (P^-_m(m))_{r,n}e^{-nt} = \lim_{t\to 0^+}  {\sinh((m-r)t)\over \sinh(mt)}  =   1- {r\over m}.
 \ee
 Applying the above to $\Psi^{m+K}_{r}$ by taking the linear combination, we see that the corresponding integral has groups of poles labelled by the set $\s^{m+K}$, and their corresponding contribution to the integral is  given by
 \be\label{trans_false}
{1\over \sqrt{k}}\Psi_{r}^{m+K}(\tfrac{1}{k})  =
  -2 {\sqrt{i} }  \sum_{r'\in \s^{m+K}}{\cal S}_{rr'}^{m+K}c_{r'}\,  e^{-2\p i{k {r'^2\over 4m}}} + {\rm perturbative ~ part},
 \ee
 where
 \be \label{eq:c_r} c_{r}:=2^{|K|} \sum_{\ell=1}^{m-1}  P_{\ell r}^{m+K} (1-{\ell\over m}).\ee
  The first term in the formal transseries  expression \eq{trans_false} encodes the contributions from the poles of the Borel transform and captures the non-perturbative contribution to the path integral.
  The second term is given by \eq{asym1} in the limit $k\to \infty$, and
   captures the asymptotic expansions, corresponding to the Ohtsuki series on the topology/Chern-Simons side. In the next subsection we will expand on the physical and topological interpretation of the above transseries, and deduce non-trivial predictions about the flat connections on three-manifolds.

Note that the same regularization procedure gives the radial limit
\be
\Psi^{m+K}_{r}(-k)  = \ex(-k{r^2\over 4m}) c_r~, 
\ee
and we can write
\be\label{trans_false2}
{1\over \sqrt{k}}\Psi_{r}^{m+K}(\tfrac{1}{k})  = {2\over \sqrt{i}} \sum_{r'\in \s^{m+K}}  {\cal S}^{m+K}_{r,r'} \Psi^{m+K}_{r'}(-k)
 + {\rm perturbative ~ part}.
\ee
The above states that $\Psi_{r}^{m+K}$ has modular property up to a smooth function, which is precisely the statement that $\Psi_{r}^{m+K}$ gives rise to a quantum modular form. The above relation will be derived and explained in another context in \S\ref{sec:otherside}.

Finally we remark on the relation between the
 resurgence integral \eq{exact_boreltrans} and the Eichler integral \eq{def:Eichler_int1}, drawing on results in \cite{MR1874536}.
Note that, although the two expressions for the false theta function $\Psi_{m,r}$, evaluated at the cusp $\tau \to {1\over k}$, look very different, they are in fact extremely closely related. To see this, note that upon an obvious change of variables the resurgence integrals \eq{resurgence_equality} can be rewritten as an integral of
\be
{e^{-y^2/\tau}\over \sqrt{\tau}} {\sinh((m-r)\pi y) \over \sinh(m\pi y)} = C_m{e^{-y^2/\tau}\over \sqrt{\tau}} \lim_{n_\ast\to \infty} \sum_{ n= -n_\ast }^{n_\ast} {\sin(r\pi {n\over m})\over {y-i{n\over m}}}
\ee
where $C_m$ is a unimportant constant that depends only on $m$. Using the equality between two integrals
\be
\int_{-\infty}^{\infty} {e^{-\pi t y^2}\over y-ir} dy =
\pi i r \int_{0}^\infty {e^{-\p r^2 u} \over \sqrt{u+t} }du
\ee
and exchanging the sum and the integral, as was done in the proof for Lemma 3.3 of \cite{MR1874536}, one immediately see that
\be\label{sihn_int_theta}
\int_{0}^\infty dy\, {e^{-y^2/\tau}\over \sqrt{\tau}} {\sinh((m-r)\pi y) \over \sinh(m\pi y)} = c \int_0^{\infty} {du} {\theta^1_{m,r}(u)\over\sqrt{u+\tau}}
\ee
with some unimportant factor $c\in \CC$. Clearly, the above calculation also extends easily when $\Psi_{m,r}$ is replaced with the folded false theta $\Psi^{m+K}_r$.
As we will see, this is precisely the period integral \eq{periodEich} and whose appearance in the resurgence computation for the false theta function can be understood through the identity \eq{asymp_false} between the Eichler integral and the non-holomorphic Eichler integral as far as the asymptotic series are concerned.

\subsection{Flat connections from modularity}
\label{sec:modularity}

In this subsection we will explain how to extract information about flat connections on three-manifolds from
the modular-like properties of the false theta functions discussed in \S\ref{subsec:resurgence_modular}.
As mentioned earlier, our main class of examples is Seifert manifolds with three singular fibres, although similar ideas and methods are also applicable to more general examples.

The starting point is the observation (which will be profusely demonstrated in \S\ref{sec:examples}) that the $q$-series invariants $\widehat{Z}_b (M_3)$ can often be expressed as
\be
\label{def:Emb}
\widehat{Z}_b (M_3) \; = \; c  \left( q^\delta  \Psi_r^{m+K} + d \right)
\,,\qquad c\in \CC,\quad \delta\in \QQ
\,,\quad d \in \ZZ [q]
\ee
where $b$ denotes a boundary condition, as in Figure~\ref{fig:halfindex},
and $d$ is a polynomial in $q$ (typically, just a single monomial).\footnote{The origin of $d$ may seem
somewhat unclear, especially when compared to \eqref{Zan}. At a technical level, it originates from
a regularization of an infinite sum, which we expect to be cured by introducing $t$-dependence or, equivalently,
working at the categorified level, with the space of BPS states. One could also think of \eqref{def:Emb}
as a sum of two blocks $\widehat{Z}_b (M_3)$, which happen to have the same value of $\text{CS} (b)$
and such that one of them is $d$.}
Note that, for a given three-manifold $M_3$, changing the boundary condition $b$ changes the corresponding $r\in \sigma^{m+K}$ but the Weil representation labelled by $m+K$ remains the same. In other words, given the same three-manifold $M_3$ (or bulk 3d $\N=2$ theory)  we will have the analogous relation between $\widehat{Z}_{b'}(M_3)$ and $\Psi_{r'}^{m+K}$. For this reason, in the rest of this subsection it will be convenient to omit the label ${m+K}$ in order to avoid clutter.

The resulting homological blocks \eqref{def:Emb} are combined into $Z_a$,
labelled by abelian flat connections \eqref{eqn:flatconn} 
\be
\label{relateZhatZ}
Z_a(k) =  \ex(k\lambda(a,a)) \sum_b S^{(A)}_{ab} \widehat{Z}_b \lvert_{q\to \ex({1\over k})} . 
\ee
These functions $Z_a$ have the interpretation as the Borel resummed perturbative expansions near the corresponding abelian flat connections, and are further assembled into $SU(2)$ Chern-Simons partition function \eq{relation_CS_blocks} upon specialization $q \to \ex(\tfrac{1}{k})$ and a sum over  the abelian flat connections ``$a$''.

As we have seen in the previous subsection, a given $r\in \sigma^{}$ labels a group of poles in the integral expression for \eq{resurgence_equality} for $\sqrt{\t} \Psi_{m,s}(\t)$, each contributing a residue given by ${\cal S}_{rs}$ (up to an unimportant overall factor). From \eq{def:Emb} we see that they translate into poles giving contribution to the homological blocks, and hence should correspond to certain saddle point configuration of the path integral formulation of the half-index.
When combined into the physical quantities  $Z_a$ and $Z_{CS} (M_3)$ that arise from Chern--Simons theory,
{\it cf.} \eq{eqn:flatconn} and \eq{relateZhatZ}, the following things can happen to these  poles:\footnote{The Chern-Simons partition function $Z_{CS} (M_3)$ coincides with $\text{WRT}(M_3,k)$ when the gauge group is $SU(2)$ and the level is $k$. It is normalized such that: $$Z_{CS}(S^2 \times S^1) = 1, \quad \text{and} \quad Z_{CS}(S^3) = \sqrt{\frac{2}{k}}\sin\left({\frac{\pi}{k}}\right).$$}
\begin{enumerate}
\item\label{case1} A pole contributes to the transseries of $\widehat{Z}_b$, but this contribution vanishes
upon further re-assembly into $Z_a$.
\item A pole contributes  to the transseries of $\widehat{Z}_b$ and the contribution does not vanish when  re-combined into $Z_a$.
\begin{enumerate}
\item\label{case2} Moreover, the contribution to $Z_a$ do not vanish when combined further into $Z_{CS} (M_3)$ by summing over $a$ and  over all the (infinitely many) poles in the group.
\item \label{case3} The contributions to $Z_a$ sum up to zero after performing the two additional sums that gives $Z_{CS} (M_3)$.
\end{enumerate}
\end{enumerate}
\noindent
{}From the physical interpretation of $\widehat{Z}_b$, $Z_a$ and $Z_{CS} (M_3)$ in the 3d-3d correspondence,
we can give the following physical interpretation to the above types of poles:
\begin{enumerate}
\item
\label{case1}
They correspond to ``phantom'' saddles of the path integral for $\widehat{Z}_b$ that may not even correspond to flat $\SL(2,\CC)$ connections on $M_3$. (For example, ``renormalon'' saddles are familiar examples of this behavior in resurgent analysis of QFT.)
\item
They correspond to saddle points of the path integral for $\widehat{Z}_b$ that arise from non-abelian $\SL(2,\CC)$ flat connections on $M_3$. (Note that according to a Theorem in \cite{Gukov:2016njj}, only non-abelian flat connections can appear in transseries contributions to a Borel resummation of a perturbative expansion around an abelian flat connection.)
\begin{enumerate}
\item\label{case2} Moreover, they correspond to ``real'' non-abelian flat connections that can be conjugated inside $G=SU(2)$.
\item \label{case3} They correspond to ``complex'' non-abelian flat connections that can {\it not} be conjugated into $G=SU(2)$.
\end{enumerate}
\end{enumerate}
\noindent
As a result, from the modularity of the false theta functions we can read off the behaviour of the different poles of the integral expression for homological blocks, and thereby deduce predictions on non-abelian flat connections as above.

To turn words into equations, we define the following quantities.
Let $n_B = |\sigma|$ be size of the Weil representation described in \S\ref{sec:WeilRep},
and denote by $n_A$ the number of abelian $\SL(2,\CC)$ flat connections on $M_3$,
{\it i.e.} the size of the modular $S$-matrix \eqref{def:Sabelian}.
Consider the two matrices
\begin{gather}\label{M3matrices}
\begin{split}
S(M_3) = S^{(A)}.{\bf Emb}. (S^{(B)})^{-1} \\
T (M_3) = T^{(A)}.{\bf I}. (T^{(B)})^{-1}
\end{split}
\end{gather}
where {\bf Emb} and {\bf I} are $n_A \times n_B$ matrices.
The first is the embedding matrix defined by  {\bf Emb}$_{ar} =c$ iff \eq{def:Emb} holds.
In particular, {\bf Emb}$_{ar} =0$  when the $q$-series $\widehat{Z}_a$ does not involve the false theta function $\Psi_{m,r}$.
The second matrix ${\bf I}$ is the matrix with all entries equal to one.
The prediction, reflecting the interpretation 2-(i),  is then
\begin{gather}
\label{rule_nonAb}
\boxed{
\begin{split}
&\left\{~ \ex(\text{CS}({\mathfrak a}))~\lvert~ {\mathfrak a} {\text{ is a non-abelian $\SL(2,\CC)$ flat connection on $M_3$}}~\right\}  \\
&~~~~= ~\left\{~ T (M_3)_{ar} ~\lvert~ a, r ~\text{such that}~ S(M_3)_{ar}\neq 0  ~\right\}
\end{split}}
\end{gather}
Lets denote the elements of the set on the right-hand side by $\ex(\alpha)$, and write
\be
\sum_{a,r}  T(M_3)_{a,r} S(M_3)_{a,r} c_r = \sum_\a \ex(\alpha) \, {\bf c}_\alpha.
\ee
In other words, we have
\be\label{def:calpha}
{\bf c}_\a  = \sum_{ (a,r)  } S(M_3)_{a,r}c_{r}
\ee
where the sum on the right-hand side is over the pairs $(a,r)$ satisfying $T (M_3)_{ar}=\ex(\alpha)$.
In terms of these quantities, the interpretation 2-(ii) translates into the following prediction on complex
non-abelian flat connections:
\begin{gather}
\label{rule_complex}
\boxed{
\begin{split}
&\left\{~ \ex(\text{CS}({\mathfrak a}))~\lvert~ {\mathfrak a} {\text{ is a non-abelian  $SU(2)$ flat connection on $M_3$}}~\right\}  \\
&~~~~= ~\left\{~ \ex(\alpha) ~\lvert~ {\bf c}_{\alpha}\neq 0~\right\}
\end{split}}
\end{gather}

\begin{figure}[h!]
\begin{center}
\begin{tikzpicture}
  [node distance=.8cm,
  start chain=going below,]
     \node[punktchain, join] (step1) {Compute $\widehat{Z}_b$};
     \node[punktchain, join] (step2)      {Identify Weil representation $(m,K)$ \\ with $\widehat Z_b = cq^\delta (\Psi^{m+K}_r +d)$};
     \node[punktchain, join] (step3)      {Compute the modular matrices \\ $S(M_3)$ and $T(M_3)$ \eq{M3matrices} };
     \node[punktchain, join] (step4) {Find non-abelian flat \\ connections \eq{rule_nonAb}};
     \node[punktchain, join] (step5) {Compute $\ex(-\alpha)$ and ${\bf c}_\a$ \eq{def:calpha}};
     \node[punktchain, join ] (step6) {Find complex flat \\ connections \eq{rule_complex}};
\end{tikzpicture}
\caption{\small From plumbing data to flat connections.}
\label{fig:reverse}
\end{center}
\end{figure}
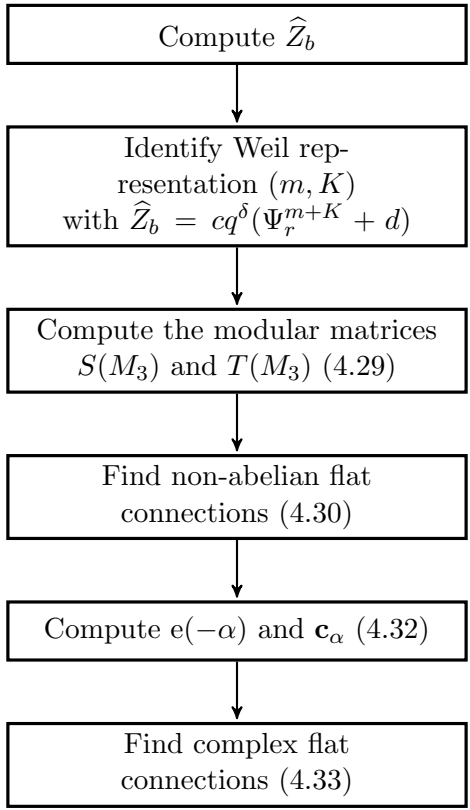

In operational terms, the steps of retrieving the information about non-abelian flat connections from the plumbing data of a three-manifold are summarized in Figure~\ref{fig:reverse}. We note that the above rules only give the set of the corresponding Chern--Simons invariants (mod $\ZZ$) and a priori cannot distinguish different flat connections with the same Chern--Simons invariants.

\section{Logarithmic CFTs from three dimensions}
\label{sec:logCFT}

Unusual modular transformations of the combined 3d-2d indices \eqref{BlockD2S1}
and 3-manifold invariants $\widehat{Z}_a (M_3)$ studied in this paper also
appear as one of the key features in logarithmic conformal field theories (log-CFTs, for short).
The goal of this section is to explain, qualitatively as well as quantitatively,
that this is not an accident and there are good reasons why half-indices of 3d-2d combined systems
and $q$-series invariants $\widehat{Z}_a (M_3)$ in many cases are expected to be related to log-CFTs.

Among other things, this offers a new way of looking at logarithmic CFTs, connecting them to
supersymmetric 3d $\N=2$ theories, including theories $T[M_3]$ coming from 3-manifolds.
We hope that, in the future, this new perspective will help to shed light on still rather
mysterious nature of log-CFTs.

\subsection{$\widehat{Z}_a (M_3)$ as characters of log-VOAs}

The first qualitative, yet conceptual indication that our setup illustrated in Figure~\ref{fig:halfindex}
has something to do with logarithmic CFTs comes from the fact that, in many cases, log-CFTs can be thought of
as ``deformations'' of more familiar ordinary conformal field theories, such as free theories and lattice VOAs.
For example, as we review shortly, this perspective has been very successful in constructing
various log-VOAs as kernels of screening operators \cite{Feigin:2005zx,Feigin:2005xs},
which are larger compared to cohomologies of the same screenings used in the construction
of minimal models \cite{Felder,TsuchiyaKanie}.

This is similar to how, as explained in \S\ref{sec:3dphysics}, the 
interaction with 3d degrees of freedom can ``deform'' the standard modular properties of $\Theta_{2d}^{(a)} (x)$ and give rise 
 to objects such as false theta functions, mock modular forms, or more general quantum modular forms.
Recall \cite{Gadde:2013wq}, that $\Theta_{2d}^{(a)} (x)$ is the elliptic genus of 2d $\N=(0,2)$ boundary theory $\mathcal{B}_a$.
When coupled to 3d $\N=2$ theory, its elliptic genus is no longer modular in the traditional sense and,
as a result, the combined index \eqref{BlockD2S1} can become a pseudocharacter of the type we already encountered in \S\ref{sec:MTC}.

For example, for logarithmic VOAs constructed from free fields and screening operators,
it is natural to expect that free fields describe $\mathcal{B}_a$, whereas screening operators
correspond to coupling with 3d $\N=2$ theory.
Relegating a more systematic study of this interpretation to future work, here we note that
concrete expressions for $\Theta_{2d}^{(a)} (x)$ in our examples indeed involve characters of lattice VOAs,
{\it cf.} \eqref{eqn:gentheta} and \eqref{Zunnormalized} below.

Simple examples of logarithmic VOAs constructed from free fields and screening operators
are the singlet and triplet $(1,p)$ models, originally introduced in \cite{Kausch:1990vg}.
In both cases, the starting point is a free scalar field $\varphi$ with the OPE
\be
\partial \varphi (z)
\, \partial \varphi (w)
\; \sim \; \frac{1}{(z-w)^2}
\ee
and the stress-tensor, {\it cf.} \eqref{aaa},
\be
T(z) \; = \; \frac{1}{2} \partial \varphi (z) \partial \varphi (z) + \frac{\alpha_0}{2} \partial^2 \varphi (z).
\ee
The modes of $\partial \varphi (z)$ generate the Heisenberg algebra $[a_m,a_n] = m \delta_{m+n,0} {\bf 1}$,
while the modes of $T(z)$ generate the Virasoro algebra with the central charge \eqref{1pccharge}.
There are two screening operators, often called ``long'' and ``short'' screening operators, respectively:
\be
S_+ \; = \; \oint e^{\alpha_+ \varphi}
\,, \qquad \qquad
S_- \; = \; \oint e^{\alpha_- \varphi}
\label{ScreeningS}
\ee
that commute with the stress-tensor, {\em i.e.} $[S_{\pm}, T(z)] = 0$.

Then, the singlet and triplet $(1,p)$ vertex algebras are realized
as kernels of the ``short'' screening
operator \cite{Adamovic:2002nj,Fuchs:2003yu,Feigin:2005zx,Feigin:2005xs,Nagatomo:2009xp}:
\be
\mathcal{M}_p \; = \; \text{Ker}_{\mathcal{F}_0} \, S_-
\ee
\be
\mathcal{W}_p \; = \; \text{Ker}_{\mathcal{V}_L} \, S_-
\ee
on the Heisenberg algebra $\mathcal{F}_0$ (= the Fock space of weight 0)
and on the lattice VOA $\mathcal{V}_L$ for $L = \alpha_+ \mathbb{Z} = \sqrt{2p} \mathbb{Z}$, respectively.
In other words, $\mathcal{W}_p$ is a maximal local subalgebra of $\mathcal{V}_L$
in the kernel of the ``short'' screening operator $S_-$,
and $\mathcal{M}_p = \mathcal{F}_0 \cap \mathcal{W}_p$ is the analogous subalgebra of $\mathcal{F}_0$.
This gives an alternative description of the tripled algebra $\mathcal{W}_p$
that we already discussed in \S\ref{sec:MTC},
and both algebras $\mathcal{M}_p$ and $\mathcal{W}_p$ have the central charge \eqref{1pccharge}.

The singlet $(1,p)$ algebra has Fock modules $\mathcal{F}_{\lambda}$ of highest weight $\lambda \in \CC$
(also called Feigin-Fuchs modules when understood as Virasoro modules),
and modules $M_{1,s}$ with $1 \le s \le p$.
Their characters take the form \cite{Flohr:1995ea,Adamovic:2007qs}:
\be
\chi({\mathcal{F}_{\lambda}};\tau) \; = \;
\frac{q^{\frac{1}{2} (\lambda - \frac{\alpha_0}{2})^2}}{\eta (q)}
\ee
\be
\chi ({M_{1,s}};\tau) \; = \;
\frac{1}{\eta(q)} \sum_{n \ge 0}
\left( q^{\frac{1}{4p} (2pn + p - s)^2} - q^{\frac{1}{4p} (2pn + p + s)^2} \right)
\; = \; \frac{\Psi_{p,p-s} (\tau)}{\eta (\tau)}
\label{singletchar}
\ee

Before we identify these characters with 3d-2d indices \eqref{BlockD2S1}
and $q$-series invariants of 3-manifolds, we should point out that,
following \cite{Gukov:2016gkn,GPPV}, throughout the paper we
suppress\footnote{{\it cf.} \eqref{ZSQCD} where this factor is, in fact, present.}
the factor of $(q;q)_{\infty}$ (cf. \eq{def:qPochhammer}) in the physical index\footnote{For more general gauge groups, 
this relation involves a factor of $\eta (q)^{\text{rank} (G) }$ in the denominator~\cite{Gukov:2016gkn}.}
\be
\widehat{Z}^{(\text{unred})}_a (q) \; = \; \frac{\widehat{Z}_a (q)}{(q;q)_{\infty}}
\label{Zunnormalized}
\ee
and instead use $\widehat{Z}_a (q)$, which often takes a more compact form.
This physical version \eqref{Zunnormalized} of the index $\widehat{Z}_a (q)$
is sometimes called {\it unreduced} or {\it un-normalized}.

Taking into account this normalization, we can rephrase our discussion in
\S\ref{sec:resurgence}, in particular \eqref{def:Emb},
by saying that in theories where the normalized index
\be
\widehat{Z}_a (q) \; = \; \Psi_{p,s_1}(\tau) + \Psi_{p,s_2} (\tau) + \ldots
\ee
is given by a linear combination of false theta functions \eqref{def:partial2},
we have, {\it cf.} Table~\ref{tab:dict}:
\be
\widehat{Z}^{(\text{unred})}_a (q)\; = \;
\chi \left( M_{1,p-s_1} \oplus M_{1,p-s_2} \oplus \ldots;\tau\right)
\ee
In other words, the properly normalized physical index \eqref{Zunnormalized}
is equal to the character of a $(1,p)$ singlet VOA module
\be
M_{1,p-s_1} \oplus M_{1,p-s_2} \oplus \ldots
\label{MMmodule}
\ee
Note, although this module looks reducible, perhaps it indicates existence of
an extension to a larger log-VOA, where $\widehat{Z}^{(\text{unred})}_a (q)$
can be identified with a character of a less reducible module.
A positive indication for this comes from the fact that, in many of our examples,
every term $M_{1,s}$ is always accompanied by $M_{1,p-s}$ in \eqref{MMmodule}.

In particular, when composed with 3d-3d correspondence, this intriguing duality
between logarithmic CFTs and 3d $\N=2$ theories with half-BPS boundary conditions
implies that all Seifert manifolds with 3 singular fibers correspond to modules
of $(1,p)$ singlet model.
The modules are determined by the data of the Weil representation corresponding to $M_3$,
which, in turn,
can be obtained using the general technique outlined in \S\ref{sec:resurgence}.
It will be illustrated in many examples in \S\ref{sec:examples}.

\begin{table}[htb]
\centering
\renewcommand{\arraystretch}{1.3}
\begin{tabular}{|@{\quad}c@{\quad}|@{\quad}c@{\quad}|@{\quad}c@{\quad}| }
\hline  {\bf 3-manifold} & ~~~~$m+K$~~ & ~~{\bf module of a singlet log-VOA}
\\
\hline
\hline $\Sigma (2,3,5)$ & ~~~~$30+6,10,15$
& $M_{1,1} \oplus M_{1,11} \oplus M_{1,19} \oplus M_{1,29}$ \\
\hline $\Sigma (2,3,7)$ & ~~~~$42+6,14,21$
& $M_{1,1} \oplus M_{1,13} \oplus M_{1,29} \oplus M_{1,41}$ \\
\hline
\end{tabular}
\caption{Weil representations and the corresponding modules
of the logarithmic $(1,p)$ singlet CFT for simple homology spheres.
The sum over modules is precisely the sum over elements
of the orthogonal group $O_m$ introduced above \eqref{Om_action}.}
\label{tab:logCFTmodules}
\end{table}

It would be interesting to study a relation between logarithmic VOAs assigned to
3-manifolds here, and vertex algebras $\text{VOA} [M_4]$ assigned to 4-manifolds
bounded by such 3-manifolds via
the duality \cite{Gadde:2013sca,Dedushenko:2017tdw,Feigin:2018bkf}.
Another natural question is: For which class of 3d $\N=2$ theories
(and boundary conditions $\mathcal{B}_a$) the combined 3d-2d half-indices \eqref{BlockD2S1}
produce characters of logarithmic CFTs?
And, conversely, which logarithmic CFTs arise in this correspondence?
We hope to explore these questions in the future work.

Here and in \S\ref{sec:MTC}, we found several connections relating 3-manifold invariants $\widehat{Z}_a (q)$
with logarithmic CFTs and non-semisimple MTCs.
On the other hand, in a parallel line of development, ``logarithmic'' 3-manifold invariants based on non-semisimple MTCs
were studied in~\cite{logBCGPMa,logBCGPMb,logBBG} which, therefore, we expect to be related to $\widehat{Z}_a (q)$.
We plan to pursue this direction in the future work.

\subsection{Hyperbolic $M_3$ and non-$C_2$-cofinite log-VOAs}

Already at this early stage, the connections between 3-manifolds and logarithmic CFTs can teach us a valuable lesson.
Namely, they can help us understand the answer to the following important question:
What is it about 3-manifolds whose invariants $\widehat{Z}_a (q)$ can be expressed
in terms of false theta functions and mock modular forms,
as opposed to more complicated modular objects?

If we combine several clues from the above, the answer seems to be triggered by whether
the corresponding log-VOA is $C_2$-cofinite or not, and whether a 3-manifold $M_3$
admits only $G_{\CC}$ flat connections with rational values of $\text{CS} (\alpha)$,
\be
\text{CS} (\alpha) \; \in \; \mathbb{Q}
\qquad\qquad \text{for all~}\alpha \in \mathcal{M}_{\text{flat}} \left( G_{\CC} , M_3\right).
\label{hyperbtest}
\ee
Indeed, anticipating a close relation between $\text{MTC} [M_3]$ described in \S\ref{sec:mod}
and the tensor category of a log-VOA associated to $M_3$ via the dictionary summarized in Table~\ref{tab:dict},
we expect that modules of the latter have conformal dimensions $\Delta_{\alpha}$ related to
values of $\text{CS} (\alpha)$ as in \eqref{CSconfdim}.

Then, if condition \eqref{hyperbtest} fails for some $\alpha$,
it means that the corresponding logarithmic CFT must have at least some representations
with irrational conformal dimensions $\Delta_{\alpha}$,
and such vertex algebras can not be $C_2$-cofinite\footnote{Among other things,
the $C_2$-cofiniteness means that VOA has finitely many inequivalent irreducible modules~\cite{MR1317233}.
See also \cite{MR2875849} for a nice exposition and various ways to understand this condition.}.
Indeed, Miyamoto proved \cite{Miyamoto:2002ar} (see \cite{Creutzig:2016fms} for a lucid review)
that values of conformal dimensions and the central charge in a $C_2$-cofinite VOA must all be rational.
Curiously, the condition \eqref{hyperbtest} holds for all examples of 3-manifolds considered in this paper,
which is probably why in all cases we find a relation to $C_2$-cofinite log-VOAs.

On the other hand, hyperbolic 3-manifolds have at least one $\SL(2,\CC)$ flat
connection $\alpha_{\text{geom}}$ --- sometimes called ``geometric'' or ``hyperbolic'' ---
and its complex conjugate, such that $\text{Im} \, \text{CS} (\alpha_{\text{geom}}) \ne 0$.
This necessarily violates the condition \eqref{hyperbtest} and, based on the above considerations,
we expect hyperbolic 3-manifolds to be related to logarithmic vertex algebras which are {\it not} $C_2$-cofinite.
In particular, this suggests what one should expect of the $q$-series invariants $\widehat{Z}_a (M_3)$
for hyperbolic $M_3$, assuming the relation between 3-manifold invariants $\widehat{Z}_a (M_3)$
and characters of logarithmic VOAs continues to hold in the hyperbolic case
as well.\footnote{If it indeed passes further tests, the condition \eqref{hyperbtest}
perhaps deserves the name ``$C_2$-cofiniteness for 3-manifolds.''}

\section{Examples}
\label{sec:examples}

In the first part of this section, we analyze the definition of the homological blocks provided in \cite{GPPV} for plumbed 3-manifolds and show that their convergence only depends on the sign of the diagonal entries of $M^{-1}$ corresponding to high-valency vertices (vertices with more than two edges incident to them, $\deg v > 2$).

This enables us to extend the definition of the $q$-series invariants $\widehat{Z}_b (q)$ to a wider range of plumbed 3-manifolds, including those with indefinite plumbings  related  to the negative-definite ones via Kirby moves.
For positive-definite plumbings and their Kirby-equivalents, a new procedure is proposed in \S\ref{sec:otherside} to define the corresponding $q$-series invariants $\widehat{Z}_b (q)$.

In the second part of this section, we explicitly compute the new invariants $\widehat{Z}_b (q)$ for some examples of Seifert manifolds with three singular fibers. In addition, we examine the properties of these manifolds through the modular perspective outlined in \S\ref{sec:resurgence}. In particular, we provide asymptotic expansions of WRT invariants (or equivalently, the transseries expansions of Chern-Simons partition functions as in \S\ref{sec:resurgence}) for selected examples:
$$Z_{CS}(M_3) \; \sim \; \displaystyle \sum_{\alpha} e^{2 \pi i k \text{CS}(\alpha)} Z_{\text{pert}}^{(\alpha)} (k)$$
where $\alpha$ runs over {\it all} flat connections on $M_3$. In the above formula, $Z_{\text{pert}}^{(\alpha)}$ will be referred to as the {\em transseries} of the saddle point $\alpha$ ({\it i.e.} flat connection $\alpha$).

\subsection{Definite and indefinite plumbings}
\label{sec:indef}

Given a plumbing graph with framing coefficients $a_i \in \mathbb{Z}$,
there is an associated surgery link, see Figure~\ref{fig:plumb}.
Performing surgery along the link, we obtain a ``plumbed'' manifold $M_3$ \cite{NR}. In particular, all Seifert manifolds $M_3 = M(b;\{q_i/p_i\}_i)$ are plumbed manifolds, because the rational surgery coefficients can be realized by a series of 3d Kirby moves and continued fraction:\footnote{The orientation convention is such that Poincare homology sphere is represented by a $-E_8$ plumbing graph, \textit{i.e.}, $M(-2;\frac{1}{2},\frac{2}{3},\frac{4}{5})$.}
\begin{equation}
\begin{array}{c} \overset{\displaystyle{-p/q}}{\bullet} \end{array} =
\begin{array}{cccc}
\overset{\displaystyle{a_1}}{\bullet}
\frac{\phantom{xxx}}{\phantom{xxx}}
& \overset{\displaystyle{a_2}}{\bullet}
\frac{\phantom{xxx}}{\phantom{xxx}}
& \overset{\displaystyle{a_3}}{\bullet}
& \cdots
\end{array}
\ \text{where} \quad \frac{q}{p} = - \dfrac{1}{a_1 - \dfrac{1}{a_2 - \dfrac{1}{a_3 - \cdots}}}.
\label{cfractsurg}
\end{equation}

\begin{figure}[htb] \centering
\includegraphics{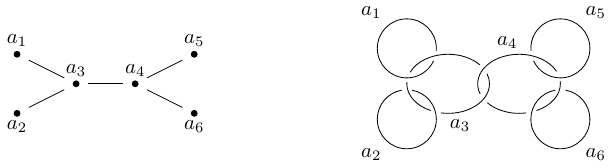}
\caption{A plumbing graph (left) and the associated surgery link (right).}
\label{fig:plumb}
\end{figure}

\noindent
Any Seifert manifold $M_3 = M(b;\{q_i/p_i\}_i)$ has a plumbing presentation illustrated in Figure~\ref{fig:SeifPlum}.
Such a plumbing graph has only one high-valency vertex and the rational surgeries along fibers are realized
by continued fractions, as in \eqref{cfractsurg}.
\begin{figure}[htb]
\centering
\includegraphics{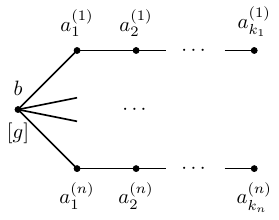}
\caption{\label{fig:SeifPlum} \small Plumbing graph for a Seifert manifold $M\big(b,g;\{\frac{q_i}{p_i}\}_{i=1}^n\big)$.}
\end{figure}

\noindent
A tree-shaped graph with $L$ vertices has a $L \times L$ adjacency matrix:
$$M_{ij} = \begin{cases} a_i &\text{if} \quad i = j \\ 1 &\text{if} \quad (i,j) \in \text{Edges} \\ 0 &\text{otherwise,} \end{cases}$$
which is precisely the linking form in \eq{eqn:gentheta}. Together with \eq{Fxexamples}, and \eq{eqn:gentheta}, we can compute the half-index \eq{Zalocalization} in the following form:\footnote{For convenience, we have chosen $M$ to be negative-definite. The condition can be relaxed, which is an interesting topic from the viewpoint of ``going to the other side.'' We will come back to this in Section~\ref{sec:otherside}. Also, to avoid clutter, we write \eqref{eqn:homblock} for $g=0$ Seifert manifolds; a more general expression for arbitrary genus $g$ involves the combination of \eqref{Sungnitpick} and \eqref{Fxexamples} as the integrand.}
\begin{multline}\label{eqn:homblock}
\widehat{Z}_b(q) = q^{-\frac{3L+\sum_v a_v}{4}} \cdot \mathrm{v.p.} \int_{|z_v|=1} \prod_{v \in \text{Vertices}} \frac{d z_v}{2 \pi i z_v} (z_v - 1/z_v)^{2 - \deg_v} \\
\times \sum_{\ell \in 2M\mathbb{Z}^L+b}q^{-\frac{(\ell,M^{-1}\ell)}{4}} \prod_{v\in \text{Vertices}} z_v^{\ell_v}.
\end{multline}
Here, $\mathrm{v.p.}$ indicates that we are performing a principal value integral, and $b \in 2 \mathrm{Coker}(M) + \delta$ modulo Weyl group action $b \leftrightarrow -b$. Although we have chosen $\delta \in \mathbb{Z}^L$ such that $\delta_v \equiv \deg_v \! \! \mod 2$, different choices of $\delta$ would only permute the homological blocks of a given plumbed manifold.

According to \cite{GPPV}, a particular combination of the homological blocks gives $SU(2)$ Chern-Simons partition function on $M_3$ in the radial limit $|q| \rightarrow 1$,
{\it cf.} \eq{relation_CS_blocks}. Specifically, for plumbed manifolds we have
\begin{equation}\label{eqn:ZCSdecomp}
\begin{gathered}
Z_{CS}(M_3) =\frac{1}{2i\sqrt{2k}}  \sum_{a} e^{2 \pi i k \text{CS}(a)} \sum_{b} S^{(A)}_{ab} \widehat{Z}_b(q), \\
a \in \mathrm{Coker}(M)/\mathbb{Z}_2 \ \overset{\text{setwise}}{=} \ \mathrm{Tor}H_1(M_3)/\mathbb{Z}_2,  \\
b \in (2\mathrm{Coker}(M)+\delta)/\mathbb{Z}_2, \quad \text{CS}(a) = - (a, M^{-1}a) \mod \mathbb{Z},\\
\text{and} \quad
S^{(A)}_{ab} = { \sum_{a' \in \{\ZZ_2\text{-orbit of }a\}} e^{2 \pi i (a',M^{-1}b)}  \over  \sqrt{|{\text{Tor}}H_1(M_3)|}}
\end{gathered}
\end{equation}
where the $\ZZ_2$ acts as the Weyl group on $H_1(M_3)$ by $a \leftrightarrow -a$.

When the plumbing graph is composed of a single high-valency vertex, \textit{i.e.}, when $M_3$ is a Seifert fibered manifold, the following theorem determines whether the homological blocks of $M_3$ defined via equation \eq{eqn:homblock} provide convergent $q$-series inside the unit disk and when they converge outside the unit disk.
\begin{lemma}
\label{thm:1-boundy}
Take $M_3$ to be a plumbed 3-manifold, whose plumbing graph $\mathcal{G}$ is a tree. Denote by $M$ the adjacency matrix of $\mathcal{G}$ and by $M^{-1}$ its inverse. Assume there is only one high-valency vertex and let $v_0$ denote the entry associated to this vertex in the adjacency matrix. Then, if $(M^{-1})_{v_0v_0}<0$ the homological blocks associated to $M_3$ are well-defined $q$-series, convergent for $|q| < 1$. On the other hand, if $(M^{-1})_{v_0v_0}>0$, the homological blocks converge for $|q|>1$.

More generally, when there are multiple high-valency vertices, let $\{v_i\}$ be the set of high-valency vertices in the plumbing graph of $M_3$. 
The homological blocks converge for $|q| < 1$ (respectively $|q| > 1$) when all 
 $(M^{-1})_{v_iv_i} <0$  (respectively $>0$) for all $v_i$'s. 
\end{lemma}
\begin{proof} To prove the above theorem we have to analyze the asymptotic growth of the formula for the homological blocks.
From equation~\eq{eqn:homblock} we have \cite{GPPV}
\be\label{eqn:zHatRegul}
\widehat{Z}_b(q)  = 2^{-L}q^{\Delta} \sum_{\ell \in 2M\mathbb{Z}^L+b}\, F_{1}^{\ell}\, q^{-\frac{(\ell,M^{-1}\ell)}{4}}, \quad b \in (2 \mathrm{Coker}M + \delta)/\mathbb{Z}_2
\ee
where the integer coefficients $F_{1}^{\ell}$ are generated as follows (note $2M\mathbb{Z}^L+b \subset 2\mathbb{Z}^L + \delta$):
\be\label{eqn:Fl}
\sum_{\ell \in 2\mathbb{Z}^L + \delta} F^\ell_1 \prod_v x_v^{\ell_v} =
\prod_v \left\{ \underset{\text{at} \ x_v \rightarrow 0}{\text{\scriptsize Expansion}} \frac{1}{(x_v - 1/x_v)^{\deg v-2}} +  \underset{\text{at} \ x_v \rightarrow \infty}{\text{\scriptsize Expansion}} \frac{1}{(x_v - 1/x_v)^{\deg v-2}} \right\}.
\ee
When $\deg v \leq 2$, the $x_v$-expansion on the RHS terminates at a finite order. Thus, $F^{\ell}_1$ vanishes when $|\ell_v|$ is large enough for all but one coordinates $\ell_v$ of $L$-dimensional vectors $\ell \in 2\mathbb{Z}^L + \delta$.

The only exception is $\ell_{v_0}$ which corresponds to the unique high-valency vertex $v_0$ ($\deg v_0 > 2$.) Explicitly,
$$F^{\ell}_1 \neq 0 \quad \Leftrightarrow \quad \ell_v = \begin{cases} \ell_{v_0} &\text{if} \ v = v_0, \ \ell_{v_0} \in \mathbb{Z} \\ 0 &\text{if} \ \deg v = 2 \\ 1 &\text{if} \ \deg v = 1. \end{cases} $$
Degree-zero vertices are irrelevant as we only consider connected graphs. This implies that the $q$-exponents in the RHS of \eq{eqn:zHatRegul} have the following  behavior: 
\be
q^{-\frac{(\ell,M^{-1}\ell)}{4} } = 
q^{-\frac{(M^{-1})_{v_0v_0}(\ell_{v_0})^2}{4} +O(1) } 
\ee
as $|\ell| \to \infty$ and if $F^{\ell}_1\neq 0$. This completes the proof for the first part. The proof proceeds in an identical way to the plumbing graphs with multiple high-valency vertices. 
\end{proof}

The above result shows that the validity of \eq{eqn:homblock} depends solely on the $M^{-1}$ entries at high-valency vertices.
Let us make a few remarks:
\begin{itemize}
\item \textit{orientation reversal.} It is important to note that the homological blocks in Lemma~\ref{thm:1-boundy} are computed from \eq{eqn:zHatRegul}, which is a result of a particular regularization (see Appendix A of \cite{GPPV}).
Therefore, when the formula does not define a convergent $q$-series inside the unit disc, an alternative computational scheme is required. In particular, given a 3-manifold $M_3 ({\mathcal{G}})$, the oppositely oriented manifold $- M_3 (\mathcal{G})$ provides the natural companion of $M_3 ({\mathcal{G}})$ on the other side of the $q$-plane. If the former has $(M^{-1})_{v_0v_0}<0$, the latter has $(M^{-1})_{v_0v_0}>0$ and \eq{eqn:zHatRegul} cannot be implemented to reproduce the associated homological blocks which are convergent for $|q|<1$. Therefore, when $(M^{-1})_{v_0v_0}>0$ we need to extend the definition of homological blocks outside the unit disc. This will be the central topic of \S\ref{sec:otherside}, where we conjecture a new procedure to derive the homological blocks of three manifolds with $(M^{-1})_{v_0v_0}>0$.

\item \textit{Kirby moves.} The signature of the plumbing data may not be invariant under 3d Kirby moves. An example is illustrated in \eq{eqn:indefPlumb} where two homeomorphic manifolds have different signatures: the LHS is neither positive nor negative definite, while the RHS is positive-definite.
\begin{equation}
\label{eqn:indefPlumb}
\begin{aligned}
\begin{array}{ccc}
& \overset{\displaystyle{-2}}{\bullet} &  \\
& \vline &   \\
\overset{\displaystyle{-2}}{\bullet}
\frac{\phantom{xxx}}{\phantom{xxx}}
& \underset{\displaystyle{-1}}{\bullet} &
\frac{\phantom{xxx}}{\phantom{xxx}}
\overset{\displaystyle{-3}}{\bullet}
\end{array} \quad
&\cong \quad
\begin{array}{c}
\text{orientation} \\
\text{reversal of}
\end{array}
\left(
\begin{array}{cccc}
& \overset{\displaystyle{-2}}{\bullet} & & \\
& \vline & &  \\
\overset{\displaystyle{-2}}{\bullet}
\frac{\phantom{xxx}}{\phantom{xxx}}
& \underset{\displaystyle{-2}}{\bullet} &
\frac{\phantom{xxx}}{\phantom{xxx}}
\overset{\displaystyle{-2}}{\bullet} &
\frac{\phantom{xxx}}{\phantom{xxx}}
\overset{\displaystyle{-2}}{\bullet}
\end{array} \right) \\[2ex]
M(-1;\tfrac{1}{2},\tfrac{1}{2},\tfrac{1}{3}) \quad &\cong \quad -M(-2;\tfrac{1}{2},\tfrac{1}{2},\tfrac{2}{3})
\end{aligned}
\end{equation}
Without Lemma~\ref{thm:1-boundy}, it is necessary to find for each indefinite manifold a homeomorphic, definite manifold to determine the convergence of homological blocks. In the above example, for instance, the RHS is positive-definite, so the homological blocks of $M(-1;\tfrac{1}{2},\tfrac{1}{2},\tfrac{1}{3})$ would converge \textit{outside} the unit disc. By the lemma, however, the domain of convergence can be immediately read off from $M^{-1}$ of the LHS (and of course, they converge for $|q| > 1$). It is also easy to see how 3d Kirby moves preserve the domain of convergence, as provided in Appendix~\ref{appendix:3dkirby}.
\item \textit{multiple high-valency vertices.} Lemma~\ref{thm:1-boundy} states that \eq{eqn:zHatRegul} does not reproduce convergent $q$-series when there appears multiple high-valency vertices whose $M^{-1}$ entries have different signs. This implies that homological blocks \eq{eqn:homblock} must be computed by other means than the regularization scheme \eq{eqn:zHatRegul}.
We will return to these examples in future work.
\end{itemize}

\subsection{Example: $M(-1;\frac{1}{2},\frac{1}{3},\frac{1}{9})$}

We first demonstrate the modularity dictionary and the steps outlined in Figure~\ref{fig:reverse} with the specific example of a Seifert manifold $M(-1;\frac{1}{2},\frac{1}{3},\frac{1}{9})$.

\subsubsection{$q$-series invariants}

The Seifert manifold has $\mathrm{Tor} H_{1}(M(-1;\frac{1}{2},\frac{1}{3},\frac{1}{9})) = \mathbb{Z}_{3}$ and the following plumbing graph:
\begin{equation}
\begin{array}{ccc}
& \overset{\displaystyle{-3}}{\bullet} & \\
& \vline & \\
\overset{\displaystyle{-2}}{\bullet}
\frac{\phantom{xxx}}{\phantom{xxx}}
& \underset{\displaystyle{-1}}{\bullet} &
\frac{\phantom{xxx}}{\phantom{xxx}}
\overset{\displaystyle{-9}}{\bullet}
\end{array}
\end{equation}
To compute its $q$-series invariants $\widehat Z_a (q)$, we first write down its adjacency matrix:
$$M = \begin{pmatrix} -1 & 1 & 1 & 1  \\ 1 & -2 & 0 & 0 \\ 1 & 0 & -3 & 0  \\ 1 & 0 & 0 & -9   \end{pmatrix}.$$
As is well known (see {\it e.g.} \cite{GS}), the cokernel of $M$ is isomorphic to $\mathrm{Tor} H_{1}(M(-1;\frac{1}{2},\frac{1}{3},\frac{1}{9}))$:
\begin{multline}
\mathrm{Coker}(M) = \mathbb{Z}^4/M\mathbb{Z}^4 = \big\langle (0,0,0,0), (1,0,-1,-6),(1,0,-2,-3) \big\rangle \\
 \cong \mathrm{Tor} H_{1}(M(-1;\tfrac{1}{2},\tfrac{1}{3},\tfrac{1}{9})) = \mathbb{Z}_3.
\end{multline}
The Weyl group action maps a cokernel element to its sign inverse. Therefore, the first element, $(0,0,0,0)$, is mapped to itself, while the others are conjugate to each other, \textit{i.e.}, $(1,0,-1,-6) = -(1,0,-2,-3) \in \mathbb{Z}^4/M\mathbb{Z}^4.$ Thus,
\begin{gather}
\mathrm{Coker}(M)/\mathbb{Z}_2 = \big\langle (0,0,0,0), (1,0,-1,-6) \big\rangle  \label{eqn:a}\\
(2\mathrm{Coker}(M)+\delta)/\mathbb{Z}_2 = \big\langle (1,-1,-1,-1), (3,-1,-3,-13) \big\rangle \label{eqn:b}
\end{gather}
where $\delta = (1,-1,-1,-1)$ is given by $\delta_v =\deg_v -2$, as in \cite{GPPV}.
\iffalse
Now, we are ready to compute homological blocks. Recall that for plumbed manifolds, homological blocks are given by:
\begin{multline}
\widehat{Z}_b(q) = q^{-\frac{3L+\sum_v a_v}{4}} \cdot \mathrm{v.p.} \int_{|z_v|=1} \prod_{v \in \text{Vertices}} \frac{d z_v}{2 \pi i z_v} (z_v - 1/z_v)^{2 - \deg_v} \Theta^{-M}_b(z), \\
\text{where} \quad \Theta^{-M}_b(x) = \sum_{l \in 2M\mathbb{Z}^L+b} q^{-\frac{(l,M^{-1}l)}{4}} \prod_{i =1}^L x_i^{l_i}.
\end{multline}
Above, $\mathrm{v.p.}$ indicates that we are performing a principal value integral, and $b \in 2 \mathrm{Coker}(M) + \delta$ modulo Weyl group action.
\fi
Then, the $q$-series invariants $\widehat{Z}_b (M_3)$ are given by \eq{eqn:homblock}:
\begin{align}
 &\widehat{Z}_{(1,-1,-1,-1)}(q) = q+q^5-q^6-q^{18}+q^{20}+\ldots \label{eqn:Zhat1}\\
 &\widehat{Z}_{(3,-1,-3,-13)}(q) = -q^{4/3}(1+q^2-q^7-q^{13}+q^{23}+\ldots) \label{eqn:Zhat2}.
\end{align}

\subsubsection{Weil representation: 18+9}

To homological blocks of $M_3$ one can associate a Weil representation, labelled by the pair $m$ and $K$. Explicitly, they are related via \eq{def:Emb}.
Let us first determine $m$ via modularity dictionary.
Recall, that the non-perturbative part of the transseries \eq{trans_false} for $\Psi^{m+K}_r$, of the form $\sim e^{- 2 \pi i k (r')^2/4m}$, should capture the contributions from non-abelian flat connections.
For a Seifert manifold $M(b, \{ q_i/p_i \}_{i=1}^n)$, the denominator of $\text{CS}(a)$ for $a$ non-abelian is a l.c.m. of $4 p_i$, where $p_i$ are the orders of singular fibers in the Seifert invariant \cite{Auckly}. As a result, we claim that for $M_3= M(b, \{ q_i/p_i \}_{i=1}^n)$ we have
\be
4m = \mathrm{l.c.m.}\Big( 4\{ p_i \}_{i=1}^n \; \cup \; \{ \text{Denominators of $\text{CS}(a)$} \}_{0 \neq a \in \mathrm{Coker}M/\mathbb{Z}_2} \Big).
\label{eqn:mformula}
\ee

For the current example, we can easily compute $\text{CS}(a)$ for abelian flat connections from the cokernel elements computed in \ref{eqn:a}:
$$\text{CS}(a) = -(a,M^{-1}a) = \begin{cases} 0 \quad \mod \quad \mathbb{Z} & \text{for} \quad a = (0,0,0,0) \\ \frac{1}{3} \quad \mod \quad \mathbb{Z} & \text{for} \quad a = (1,0,-1,-6).  \end{cases}$$
Combining (\ref{eqn:mformula}) and the $\text{CS}(a)$ computed, we conclude:
$$4m = \mathrm{l.c.m.}(8,12,36,3) =72 \quad \Rightarrow \quad m = 18.$$
Correspondingly, the possible $K$ giving rise to irreducible representations are $K=\{1,2\}$ and $K=\{1,9\}$. A simple calculation reveals that the relevant irreducible representation is that labelled by $m+K=18+9$.

In summary, we have
\begin{equation}\label{summary189}
\begin{aligned}
&\sigma^{18+9} = \{1,3,5,7\}\\
&\widehat{Z}_{(1,-1,-1,-1)}(q) = q^{71/72}\Psi^{18+9}_1(\tau)  \\
&\widehat{Z}_{(3,-1,-3,-13)}(q) = -q^{71/72}\Psi^{18+9}_5(\tau).\\
\end{aligned}
\end{equation}

Next, we proceed to compute the composite matrices $S(M_3)$ and $T(M_3)$, defined in \eq{M3matrices}.

\subsubsection{Computing $S(M_3)$ and $T(M_3)$}

Let us write down all the relevant matrices for the current example.
First, recall that $S^{(A)}$ is the linking pairing on $\mathrm{Tor}H_1(M_3)$ in \eq{eqn:ZCSdecomp}. For $M_3 = M(-1;\frac{1}{2},\frac{1}{3},\frac{1}{9})$,
\begin{equation}
S^{(A)} = \frac{1}{\sqrt{3}}\begin{pmatrix} 1 & 1 \\ 2 & -1 \end{pmatrix}.
\end{equation}
Next, from \eq{def:Emb} and \eq{summary189} we can easily read off:
\begin{equation}
\text{\textbf{Emb}} = \begin{pmatrix} 1 & 0 & 0 & 0 \\ 0 & 0 & -1 & 0 \end{pmatrix}.
\end{equation}
The S-matrix of the Weil representation is easily computed from to be:
\begin{equation}
S^{(B)} = -\frac{2i}{3}
\begin{pmatrix}
A& \tfrac{3}{2}&B&C\\ \tfrac{1}{2}&0&\tfrac{1}{2}&-\tfrac{1 }{2}\\
B&\tfrac{3}{2}&-C&-A \\ C&-\tfrac{3}{2}&-A&B
\end{pmatrix}
\end{equation}
where $A,B,C = \sin(\tfrac{\p}{18}),\sin(\tfrac{5\p}{18}),\sin(\tfrac{7\p}{18})$ respectively.

Finally we combine $S^{(A)}, \text{\textbf{Emb}}$ and $S^{(B)}$ into $S(M_3)$:
\begin{equation}
S(M_3) = \begin{pmatrix} -0.23 i & 0 & 0.66 i & 0.43i \\ 0.43 i & 1.73i & 0.23 i & 0.66i \end{pmatrix},
\end{equation}
here evaluated numerically and rounded to the second decimal place.

Next, we compute the $T$ matrices. $T^{(A)}$ is the diagonal matrix with $e^{2 \pi i \text{CS}(a)}$ on the diagonal:
\begin{equation}
T^{(A)} = \exp \; 2\pi i \begin{pmatrix} 0 & 0 \\ 0 &  \frac{1}{3} \end{pmatrix}
.\end{equation}
{}From \eq{Tmatrix}, we also have
\begin{equation}
T^{(B)}=
\exp \; 2\pi i \begin{pmatrix} {1\over 72} & 0 &0 & 0 \\  0 &  {9\over 72}  &0 & 0 \\
 0 &0&  {25\over 72}   & 0 \\ 0 & 0 &0 &  {49\over 72}  \end{pmatrix}
\end{equation}
Combining all these elements, we obtain
\be\renewcommand*{\arraystretch}{1.5}
T(M_3) = \bem \ex(-\frac{1}{72}) & \ex(-\frac{9}{72}) &\ex(-\frac{25}{72}) & \ex(-\frac{49}{72}) \\ \ex(-\frac{49}{72}) & \ex(-\frac{57}{72})& \ex(-\frac{1}{72}) & \ex(-\frac{25}{72}) \eem.
\ee

\subsubsection{Non-abelian flat connections}

As advertised, we will now extract from $S(M_3)$ and $T(M_3)$ the set of Chern-Simons invariants for all non-abelian flat connections on $M_3$ and determine which of them are complex.

From $S(M_3)$ computed above, we observe that:
\begin{equation}
\{T(M_3)_{ar} | \text{$a,r$ such that $S(M_3) \neq 0$}\} = \{ \ex(-\tfrac{1}{72}), \ex(-\tfrac{25}{72}), \ex(-\tfrac{49}{72}), \ex(-\tfrac{57}{72})
 \}.
\end{equation}
From the rule \eq{rule_nonAb},
it follows that there are (at least) four non-abelian $\SL(2,\CC)$ flat connections, and the set of their Chern-Simons invariants is $\{ -\frac{1}{72},-\frac{25}{72},-\frac{49}{72},-\frac{57}{72} \}$ modulo $\ZZ$.

To determine which of them correspond to complex non-abelian flat connections, the next step is to compute $\mathbf{c}_\alpha$
via \eq{def:calpha}, which involves a sum over the pairs $(a,r)$ for which $T(M_3)_{a,r}=\ex(\alpha)$.
For example, when $\alpha = -\frac{1}{72}$, $(a,r) = (1,1)$ and $(2,4)$. Now, we can compute $\mathbf{c}_\alpha$:
\begin{equation}
\begin{cases}
\mathbf{c}_{-\frac{1}{72}} &=0\\
\mathbf{c}_{-\frac{25}{72}} &=1.17 i\\
\mathbf{c}_{-\frac{49}{72}} &= 0.76 i\\
\mathbf{c}_{-\frac{57}{72}} &= 1.03i .
\end{cases}
\label{eqn:calpha}
\end{equation}
So we conclude that $M_3 = M(-1;\frac{1}{2},\frac{1}{3},\frac{1}{9})$ must admit one complex non-abelian flat connection with $\text{CS} = -\frac{1}{72}$, and three $SU(2)$ non-abelian flat connections with $\text{CS} = -\frac{25}{72},-\frac{49}{72},-\frac{57}{72}.$

\subsubsection{Counting by A-polynomial}

Let us compare the above results with the computation based on a surgery presentation of $M_3$.
As explained in \cite{Gukov:2003na} and \cite[sec.5]{Gukov:2016njj},
when $M_3 = S^3_{r} (K)$ is a surgery on a knot $K \subset S^3$ with a surgery coefficient $r \in \QQ$,
flat $\SL(2,\CC)$ connections on $M_3$ are contained in the set of intersection points:
\be
\text{flat connections} \quad \hookrightarrow \qquad \{ s(x,y) := y x^r - 1 = 0 \} \,\cap\, \{ A_K (x,y)=0 \}
\label{intAsurg}
\ee
in $(\CC^* \times \CC^*) / \ZZ_2$ parametrized by $(x,y) \sim (x^{-1},y^{-1})$.
Here, $A_K (x,y)$ is the so-called A-polynomial of the knot $K$.
Note, some of the intersection points \eqref{intAsurg} may not lift to an actual representation $\pi_1 \to \SL(2,\CC)$.
Similarly, one might worry that accidental cancellations in the steps outlined in Figure \ref{fig:reverse}
could cause one to underestimate the number of flat connections on $M_3$.
Therefore, in practice, it is a good idea to compare the results produced by these two methods, when both are available.

In our present example of $M_3 = M(-1;\frac{1}{2},\frac{1}{3},\frac{1}{9})$ such an alternative method is indeed available,
thanks to a surgery presentation $M_3 = S^3_{-3} ({\bf 3}^{r})$, where $K = {\bf 3}^{r}$ is the right-handed trefoil knot.
The corresponding A-polynomial and the curve $s(x,y)=0$ are:
$$A(x,y) = (y-1)(y x^6 + 1), \quad s(x,y) = y x^{-3}-1.$$
Discarding the point $(x,y)=(-1,-1)$ that does not lift to a flat connection on $M_3$ \cite{Gukov:2016njj},
we obtain the following intersection points \eqref{intAsurg}, modulo the symmetry $(x,y)\sim(x^{-1},y^{-1})$:
$$ (x,y) \; = \; (1,1) \,, (e^{2\pi i /3},1) \,, (e^{\pi i /3},-1) \,, (e^{\pi i \frac{1}{9}},e^{\pi i /3}) \,, (-e^{\pi i \frac{4}{9}},e^{\pi i /3}) \,, (e^{\pi i \frac{7}{9}},e^{\pi i /3}) \,.$$
All abelian flat connections have $y=1$, and there are two such points in our list, in agreement with the above analysis.
The remaining four points are candidates for non-abelian flat connections, either real or complex.
Since the modularity analysis leads to the lower bound on the number of non-abelian flat connections also equal to 4 in this example,
combining the upper and lower bounds produced by these two methods we learn that the total number of non-abelian flat connections
indeed must be 4.

\subsubsection{Asymptotic expansions}

We conclude the analysis of this example by writing the asymptotic expansion of $Z_{CS}(M_3)$.
Combining the relation between the $q$-series invariants and \eq{relation_CS_blocks}
with the transseries expression for the false theta functions \eq{trans_false},
we obtain the transseries expressions at large $k$ for $Z_{CS}(M_3)$.
The results for various saddle points (flat connections on $M_3$)
are tabulated in Table~\ref{table:239}, where we omitted the overall factor $-iq^{71/72}/2\sqrt{2}$.

\begin{table}[h]\begin{center}\scalebox{0.9}{
\begin{tabular}{c c c c}
CS action & stabilizer & type & transseries \vspace{3pt}
\\ \toprule
$0$ & $SU(2)$ & central & $e^{2 \pi i k \cdot 0} \bigg( \frac{4 \pi i}{3 \sqrt{3}} k^{-3/2} + \frac{203 \pi^2}{27\sqrt{3}} k^{-5/2} + \mathcal{O}(k^{-7/2}) \bigg)$  \\[1.5ex]
$\frac{1}{3}$ & $U(1)$ & abelian & $e^{2 \pi i k \frac{1}{3}} \bigg( \sqrt{3}k^{-1/2} - \frac{11 \pi i}{4\sqrt{3}} k^{-3/2} + \mathcal{O}(k^{-5/2}) \bigg)$ \\[1.5ex]

$-\frac{25}{72}$ & $\pm1$ & non-abelian, real & $e^{-2\pi ik\frac{25}{72}} e^{\frac{3 \pi i}{4}}\bigg[ \frac{4}{3\sqrt{3}} (\cos\frac{2\pi}{9}+2\sin\frac{\pi}{18}) \bigg] $ \\[1.5ex]
$-\frac{49}{72}$ & $\pm1$ & non-abelian, real & $e^{-2\pi ik\frac{49}{72}} e^{\frac{3 \pi i}{4}} \bigg[ \frac{4}{3\sqrt{3}} (2\cos\frac{\pi}{9} + \sin\frac{\pi}{18}) \bigg] $ \\[1.5ex]
$-\frac{57}{72}$ & $\pm1$ & non-abelian, real & $e^{-2\pi ik\frac{57}{72}} e^{\frac{3 \pi i}{4}} \frac{2}{\sqrt{3}} $ \\[1.5ex]
$-\frac{1}{72}$ & $\pm1$ & non-abelian, complex & $0$ \vspace{2pt}\\ \bottomrule
\end{tabular}}
\caption{Transseries and classification of flat connections on $M(-2;\frac{1}{2},\frac{1}{3},\frac{1}{9})$. }
\label{table:239}
\end{center}
\end{table}

\subsection{Example: $M(-2;\frac{1}{2},\frac{1}{3},\frac{1}{2})$}

Let us look at one more example in detail, the Seifert manifold $M_3=M(-2;\frac{1}{2},\frac{1}{3},\frac{1}{2})$.
This example will also play a role in \S\ref{sec:otherside},
where the extension of $q$-series invariants $\widehat{Z}_a (q)$ to the lower-half plane is discussed.

Another new feature of this example is a ``center symmetry,'' a global $\mathbb{Z}_2$-symmetry distinct from the familiar Weyl group action.
We call it ``center symmetry'' because it acts on representations $\rho: \pi_1 (M_3) \to \SL(2,\CC)$
by multiplying some of the corresponding holonomies by the central elements $\pm {\bf 1}$ of $G=SU(2)$
or its complexification $G_{\CC} = \SL(2,\CC)$.
The role of this center symmetry will be discussed in details toward the end of this example.

\subsubsection{$q$-series invariants}

The manifold of interest has $\mathrm{Tor}H_1(M_3,\ZZ) = \mathbb{Z}_8$ and one of its plumbing
presentations looks like:
\begin{equation}
\begin{array}{ccc}
& \overset{\displaystyle{-3}}{\bullet} & \\
& \vline & \\
\overset{\displaystyle{-2}}{\bullet}
\frac{\phantom{xxx}}{\phantom{xxx}}
& \underset{\displaystyle{-2}}{\bullet} &
\frac{\phantom{xxx}}{\phantom{xxx}}
\overset{\displaystyle{-2}}{\bullet}
\end{array}
\end{equation}
Again, we write down the adjacency matrix and compute $a \in \mathrm{Coker}(M)$ and $b \in 2\mathrm{Coker}(M)+\delta$:
\begin{equation}
 M = \begin{pmatrix} -2 & 1 & 1 & 1  \\ 1 & -2 & 0 & 0 \\ 1 & 0 & -3 & 0  \\ 1 & 0 & 0 & -2   \end{pmatrix}
 \end{equation}
\begin{gather}\begin{split}
a \in \mathrm{Coker}(M)/\mathbb{Z}_2 =& \big\langle (0,0,0,0), (1,-1,0,-1), \\
&(0,-1,0,0), (0,0,-1,0), (0,0,0,-1) \big\rangle
\end{split}\end{gather}
\begin{gather}\begin{split}
b \in (2\mathrm{Coker}(M)+\delta)/\mathbb{Z}_2 = &\big\langle (3,-1,-5,-3), (3,-3,-5,-1), \\
&(1,-1,-1,-1), (3,-3,-1,-3), (1,-3,-1,-1) \big\rangle.
\end{split}\end{gather}
Using this input and the general tools described earlier,
we can now compute three kinds of topological invariants of $M_3=M(-2;\frac{1}{2},\frac{1}{3},\frac{1}{2})$:
1) the Chern-Simons invariants of abelian flat connections,
2) its $S^{(A)}$ matrix, and 3) its $q$-series invariants $\widehat{Z}_a (M_3)$:
\begin{equation}
\text{CS}(a) = -(a,M^{-1}a) = \begin{cases} 0 \quad \mod \quad \mathbb{Z} & \text{for} \quad a = (0,0,0,0), (1,-1,0,-1) \\ \frac{7}{8} \quad \mod \quad \mathbb{Z} & \text{for} \quad a = (0,-1,0,0), (0,0,0,-1) \\ \frac{1}{2} \quad \mod \quad \mathbb{Z} & \text{for} \quad a = (0,0,-1,0) \end{cases}.
\label{eqn:CS6+2}
\end{equation}
\begin{equation}\label{eqn:SA6plus2}
S^{(A)} = \frac{1}{\sqrt{8}} \begin{pmatrix} 1 & 1 & 1 & 1 & 1 \\ 1 & 1 & 1 & 1 & 1 \\ 2 & 2 & 0 & 0 & -2 \\ 2 & 2 & -2 & -2 & 2 \\ 2 & 2 & 0 & 0 & -2 \end{pmatrix}
 \end{equation}
 \be\label{eqn:Zhat6+2}
\begin{aligned}
&\widehat{Z}_{(3,-1,-5,-3)}(q) = q^{-1/4}(-1+q^4-q^8+q^{20}-q^{28}+q^{48}+\ldots) \\
&\widehat{Z}_{(3,-3,-5,-1)}(q) = q^{-1/4}(-1+q^4-q^8+q^{20}-q^{28}+q^{48}+\ldots) \\
&\widehat{Z}_{(1,-1,-1,-1)}(q) = q^{-3/8}(1+q-q^2+q^5-q^7+q^{12}+\ldots) \\
&\widehat{Z}_{(3,-3,-1,-3)}(q) = q^{-3/8}(-1+q-q^2+q^5-q^7+q^{12}+\ldots) \\
&\widehat{Z}_{(1,-3,-1,-1)}(q) = 2q^{1/4}(1-q^2+q^{10}-q^{16}+q^{32}-q^{42}+\ldots)
\end{aligned}
\ee
Plugging these into \eq{eqn:mformula}, we obtain:
\begin{equation}
4m = \mathrm{l.c.m.} ( 8,12, 1,2,8 ) = 24 ~ \Rightarrow
~m=6.
\end{equation}
Since $\mathrm{Ex}_6 = \{ 1, 2, 3, 6 \}$, $K$ can be either $\{1\}$, $\{ 1, 2 \}$ or $\{1, 3\}$, with the latter two corresponding to irreducible representations. With $m+K = 6+2$, we get:
\begin{equation}
\begin{aligned}
&\sigma^{6+2} = \{1,2,4\}\\
& \widehat{Z}_{(3,-1,-5,-3)}(q) = \widehat{Z}_{(3,-3,-5,-1)}(q) =  -\frac{1}{2}q^{-5/12}\Psi^{6+2}_2(\tau) \\
& \widehat{Z}_{(1,-1,-1,-1)}(q) = q^{-5/12}(2q^{1/24}- \Psi^{6+2}_1(\tau) ) \\
& \widehat{Z}_{(3,-3,-1,-3)}(q) = -q^{-5/12}\Psi^{6+2}_1(\tau)  \\
& \widehat{Z}_{(1,-3,-1,-1)}(q) =  q^{-5/12}\Psi^{6+2}_4(\tau)
\end{aligned}
\label{eqn:zhat6+2}
\end{equation}

\subsubsection{Computing $S(M_3)$ and $T(M_3)$}

Next, we proceed to compute the ``composite'' modular matrices $S(M_3)$ and $T(M_3)$.
The matrix $S^{(A)}$ has already been computed in \eq{eqn:SA6plus2}.
The embedding matrix can be read off from (\ref{eqn:zhat6+2}):
\begin{equation}
\text{\textbf{Emb}} = \begin{pmatrix} 0 & -\frac{1}{2} & 0 \\ 0 & -\frac{1}{2} & 0 \\ -1 & 0 & 0 \\ -1 & 0 & 0 \\ 0 & 0 & 1 \end{pmatrix}.
\end{equation}
The matrix $S^{(B)}$ can be computed from the projection matrix to be
\begin{equation}
S^{(B)} = -\frac{i}{2}
\begin{pmatrix}
0 & 1 & 1 \\[1.5ex]
2 & 1 & -1 \\[1.5ex]
2 & -1 & 1
\end{pmatrix}.
\end{equation}
and, when combined with $\text{\textbf{Emb}}$ and $S^{(A)}$, gives
\begin{equation}
S(M_3) = \frac{i}{\sqrt{2}}\begin{pmatrix} 0 & 1 & 0 \\ 0 & 1 & 0 \\ 2 & 0 & 0 \\ 0 & 0 & -2 \\ 2 & 0 & 0 \end{pmatrix}.
\end{equation}
Next, we compute the $T$ matrices. From (\ref{eqn:CS6+2}) we obtain
\begin{equation}
T^{(A)} = \exp \; 2\pi i  \begin{pmatrix} 0 & 0 & 0 & 0 & 0 \\  0 & 0 & 0 & 0 & 0 \\ 0 & 0 & \frac{7}{8} & 0 & 0 \\ 0 & 0 & 0 & \frac{1}{2} & 0 \\ 0 & 0 & 0 & 0 & \frac{7}{8} \end{pmatrix} \,.
\end{equation}
On the other hand, $T^{(B)} = e^{2 \pi i \frac{r^2}{4m}}\delta_{r,r'}$ for $r \in \sigma^{6+2} = \{1,2,4\}$:
\begin{equation}
T^{(B)} = \exp \; 2 \pi i \begin{pmatrix} \frac{1^2}{24} & 0 & 0 \\ 0 & \frac{2^2}{24} & 0  \\ 0 & 0 &  \frac{4^2}{24}\end{pmatrix}.
\end{equation}
Combining these two $T$-matrices with $\mathbf{I}$ (= $3 \times 5$ matrix with all entries equal to $1$), we  get:
\be\renewcommand*{\arraystretch}{1.5}
T(M_3) = \bem \ex(-\frac{1}{24}) & \ex(-\frac{4}{24}) & \ex(-\frac{16}{24}) \\ \ex(-\frac{1}{24}) & \ex(-\frac{4}{24}) & \ex(-\frac{16}{24}) \\ \ex(-\frac{4}{24}) & \ex(-\frac{7}{24}) & \ex(-\frac{19}{24}) \\  \ex(-\frac{13}{24}) & \ex(-\frac{4}{24}) & \ex(-\frac{4}{24}) \\ \ex(-\frac{4}{24}) & \ex(-\frac{7}{24}) & \ex(-\frac{19}{24}) \eem. 
\ee

\subsubsection{Non-abelian flat connections}

From the $S(M_3)$ computed in the previous subsection, we observe that:
\begin{equation} \{T(M_3)_{ar} | \text{$a,r$ such that $S(M_3) \neq 0$}\} = \{ \ex(-\tfrac{1}{6}) \}.\end{equation}
Therefore, using the rule \eq{rule_nonAb}, we predict (at least) one non-abelian $SL (2, \mathbb{C})$
flat connection with Chern-Simons invariant $-\frac{4}{24}$.
To determine whether it corresponds to a complex non-abelian flat connection, we compute $\mathbf{c}_{-\frac{1}{6}}$ via \eq{def:calpha}:
\begin{equation}
\mathbf{c}_{-\frac{1}{6}} = 2i\sqrt{2} \neq 0
\end{equation}
So we predict one $SU(2)$ non-abelian flat connection
with $\text{CS} = -\frac{4}{24}$, and no complex flat connections on $M_3=M(-2;\frac{1}{2},\frac{1}{3},\frac{1}{2})$.

\subsubsection{Asymptotic expansions}

As usual, we can assemble $\widehat{Z}_b$ into $Z_{CS}(M_3)$ to obtain the transseries for $M(-2;\frac{1}{2},\frac{1}{3},\frac{1}{2})$, summarized in Table~\ref{table:6+2} (where we omit an overall factor $-iq^{-5/12}/2\sqrt{2}$).

\begin{table}[h]\begin{center}\scalebox{0.9}{
\begin{tabular}{c c c c}
CS action & stabilizer & type & transseries \vspace{3pt}
\\ \toprule
$0$ & $SU(2)$ & central & $e^{2 \pi i k \cdot 0} \bigg( \frac{\pi i}{4\sqrt{2}}k^{-3/2} + \frac{7 \pi^2}{96\sqrt{2}}k^{-5/2} + \mathcal{O}(k^{-7/2}) \bigg)$  \\[1.5ex]
$0$ & $SU(2)$ & central & $e^{2 \pi i k \cdot 0} \bigg(\frac{\pi i}{4\sqrt{2}}k^{-3/2} + \frac{7 \pi^2}{96\sqrt{2}}k^{-5/2}+ \mathcal{O}(k^{-7/2}) \bigg)$  \\[1.5ex]
$\frac{7}{8}$ & $U(1)$ & abelian & $e^{2 \pi i k \frac{7}{8}} \bigg( -\sqrt{2}k^{-1/2} + \frac{2\sqrt{2} \pi i}{3} k^{-3/2} + \mathcal{O}(k^{-5/2}) \bigg)$  \\[1.5ex]
$\frac{7}{8}$ & $U(1)$ & abelian & $e^{2 \pi i k \frac{7}{8}} \bigg(  -\sqrt{2}k^{-1/2} + \frac{2\sqrt{2} \pi i}{3} k^{-3/2} + \mathcal{O}(k^{-5/2}) \bigg)$  \\[1.5ex]
$\frac{1}{2}$ & $U(1)$ & abelian & $e^{2 \pi i k \frac{1}{2}} \bigg( -\frac{2\sqrt{2}}{3}k^{-1/2} - \frac{11 \pi i}{54\sqrt{2}} k^{-3/2}  + \mathcal{O}(k^{-5/2}) \bigg)$  \\[1.5ex]
$-\frac{4}{24}$ & $\pm1$ & non-abelian, real & $e^{-2\pi ik\frac{4}{24}} e^{\frac{3 \pi i}{4}} 2\sqrt{2} $ \vspace{2pt}\\ \bottomrule
\end{tabular}}
\caption{Transseries for $M(-2;\frac{1}{2},\frac{1}{3},\frac{1}{2})$. }
\label{table:6+2}
\end{center}
\end{table}

\subsubsection{Center symmetry}

Note that there is a degeneracy in \eq{eqn:CS6+2}--\eq{eqn:Zhat6+2} due to an extra symmetry, {\it e.g.} the values $\text{CS}(a)$ are equal for $a = (0,0,0,0)$ and $a=(1,-1,0,-1)$, and the corresponding rows of $S^{(A)}$ also enjoy the same symmetry. From \eq{relateZhatZ} we see that the asymptotic expansions around these two abelian flat connections are, in fact, identical. Indeed, Table~\ref{table:6+2} explicitly shows several identical transseries.

Since not only the values $\text{CS}(a)$ but also the perturbative expansions around the flat connections are identical, we claim that the center symmetry is a symmetry of the moduli space. In order to understand this origin of this symmetry and to remove the degeneracy from $S^{(A)}$, we first study its action on the holonomy representations and then match the false theta functions with the ``folded'' version of the data \eq{eqn:CS6+2}--\eq{eqn:Zhat6+2} obtained by modding out the center symmetry.

The fundamental group of a Seifert manifold  $M_3=M(b, \{ q_i/p_i \}_{i=1}^n)$ is given by
$$\pi_1(M_3) = \langle x_1, x_2, x_3, h \ | \ h \ \text{central}, \ x_i^{p_i} = h^{-q_i}, \ x_1 x_2 x_3 = h^b \rangle.$$
We can classify $SU(2)$ flat connections by the $SU(2)$ representations of the fundamental group into $SU(2)$:
$$\rho : \big( \pi_1(M_3) \longrightarrow SU(2) \big)/\text{conj.}$$
modulo gauge transformations. Concretely, we can characterize such representations by the images of $\pi_1(M_3)$ generators.
In our present example, they are given by (before modding out by gauge transformations):
\begin{gather}
\begin{split}
\rho(x_i)  &= g_i  \Big( \begin{smallmatrix} \ex(\lambda_i) & 0 \\ 0 &  \ex(-\lambda_i)  \end{smallmatrix} \Big) g_i^{-1}, i=1,2,3 \\
\rho(h) &= \Big( \begin{smallmatrix} \ex(\lambda) & 0 \\ 0 &  \ex(-\lambda)  \end{smallmatrix} \Big)
\end{split}
\end{gather}
where $g_i$ represent arbitrary gauge transformations that are compatible with the group structure of $\pi_1(M_3)$.
The Weyl group acts on each $\rho(x_i)$ via conjugation by $\left( \begin{smallmatrix} 0 & -1 \\ 1 & 0 \end{smallmatrix} \right)$, hence
$\l_i \leftrightarrow -\l_i$. In what follows we will identify holonomy variables $\l_i$ related by the action of the Weyl group, as they correspond to the same flat connection. Table~\ref{table:6+2hol} shows holonomy variables $(\l, \l_1, \l_2, \l_3)$ which classify the group homomorphisms $\rho$ and their Chern-Simons invariants computed as in \cite{Auckly}.

\begin{table}[htb]
\centering
\begin{tabular}{c c c c}
CS invariant & type & $(\l, \l_1, \l_2, \l_3)$ & center symmetry\\
\hline \\
$0$ & abelian & $ ( 0,0,0,0 ) $ & $( 0,0,0,0 ) \mapsto ( 0,\frac{1}{2},0,\frac{1}{2} )$  \\[1.5ex]
$0$ & abelian & $ ( 0,\frac{1}{2},0,\frac{1}{2} ) $ & $( 0,\frac{1}{2},0,\frac{1}{2} ) \mapsto (0,0,0,0) $\\[1.5ex]
$\tfrac{1}{2}$ & abelian & $ (\frac{1}{2},\frac{1}{4},\frac{1}{2},\frac{1}{4} ) $ & $(\frac{1}{2},\frac{1}{4},\frac{1}{2},\frac{1}{4} ) \mapsto (\frac{1}{2},\frac{1}{4},\frac{1}{2},\frac{1}{4} )$ \\[1.5ex]
$\tfrac{7}{8}$ & abelian & $ (\frac{1}{4},\frac{5}{8},\frac{1}{4},\frac{1}{8} ) $ & $(\frac{1}{4},\frac{5}{8},\frac{1}{4},\frac{1}{8} )  \mapsto (\frac{1}{4},\frac{1}{8},\frac{1}{4},\frac{5}{8} ) $\\[1.5ex]
$\tfrac{7}{8}$ & abelian & $ (\frac{1}{4},\frac{1}{8},\frac{1}{4},\frac{5}{8} ) $ & $(\frac{1}{4},\frac{1}{8},\frac{1}{4},\frac{5}{8} ) \mapsto (\frac{1}{4},\frac{5}{8},\frac{1}{4},\frac{1}{8} ) $\\[1.5ex]
$-\tfrac{4}{24}$ & non-abelian & $( \frac{1}{2},\frac{1}{4},\frac{1}{6},\frac{1}{4} )$  & $( \frac{1}{2},\frac{1}{4},\frac{1}{6},\frac{1}{4} ) \mapsto ( \frac{1}{2},\frac{1}{4},\frac{1}{6},\frac{1}{4} )$
\end{tabular}
\caption{Holonomy variables and Chern-Simons invariants of $SU(2)$ flat connections on $M(-2;\frac{1}{2},\frac{1}{3},\frac{1}{2})$, along with the action of center symmetry on them.}
\label{table:6+2hol}
\end{table}

Apart from the Weyl group, we conjecture that there is an outer automorphism acting on the moduli space, which permutes different components of the moduli space. In terms of the holonomy angles $\lambda_i$, it acts by
\be
\quad (\l,\l_1,\l_2,\l_3) = (\l,\l_1+\tfrac{1}{2},\l_2,\l_3+\tfrac{1}{2}).
\ee
For instance, this maps one abelian flat connection to another as
\begin{equation}
\left(\tfrac{1}{4}, \tfrac{1}{8},\tfrac{1}{4},\tfrac{5}{8} \right) + \left(0,\tfrac{1}{2},0,\tfrac{1}{2} \right) \ \sim \ \left(\tfrac{1}{4}, \tfrac{5}{8}, \tfrac{1}{4},\tfrac{1}{8} \right),
\label{eqn:degNAfc232}
\end{equation}
where we have taken the action of the Weyl group into account. The orbits of center symmetry are shown in Table~\ref{table:6+2hol}.

\iffalse
Under the center symmetry, holonomy angles turn out to have the same Chern-Simons invariants. For the current example, we follow \cite{Auckly} to find:
\be
\text{CS} = \begin{cases} 0 &\text{for} \ \left(0,0,0,0\right) \sim \left( 0, \tfrac{1}{2}, 0, \tfrac{1}{2} \right) \\
\frac{7}{8} &\text{for} \ \left(\tfrac{1}{4}, \tfrac{1}{8},\tfrac{1}{4},\tfrac{5}{8}\right) \sim \left(\tfrac{1}{4},\tfrac{5}{8},\tfrac{1}{4},\tfrac{1}{8}\right) \\
\frac{1}{2} &\text{for} \ \left(\tfrac{1}{2},\tfrac{1}{4},\tfrac{1}{2},\tfrac{1}{4}\right), \end{cases}
\label{eqn:CS6+2holonomy}
\ee
which is consistent with \eq{eqn:CS6+2}.
\fi

We claim that the outer automorphism is not only a symmetry of flat connections but also of the moduli space of all connections.
First note that the center symmetry is also manifest in the data of the abelian flat connections \eq{eqn:CS6+2}--\eq{eqn:Zhat6+2}.
Indeed, the corresponding values $\text{CS}(a)$ are equal, \textit{e.g.}, for $a = (0,0,0,0)$ and $a=(1,-1,0,-1)$
and this symmetry is also manifest in the corresponding rows of the $S$-matrix $S^{(A)}$.
As a result, from \eq{relateZhatZ} we see that the asymptotic expansions around these two
abelian flat connections are identical. The prediction is consistent with the transseries in Table~\ref{table:6+2}.
Since not only the values $\text{CS}(a)$ but also the perturbative expansions around the flat connections are identical,
we conclude that the center symmetry is indeed a symmetry of the moduli space.

Next, let us see what happens when we identify flat connections related by the center symmetry.
The data of the abelian flat connections becomes:
\begin{equation}
\begin{gathered}
\text{CS}(a) =  \begin{cases} 0 \quad \mod \quad \mathbb{Z} & \text{for} \quad a = (0,0,0,0)\sim(1,-1,0,-1) \\
\frac{7}{8} \quad \mod \quad \mathbb{Z} &\text{for} \quad a=(0,-1,0,0) \sim (0,0,0,-1) \\
\frac{1}{2} \quad \mod \quad \mathbb{Z} &\text{for} \quad a = (0,0,-1,0) \end{cases} \\
S^{(A)} =  \frac{1}{\sqrt{2}} \begin{pmatrix} 2 & 2 & 1 \\ 4 & 0 & -2 \\ 2 & -2 & 1  \end{pmatrix}\\
\widehat{Z}_0(q) = \widehat{Z}_{(3,-1,-5,-3)}(q) + \widehat{Z}_{(3,-3,-5,-1)}(q) = -q^{-5/12}\Psi^{6+2}_2(\tau)  \\
\widehat{Z}_1(q) = \widehat{Z}_{(1,-1,-1,-1)}(q) + \widehat{Z}_{(3,-3,-1,-3)}(q) = 2q^{-5/12}(1-\Psi^{6+2}_1(\tau) ) \\
\widehat{Z}_2(q) = \widehat{Z}_{(1,-3,-1,-1)}(q) =  q^{-5/12}\Psi^{6+2}_4(\tau) .
\end{gathered}
\end{equation}
One can easily see now that $S^{(A)}$ is now non-degenerate and, furthermore, false theta functions match perfectly the ``folded'' homological blocks without degneracy.
Therefore, we may conclude that the modularity dictionary should be used after modding out by the symmetries of the moduli space.

\subsection{Example: $M(-1;\frac{1}{2},\frac{1}{3},\frac{1}{10})$}

We present another example with the center symmetry. This time, it is \textit{necessary} to mod out by center symmetry in order to find an appropriate Weil representation $m+K$.

\begin{table}[htb]
\centering
\begin{tabular}{c c c c}
CS invariant & type & holonomy angles & center symmetry \\
\hline \\
$0$ & abelian & $ ( 0,0,0,0 ) $ & $ ( 0,0,0,0 ) \mapsto ( 0,\tfrac{1}{2},0,\tfrac{1}{2}) $ \\[1.5ex]
$0$ & abelian & $ ( 0,\tfrac{1}{2},0,\tfrac{1}{2} ) $ & $ ( 0,\tfrac{1}{2},0,\tfrac{1}{2} ) \mapsto (0,0,0,0)$ \\[1.5ex]
$\tfrac{1}{4}$ & abelian & $ (\tfrac{1}{2},\tfrac{1}{4},\tfrac{1}{2},\tfrac{1}{4} ) $ & $(\tfrac{1}{2},\tfrac{1}{4},\tfrac{1}{2},\tfrac{1}{4} ) \mapsto (\tfrac{1}{2},\tfrac{1}{4},\tfrac{1}{2},\tfrac{1}{4} )  $  \\[1.5ex]
$-\tfrac{25}{60}$ & non-abelian & $( \frac{1}{2},\frac{1}{4},\frac{1}{6},\frac{1}{4} )$ & $( \frac{1}{2},\frac{1}{4},\frac{1}{6},\frac{1}{4} ) \mapsto ( \frac{1}{2},\frac{1}{4},\frac{1}{6},\frac{1}{4} )$ \\[1.5ex]
$-\tfrac{49}{60}$ & non-abelian & $( \frac{1}{2},\frac{1}{4},\frac{1}{6},\frac{3}{20} )$ & $( \frac{1}{2},\frac{1}{4},\frac{1}{6},\frac{3}{20} ) \mapsto ( \frac{1}{2},\frac{1}{4},\frac{1}{6},\frac{7}{20} )$ \\[1.5ex]
$-\tfrac{49}{60}$ & non-abelian & $( \frac{1}{2},\frac{1}{4},\frac{1}{6},\frac{7}{20} )$ & $( \frac{1}{2},\frac{1}{4},\frac{1}{6},\frac{7}{20} ) \mapsto ( \frac{1}{2},\frac{1}{4},\frac{1}{6},\frac{3}{20} )$
\end{tabular}
\caption{Holonomy angles and Chern-Simons invariants of $SU(2)$ flat connections on $M(-1;\frac{1}{2},\frac{1}{3},\frac{1}{10})$, along with the action of center symmetry.}
\label{table:24+10hol}
\end{table}
As before, we characterize flat connections by holonomy angles $\l$. The angles and their Chern-Simons invariants are summarized in Table~\ref{table:24+10hol}. The center symmetry acts by
$$(\l,\l_1,\l_2,\l_3) \ \mapsto \ (\l,\l_1+\tfrac{1}{2},\l_2, \l_3+\tfrac{1}{2}).$$

\subsubsection{$q$-series invariants}

The manifold of interest has $\mathrm{Tor}H_1(M_3) = \mathbb{Z}_4$ and the following plumbing graph:
\begin{equation}
\begin{array}{ccc}
& \overset{\displaystyle{-3}}{\bullet} & \\
& \vline & \\
\overset{\displaystyle{-2}}{\bullet}
\frac{\phantom{xxx}}{\phantom{xxx}}
& \underset{\displaystyle{-1}}{\bullet} &
\frac{\phantom{xxx}}{\phantom{xxx}}
\overset{\displaystyle{-10}}{\bullet}
\end{array}
\end{equation}
From its adjacency matrix, we can compute:
\begin{equation}
\begin{aligned}\notag
a \in \mathrm{coker}M/\mathbb{Z}_2 &= \langle (0,0,0,0), (1,-1,0,-5),(1,0,-1,-7) \rangle \\
b \in (2 \mathrm{coker} M + \delta)/\mathbb{Z}_2 &= \langle (1,-1,-1,-1), (3,-3,-1,-11), (3,-1,-3,-15) \rangle
\end{aligned}
\label{eqn:2310ab}
\end{equation}
\begin{gather}
\text{CS}(a) =  -(a,M^{-1}a) =  \begin{cases} 0 \quad \mod \quad \mathbb{Z} & \text{for} \quad a = (0,0,0,0), (1,-1,0,-5) \\
\frac{1}{4} \quad \mod \quad \mathbb{Z} &\text{for} \quad (1,0,-1,-7) \end{cases} \label{eqn:aMa2310} \\
S^{(A)} = \frac{1}{2} \begin{pmatrix} 1 & 1 & 1 \\ 1 & 1 & 1 \\ 2 & 2 & -2 \end{pmatrix}\\
\widehat{Z}_{(1,-1,-1,-1)}(q)=q^{5/4}(1+q^6-q^{28}+q^{62}+ \cdots) \label{eqn:2310Zhat1}\\
\widehat{Z}_{(3,-3,-1,-11)}(q) =q^{13/4}(-1-q^{12}+q^{14}+q^{38}-q^{82}+\cdots) \\
\widehat{Z}_{(3,-1,-3,-15)}(q) = -q^{3/2}(1-q^3+q^4-q^{11}+q^{19}-q^{32}-q^{52}+\cdots) \label{eqn:2310Zhat3}
\end{gather}
From which it follows that:
\begin{equation}
4m = \mathrm{l.c.m.} ( 8,12,40, 1,4 ) = 120 ~ \Rightarrow ~ m = 30.
\end{equation}

\subsubsection{Folding with the center symmetry}

Unlike what happens in the previous example, here the homological blocks~\eq{eqn:2310Zhat1}--\eq{eqn:2310Zhat3} do not correspond to any level 30 false theta function (although they do correspond to certain level 60 false theta functions). In what follows we show how this problem is resolved by folding with the center symmetry.

First note that the center symmetry is also manifest in the data of the abelian flat connections \eq{eqn:aMa2310}--\eq{eqn:2310Zhat3}.
Indeed, the values $\text{CS}(a)$ are equal for $a = (0,0,0,0)$ and $a=(1,-1,0,-5)$.
Moreover, the corresponding rows of $S^{(A)}$ enjoy this symmetry as well.
As a result, from \eq{relateZhatZ} we see that the asymptotic expansions around these two abelian flat connections are identical. Since not only $\text{CS}(a)$ but also the perturbative expansions around the flat connections are identical, this indicates that the center symmetry is indeed a symmetry of the moduli space.

Next, let us see what happens when we identify flat connections related by the center symmetry.
The data of the abelian flat connections becomes:
\begin{equation}
\begin{gathered}
\text{CS}(a) =  \begin{cases} 0 \quad \mod \quad \mathbb{Z} & \text{for} \quad a = (0,0,0,0)\sim(1,-1,0,-5) \\
\frac{1}{4} \quad \mod \quad \mathbb{Z} &\text{for} \quad a=(1,0,-1,-7) \end{cases} \\
S^{(A)} =  \begin{pmatrix} 1 & 1 \\ 1 & -1 \end{pmatrix}\\
\widehat{Z}_0(q) = \widehat{Z}_{(1,-1,-1,-1)}(q) + \widehat{Z}_{(3,-3,-1,-11)}(q) = q^{5/4}(1-q^2+q^6-q^{14}+q^{16}+\cdots) \\
\widehat{Z}_1(q) = \widehat{Z}_{(3,-1,-3,-15)}(q) = -q^{3/2}(1-q^3+q^4-q^{11}+q^{19}-q^{32}-q^{52}+\cdots) .
\end{gathered}
\end{equation}
As expected, now $S^{(A)}$ is non-degenerate and, furthermore, false theta functions perfectly match the ``folded'' $q$-series invariants $\widehat{Z}_a (M_3)$. This supports our proposal for applying the modularity dictionary {\it after} modding out by the symmetries of the moduli space. The resulting Weil representation is $m+K = 15+5$:
\begin{equation}
\begin{aligned}
&\sigma^{15+5} =\{1,2,4,5,7,10\} ~~(\text{irrep, genus 0}) \\
&\widehat{Z}_0(q) = q^{37/30}\Psi^{15+5}_1(\tau) \\
&\widehat{Z}_1(q) = -q^{37/30}\Psi^{15+5}_4(\tau) \,.
\end{aligned}
\end{equation}

\subsubsection{$S(M_3)$, $T(M_3)$, and the asymptotic expansions}

As before, we can proceed to compute the (numeric values of the) composite  matrices $S(M_3)$ and $T(M_3)$:
\begin{equation}
\begin{aligned}
\text{\textbf{Emb}} &= \begin{pmatrix} 1 & 0 & 0 & 0 & 0 & 0 \\ 0 & 0 & -1 & 0 & 0 & 0 \end{pmatrix},  \quad T^{(A)} = \exp \; 2 \pi i \begin{pmatrix}0 & 0 \\  0 & \frac{1}{4} \end{pmatrix} \\
S^{(B)} &= i \begin{pmatrix} 0.20 & -0.51 & -0.20 & -0.32 & -0.51 & -0.32 \\ -0.51 & -0.20 & -0.51 & -0.32 & 0.20 & 0.32 \\ -0.20 & -0.51 & -0.20 & 0.32 & 0.51 & -0.32 \\ -0.63 & -0.63 & 0.63 & 0.32 & -0.63 & 0.32 \\ -0.51 & 0.20 & 0.51 & -0.32 & 0.20 & -0.32 \\ -0.63 & 0.63 & -0.63 & 0.32 & -0.63 & -0.32 \end{pmatrix} \\
T^{(B)} &= \exp \; 2 \pi i \cdot \mathrm{diag} \big( \tfrac{1}{60}, \tfrac{4}{60}, \tfrac{16}{60},\tfrac{25}{60}, \tfrac{49}{60},\tfrac{100}{60} \big)
\end{aligned}
\end{equation}
It follows that
\begin{equation}
\begin{aligned}
S(M_3) &= i \begin{pmatrix} -0.39 & 0 & 0 & 0.63 & 1.02 & 0 \\ 0 & 1.02 & 0.39 & 0 & 0 & 0.63 \end{pmatrix} \\
\renewcommand*{\arraystretch}{1.5}
T(M_3) &= \bem \ex(-\frac{1}{60}) & \ex(-\frac{4}{60}) & \ex(-\frac{16}{60}) & \ex(-\frac{25}{60}) & \ex(-\frac{49}{60}) & \ex(-\frac{100}{60}) \\[1.5ex] \ex(-\frac{46}{60}) & \ex(-\frac{49}{60}) & \ex(-\frac{1}{60}) & \ex(-\frac{10}{60}) & \ex(-\frac{34}{60}) & \ex(-\frac{25}{60}) \eem,
\end{aligned}
\end{equation}
from which we conclude that the Chern-Simons invariants of non-abelian flat connections are $-\tfrac{1}{60},-\tfrac{25}{60},-\tfrac{49}{60}.$
As only $\mathbf{c}_{-\frac{1}{60}}$ vanishes, we predict that there are two real non-abelian flat connections with $\text{CS} = -\frac{25}{60},-\frac{49}{60}$ and one (or two, but related by the center symmetry) complex flat connections with $\text{CS} = -\frac{1}{60}$. The asymptotic expansions are computed and summarized in Table~\ref{table:15+5center}, where we have omitted the overall factor $ -iq^{-37/30}/2\sqrt{2}$.

\begin{table}[h]\begin{center}\scalebox{0.9}{
\begin{tabular}{c c c c}
CS action & stabilizer & type & transseries \vspace{3pt}
\\ \toprule
$0$ & $SU(2)$ & central & $e^{2 \pi i k \cdot 0} \bigg( \pi i k^{-3/2} + \frac{283 \pi^2}{60} k^{-5/2} + \mathcal{O}(k^{-7/2}) \bigg)$ \\[1.5ex]
$\frac{1}{4}$ & $U(1)$ & abelian & $e^{2 \pi i k \frac{1}{4}} \bigg( \frac{4}{3}k^{-1/2} - \frac{49 \pi i}{135}k^{-3/2} + \mathcal{O}(k^{-5/2}) \bigg)$\\[1.5ex]
$-\frac{25}{60}$ & $\pm1$ & non-abelian, real & $e^{-2\pi ik \frac{25}{60}} e^{\frac{3 \pi i}{4}} \cdot\frac{1}{\sqrt{10}} $\\[1.5ex]
$-\frac{49}{60}$ & $\pm1$ & non-abelian, real & $e^{-2\pi ik \frac{49}{60}} e^{\frac{3 \pi i}{4}} \cdot \frac{4\sqrt{2}}{\sqrt{15}} (\cos\frac{\pi}{30}+\sin\frac{2\pi}{15})  $  \\[1.5ex]
$-\frac{1}{60}$ & $\pm1$ & non-abelian, complex & $0$ \vspace{2pt}\\ \bottomrule
\end{tabular}}
\caption{Transseries and classification of flat connections on $M(-1;\frac{1}{2},\frac{1}{3},\frac{1}{10})$, \textit{after} modding out the center symmetry.}
\label{table:15+5center}
\end{center}
\end{table}

Note that Table~\ref{table:15+5center} is obtained \textit{after} modding out by the center symmetry.
In particular, the transseries of the ``central'' flat connection stands for the sum of two identical transseries around $a = (0,0,0,0)$ and $a=(1,-1,0,-5)$. As mentioned before, we must multiply the above answer by a factor of $\frac{1}{2}$ in order to recover the contribution from each of the two central flat connections.

Likewise, in Table~\ref{table:6+2hol} we see that there are two real non-abelian flat connections that get identified by the center symmetry.
As a check, we compute the Chern-Simons invariants from the holonomy variables using the formula
\begin{multline}
\text{CS}[(\l,\l_i);M(b, \{ q_i/p_i \}_{i=1}^n)] = - \Big( \sum_{i = 1}^3 p_i r_i \l_i^2 - q_i s_i \frac{1}{2^2} \Big) \\
= \begin{cases} -\frac{49}{60} &\text{for} \ (\l_1, \l_2, \l_3) = (\frac{1}{4},\frac{1}{6},\frac{3}{20}) \ \text{and} \ (\frac{1}{4},\frac{1}{6},\frac{7}{20}) \\ -\frac{25}{60} &\text{for} \ (\l_1, \l_2, \l_3) = (\frac{1}{4},\frac{1}{6},\frac{5}{20}). \end{cases}
\end{multline}
In the first line, $r_i$ and $s_i$ are any integers satisfying $p_i s_i - q_i r_i = 1$.
It follows that degenerate non-abelian flat connections have $\text{CS} = -\tfrac{49}{60}$. As a result, we predict that our manifold has
\begin{itemize}
\item one complex flat connection with $\text{CS} = -\frac{1}{60}$
\item two real non-abelian flat connections with $\text{CS} = -\frac{49}{60}$
\item one real non-abelian flat connection with $\text{CS} = -\frac{25}{60}.$
\end{itemize}

\subsubsection{Comparison with A-polynomial}

Note that we have not ruled out the possibility that there can be extra complex flat connections related by the center symmetry.
To investigate this, recall that since $M_3 = M(-1;\frac{1}{2},\frac{1}{3},\frac{1}{10})$
is a $-4/1$ surgery along the right-handed trefoil,
we can compute the total number of real/complex non-abelian flat connections by studying its A-polynomial.
Counting the intersection points of algebraic curves defined by equations
$$A(x,y) = (y-1)(y x^6 + 1) \quad \text{and} \quad s(x,y) = y x^{-4}-1,$$
we find a total of \textit{four} non-abelian flat connections, which agrees with the number found in the previous section.
Therefore, the complex flat connections are non-degenerate with respect to the action of center symmetry,
and we can finalize the transseries as in Table~\ref{table:15+5}. (Again, the overall factor $-iq^{-37/30}/2\sqrt{2}$ is omitted.)
\begin{table}[h]\begin{center}\scalebox{0.9}{
\begin{tabular}{c c c c}
CS action & stabilizer & type & transseries \vspace{3pt}
\\ \toprule
$0$ & $SU(2)$ & central & $e^{2 \pi i k \cdot 0} \bigg( \frac{\pi i }{2} k^{-3/2} + \frac{283 \pi^2}{120} k^{-5/2} + \mathcal{O}(k^{-7/2}) \bigg)$ \\[1.5ex]
$0$ & $SU(2)$ & central & $e^{2 \pi i k \cdot 0} \bigg( \frac{\pi i }{2} k^{-3/2} + \frac{283 \pi^2}{120} k^{-5/2} + \mathcal{O}(k^{-7/2}) \bigg)$ \\[1.5ex]
$\frac{1}{4}$ & $U(1)$ & abelian & $e^{2 \pi i k \frac{1}{4}} \bigg( \frac{4}{3}k^{-1/2} - \frac{49 \pi i}{135}k^{-3/2} + \mathcal{O}(k^{-5/2}) \bigg)$\\[1.5ex]
$-\frac{25}{60}$ & $\pm1$ & non-abelian, real & $e^{-2\pi ik \frac{25}{60}} e^{\frac{3 \pi i}{4}} \cdot\frac{1}{\sqrt{10}} $\\[1.5ex]
$-\frac{49}{60}$ & $\pm1$ & non-abelian, real & $e^{-2\pi ik \frac{49}{60}} e^{\frac{3 \pi i}{4}} \cdot \frac{2\sqrt{2}}{\sqrt{15}} (\cos\frac{\pi}{30}+\sin\frac{2\pi}{15})  $  \\[1.5ex]
$-\frac{49}{60}$ & $\pm1$ & non-abelian, real & $e^{-2\pi ik \frac{49}{60}} e^{\frac{3 \pi i}{4}} \cdot \frac{2\sqrt{2}}{\sqrt{15}} (\cos\frac{\pi}{30}+\sin\frac{2\pi}{15})  $  \\[1.5ex]
$-\frac{1}{60}$ & $\pm1$ & non-abelian, complex & $0$ \vspace{2pt}\\ \bottomrule
\end{tabular}}
\caption{Transseries and classification of flat connections on $M(-1;\frac{1}{2},\frac{1}{3},\frac{1}{10})$.}
\label{table:15+5}
\end{center}
\end{table}

\subsection{Example: $M(-1;\frac{1}{2},\frac{1}{3},\frac{1}{8})$}

\subsubsection{$q$-series invariants}

The manifold of interest has $\mathrm{Tor}H_1(M_3) = \mathbb{Z}_2$ and the following plumbing graph:
\begin{equation}
\begin{array}{ccc}
& \overset{\displaystyle{-3}}{\bullet} & \\
& \vline & \\
\overset{\displaystyle{-2}}{\bullet}
\frac{\phantom{xxx}}{\phantom{xxx}}
& \underset{\displaystyle{-1}}{\bullet} &
\frac{\phantom{xxx}}{\phantom{xxx}}
\overset{\displaystyle{-8}}{\bullet}
\end{array}
\end{equation}
From its adjacency matrix, we can compute:
\begin{equation}
\begin{aligned}
a \in \mathrm{coker}M/\mathbb{Z}_2 &= \langle (0,0,0,0), (1,-1,0,-4) \rangle \\
b \in (2 \mathrm{coker} M + \delta)/\mathbb{Z}_2 &= \langle (1,-1,-1,-1), (3,-3,-1,-9) \rangle
\end{aligned}
\label{eqn:238ab}
\end{equation}
\begin{gather}
\text{CS}(a) =  -(a,M^{-1}a) =  \begin{cases} 0 \quad \mod \quad \mathbb{Z} & \text{for} \quad a = (0,0,0,0) \\
\frac{1}{2} \quad \mod \quad \mathbb{Z} &\text{for} \quad (1,0,-1,-4) \end{cases} \label{eqn:aMa238} \\
S^{(A)}
= \frac{1}{2} \begin{pmatrix} 1 & 1 \\ 1 & 1 \end{pmatrix}\\
\widehat{Z}_{(1,-1,-1,-1)}(q) =-q^{3/4}(-1 + q^3 - q^{10}+q^{23}-q^{25}+q^{44} + \cdots) \label{eqn:238Zhat1}\\
\widehat{Z}_{(3,-3,-1,-9)}(q) =q^{5/4}(-1+q^5-q^6+q^{17}-q^{31}+q^{52}-q^{55}+\cdots)
\label{eqn:238Zhat2}
\end{gather}
From which it follows that:
\begin{equation}
4m = \mathrm{l.c.m.} ( 8,12,32, 1,2 ) = 96 ~ \Rightarrow ~ m = 24.
\end{equation}
One observes that $q$-series invariants $\widehat{Z}_b (M_3)$ correspond naturally to the irreducible Weil representation $m+K = 24+8$:
\begin{equation}
\begin{aligned}
&\sigma^{24+8} =\{1,2,5,7,8,13\}\\
&\widehat{Z}_{(1,-1,-1,-1)}(q) = q^{71/96}\Psi^{24+8}_1(\tau) \\
&\widehat{Z}_{(3,-3,-1,-9)}(q) = -q^{71/96}\Psi^{24+8}_7(\tau).
\end{aligned}
\end{equation}

\subsubsection{Computing $S(M_3)$, $T(M_3)$ and the asymptotic expansions}
As before, one can proceed to compute the (numeric values of the) composite  matrices $S(M_3), T(M_3)$. The results are

\begin{equation}
\begin{aligned}
\text{\textbf{Emb}} &= \begin{pmatrix} 1 & 0 & 0 & 0 & 0 & 0 \\ 0 & 0 & 0 & -1 & 0 & 0 \end{pmatrix},  \quad T^{(A)} = \exp \left( 2 \pi i \begin{pmatrix}0 & 0 \\  0 & \frac{1}{2} \end{pmatrix} \right) \\
S^{(B)} &= i \begin{pmatrix} 0.19 & -0.71 & -0.46 & -0.19 & -0.5 & -0.46 \\ -0.35 & 0 & -0.35 & -0.35 & 0 & 0.35 \\
-0.46 & -0.71 & -0.19 & 0.46 & 0.5 & -0.19 \\
-0.19 & -0.71 & 0.46 & 0.19 & -0.5 & 0.46 \\
-0.5 & 0 & 0.5 & -0.5 & 0 & -0.5 \\
-0.46 & 0.71 & -0.19 & 0.46 & -0.5 & -0.19  \end{pmatrix} \\
T^{(B)} &= \exp \; 2 \pi i \cdot \mathrm{diag} \left( \tfrac{1}{96}, \tfrac{4}{96}, \tfrac{25}{96},\tfrac{49}{96}, \tfrac{64}{96},\tfrac{169}{96} \right)
\end{aligned}
\end{equation}
From which it follows that
\begin{equation}
\begin{aligned}
S(M_3) &= i \begin{pmatrix} -0.54 & 0 & 1.31 & 0.54 &0 & 1.31 \\ -0.54 & 0 & 1.31 & 0.54 & 0 & 1.31 \end{pmatrix} \\
\renewcommand*{\arraystretch}{1.5}
T(M_3) &= \bem \ex(-\frac{1}{96}) & \ex(-\frac{4}{96}) & \ex(-\frac{25}{96}) & \ex(-\frac{49}{96}) & \ex(-\frac{64}{96}) & \ex(-\frac{169}{96}) \\[1.5ex] \ex(-\frac{49}{96}) & \ex(-\frac{52}{96}) & \ex(-\frac{169}{96}) & \ex(-\frac{1}{96}) & \ex(-\frac{64}{96}) & \ex(-\frac{25}{96}) \eem,
\end{aligned}
\end{equation}
from which we conclude that the Chern--Simons invariants of non-abelian flat connections are $-\tfrac{1}{96}$,$-\tfrac{25}{96}$,$-\tfrac{49}{96}$, and $-\tfrac{169}{96}$.
Since $\mathbf{c}_\alpha$ vanishes for $\alpha = -\frac{1}{96},-\frac{49}{96}$, we predict that there are two real non-abelian flat connections with $\text{CS} = -\frac{25}{96}, -\frac{169}{96}$ and two complex flat connections with $\text{CS} = -\frac{1}{96},-\frac{49}{96}$. The asymptotic expansion is computed and summarized in Table~\ref{table:24+8}. We have omitted the overall factor $-iq^{-71/96}/2\sqrt{2}$.

\begin{table}[h]\begin{center}\scalebox{0.9}{
\begin{tabular}{c c c c}
CS action & stabilizer & type & transseries \vspace{3pt}
\\ \toprule
$0$ & $SU(2)$ & central & $e^{2 \pi i k \cdot 0} \bigg( \frac{\pi i}{2} k^{-3/2} + \frac{359 \pi^2}{96} k^{-5/2} + \mathcal{O}(k^{-7/2}) \bigg)$ \\[1.5ex]
$\frac{1}{2}$ & $SU(2)$ & central & $e^{2 \pi i k \frac{1}{2}} \bigg( \frac{\pi i}{2} k^{-3/2} + \frac{359 \pi^2}{96} k^{-5/2} + \mathcal{O}(k^{-7/2}) \bigg)$ \\[1.5ex]
$-\frac{25}{96}$ & $\pm1$ & non-abelian, real & $e^{-2\pi ik \frac{25}{96}} e^{\frac{3 \pi i}{4}} \cdot\frac{1}{\sqrt{3}}\cos\frac{\pi}{6}\cos\frac{\pi}{8} $\\[1.5ex]
$-\frac{169}{96}$ & $\pm1$ & non-abelian, real & $e^{-2\pi ik \frac{169}{96}} e^{\frac{3 \pi i}{4}} \cdot\frac{1}{\sqrt{3}}\cos\frac{\pi}{6}\cos\frac{\pi}{8} $  \\[1.5ex]
$-\frac{1}{96}$ & $\pm1$ & non-abelian, complex & $0$ \\[1.5ex]
$-\frac{49}{96}$ & $\pm1$ & non-abelian, complex & $0$ \vspace{2pt}\\ \bottomrule
\end{tabular}}
\caption{Transseries and classification of flat connections on $M(-1;\frac{1}{2},\frac{1}{3},\frac{1}{8})$.}
\label{table:24+8}
\end{center}
\end{table}

Note the degeneracy in Table~\ref{table:24+8}, which arises due to the degeneracy in $S^{(A)}_{ab}$. Therefore, we verify the presence of center symmetry by studying the holonomy angles of $SU(2)$ flat connections. The angles and their Chern-Simons invariants are summarized in Table~\ref{table:24+8hol}. The center symmetry acts by an addition of $(0,\tfrac{1}{2},0,\tfrac{1}{2})$. But this time, the Chern-Simons invariants are shifted by $1/2$ due to center symmetry. Therefore, Table~\ref{table:24+8} shows transseries without ambiguity.

\begin{table}[htb]
\centering
\begin{tabular}{c c c c}
CS invariant & type & holonomy angles & center symmetry \\
\hline \\
$0$ & abelian & $ ( 0,0,0,0 ) $ & $ ( 0,0,0,0 ) \mapsto  ( 0,\tfrac{1}{2},0, \tfrac{1}{2} ) $ \\[1.5ex]
$\tfrac{1}{2}$ & abelian & $ ( 0,\frac{1}{2},0,\frac{1}{2} ) $ & $ ( 0,\frac{1}{2},0,\frac{1}{2} ) \mapsto (0,0,0,0)$ \\[1.5ex]
$-\tfrac{25}{96}$ & non-abelian & $( \frac{1}{2},\frac{1}{4},\frac{1}{6},\frac{3}{16} )$ & $( \frac{1}{2},\frac{1}{4},\frac{1}{6},\frac{3}{16} ) \mapsto ( \frac{1}{2},\frac{1}{4},\frac{1}{6},\frac{5}{16} )$\\[1.5ex]
$-\tfrac{169}{96}$ & non-abelian & $( \frac{1}{2},\frac{1}{4},\frac{1}{6},\frac{5}{16} )$ & $( \frac{1}{2},\frac{1}{4},\frac{1}{6},\frac{5}{16} ) \mapsto ( \frac{1}{2},\frac{1}{4},\frac{1}{6},\frac{3}{16} )$
\end{tabular}
\caption{Holonomy angles and Chern-Simons invariants of $SU(2)$ flat connections on $M(-1;\frac{1}{2},\frac{1}{3},\frac{1}{8})$, and the center symmetry among them.}
\label{table:24+8hol}
\end{table}

\subsection{Infinite families}

In this section, we discuss two sets of infinite families of Seifert manifolds with three singular fibers for which the steps outlined in Figure \ref{fig:reverse} can be carried out for all 3-manifolds in the family at once. These examples are the Brieskorn homology spheres and manifolds whose plumbing diagram is a D-type Dynkin diagram with ``-2" at all nodes.

\subsubsection{Brieskorn spheres}
\label{subsec:brieskorn}
A simple class of Seifert manifolds with three singular fibers are the Brieskorn spheres
\be
\Sigma (p_1,p_2,p_3) := S^5 \cap \{ (x,y,z) \in \CC^3 \; \vert \; x^{p_1} + y^{p_2} + z^{p_3} = 0 \}
\label{Brieskorn}
\ee
labeled by a triple of relatively prime integers $(p_1,p_2,p_3)$. As discussed in \cite{NR}, the Brieskorn sphere $\Sigma (p_1,p_2,p_3)$ can be associated with the Seifert data $M\left (-1; {q_1\over p_1},{q_2\over p_2},{q_3\over p_3}\right )$, satisfying \footnote{This holds for all ${1 \over p_1} +{1 \over p_2}  +{1 \over p_3}<1$, which is satisfied for all Brieskorn spheres except the Poincare homology sphere $\Sigma(2,3,5)$ which has Seifert data $M(-2;{1 \over 2},{2 \over 3},{4 \over 5})$.}
\be {q_1 \over p_1} +{q_2 \over p_2}  +{q_3 \over p_3}=1-  {1\over p_1p_2p_3}. \ee

The standard choice of orientation is that Brieskorn spheres are boundaries of negative definite plumbings. The connection between false theta functions and the WRT invariants for this class of examples was discussed in details in \cite{Hikamibrieskorn}, building on \cite{lawrence1999modular}.
These results can be understood in terms of the $q$-series invariants $\widehat{Z}_b (M_3)$ later introduced in \cite{GPPV}. In what follows we present
 them using the language of irreducible Weil representations discussed in \S\ref{sec:WeilRep}, and discuss the resurgence analysis for these manifolds.

All Brieskorn spheres are integral homology spheres, {\it i.e.} have $H_1 (M_3) = 0$.
In particular, this means that there is only one $q$-series invariant $\widehat Z_a (q)$ with $a=0$. Furthermore, there is a simplified modularity dictionary for this class of examples because the $\SL(2,\ZZ)$ representation acting on abelian flat connections is the trivial representation. This is summarized in table \ref{tab:Brieskorndictionary}.
First of all, the homological block $\widehat Z_0(q)$ is given by the false theta function $\Psi^{m+K}_r(\tau)$  (up to an overall power of $q$)
where $m=p_1p_2p_3$ and $K=\{1,p_1p_2,p_2p_3,p_1p_3\}$, and $r=m -p_1p_2-p_2p_3-p_1p_3$.
In the notation from \S \ref{sec:modularity}, we have that $S^{(A)}=T^{(A)}=I_{1\times 1}$, the 1-by-1 identity matrix. Therefore, the composite matrix $S(M_3)$ is simply the $1\times d$ matrix, where
$$d=|\sigma^{m+K}|={1\over 4}(p_1-1)(p_2-1)(p_3-1)$$ is the dimension of the Weil representation $m+K$.
Using a natural map from $\{1, \ldots , d\}$ to $\sigma^{m+K}$, and write the image of $k$ as $r_k$,
the matrix $S(M_3)$ is given by
\be
S(M_3)_{1k}= (\mathcal S^{m+K})^{-1}_{rr_k}
\ee
where $r=m -p_1p_2-p_2p_3-p_1p_3$ is fixed and $k$ runs from $1, \ldots d$.

The matrix $T(M_3)$ is $1\times d$ given by
\be
T(M_3)_{1k}=\ex(-{r_k^2\over 4m}).
\ee
From equation (\ref{rule_nonAb}), we see that $M_3$ will have a non-abelian flat connection $\alpha$ with ${\rm CS}(\alpha)=-{r_k^2\over 4m}$ as long as $(\mathcal S^{m+K})^{-1}_{rk}\neq 0.$ Furthermore, when this is the case, from (\ref{rule_complex}) the connection will be real as long as $c_r\neq 0$, where $c_r$ is as defined in equation (\ref{eq:c_r}). Note that for $\Psi^{m+K}_r= a_1\Psi_{m,r_1} + a_2 \Psi_{m,r_2} + \ldots + a_n \Psi_{m,r_n}$, $c_r=0$ iff $a_1(m-r_1)+ a_2(m-r_2)+\ldots + a_n(m-r_n)=0$. Thus we can read off directly from the components of the irreducible Weil representation $m+K$ the number of real and complex non-abelian flat connections for the corresponding Brieskorn sphere.
In the following we illustrate this explicitly with a simple example.

\begin{table}[h]
\centering
\setstretch{1.5}
\hspace{-1cm}
\begin{tabular}{cc}
\toprule
Weil representation $m+K$ & $m=p_1p_2p_3$ and $K=\{1,p_1p_2,p_2p_3,p_1p_3\}$
\\\hline
$q$-series invariant $\widehat Z_0 (q)$ & $ \Psi^{m+K}_{r}$, where $r=m -p_1p_2-p_2p_3-p_1p_3$ \\\hline
Number of (real and complex) & \multirow{2}*{$|\sigma^{m+K}|= \frac{1}{4}(p_1-1)(p_2-1)(p_3-1)$}\vspace{-3pt} \\
non-abelian flat connections & \\
\hline
CS invariants of (real or complex) & \multirow{2}*{$\text{CS}=- \tfrac{r^2}{4m}~ \forall ~r\in \sigma^{m+K}$}\vspace{-3pt}\\
non-abelian flat connections &\\\hline
CS invariants of complex& \multirow{2}*{ $\text{CS}=- \tfrac{r^2}{4m}~ s.t. ~\sum_{\ell=1}^{m-1} P_{\ell r}^{m+K} (1-{\ell\over m})=0$} \vspace{-3pt}\\ non-abelian flat connections \\\bottomrule
\end{tabular}
\caption{The modularity dictionary for Brieskorn spheres $\Sigma(p_1,p_2,p_3)$.}
\label{tab:Brieskorndictionary}
\end{table}

\paragraph{Example:}
Resurgence for the Brieskorn spheres $\Sigma(2,3,5)$ and $\Sigma(2,3,7)$ was discussed in detail in \cite{Gukov:2016njj} and for $\Sigma(2,5,7)$ in \cite{Chun:2017dbf}. We will see the reappearance of these examples in \S \ref{sec:optimal}, when we discuss going to the lower half-plane. For now, we briefly discuss resurgence for a new example to explicitly illustrate the procedure we have outlined above.

Let $M_3$ be the Brieskorn sphere $\Sigma(3,4,5)$, which can be represented by the following plumbing graph:
\be
\begin{array}{cccc}
& \overset{\displaystyle{-4}}{\bullet} & & \\
& \vline & & \\
\overset{\displaystyle{-3}}{\bullet}
\frac{\phantom{xxx}}{\phantom{xxx}}
& \underset{\displaystyle{-1}}{\bullet} &
\frac{\phantom{xxx}}{\phantom{xxx}}
\overset{\displaystyle{-3}}{\bullet} &
\frac{\phantom{xxx}}{\phantom{xxx}}
\overset{\displaystyle{-2}}{\bullet}
\end{array}
\ee

From equation (\ref{eqn:homblock}), one can easily compute the single homological block corresponding to the trivial flat connection as
\be
\widehat Z_0(q) \; = \; q^{1/2} \,(1 - q^5 - q^7 - q^{11} + q^{18}+ \ldots ).
\ee
In terms of false theta functions, this is given by,
\be
\widehat Z_0(q) \; =q^{-49/240}\Psi^{60+12,15,20}_{13}(\tau) =q^{-49/240}(\Psi_{60,13}-\Psi_{60,37}-\Psi_{60,43}-\Psi_{60,53})(\tau).
\ee
The irreducible $\SL(2,\ZZ)$ representation is given by $m+K=60+12,15,20$. This has dimension $d= |\sigma^{60+12,15,20}|= {1\over 4}(3-1)(4-1)(5-1)=6$, and contains elements $\sigma^{60+12,15,20}=\{1,2,7,11,13,14\}.$
The corresponding set of false theta functions is given by,
 \begin{gather}
  \begin{split}\nonumber
   \Psi^{60+12,15,20}_{1}(\tau) &=  (\Psi_{60,1}-\Psi_{60,31}-\Psi_{60,41}-\Psi_{60,49})(\tau)\\\nonumber
     \Psi^{60+12,15,20}_{2}(\tau) &=  (\Psi_{60,2}+\Psi_{60,22}+\Psi_{60,38}+\Psi_{60,58})(\tau)\\\nonumber
        \Psi^{60+12,15,20}_{7}(\tau) &=  (\Psi_{60,7}+\Psi_{60,17}+\Psi_{60,23}-\Psi_{60,47})(\tau)\\\nonumber
           \Psi^{60+12,15,20}_{11}(\tau) &=  (\Psi_{60,11}+\Psi_{60,19}+\Psi_{60,29}-\Psi_{60,59})(\tau)\\ \nonumber
      \Psi^{60+12,15,20}_{13}(\tau) &=  (\Psi_{60,13}-\Psi_{60,37}-\Psi_{60,43}-\Psi_{60,53})(\tau)\\\nonumber
 \Psi^{60+12,15,20}_{14}(\tau) &=  (\Psi_{60,14}+\Psi_{60,26}+\Psi_{60,34}+\Psi_{60,46})(\tau).
     \end{split}
\end{gather}

     From this we find that the embedding matrix is simply
     \be
     \text{\bf Emb}=\begin{pmatrix} 0&0&0&0&1&0\end{pmatrix},
     \ee
     which leads to a matrix $S(M_3)$ given by
     \bea
     S(M_3)= \text{\bf Emb}.({\cal S}^{60+12,15,20})^{-1}
     \eea

     Furthermore, the matrix $T(M_3)$ is  given by,
\begin{gather}
\begin{split}
T(M_3)&={\bf 1}_{1\times 6}.({\cal T}^{60+12,15,20})^{-1}\\
&=\begin{pmatrix}\ex\left (-{1\over 240}\right)&\ex\left(-{4\over 240}\right)&\ex\left(-{49\over 240}\right)&\ex\left(-{121\over 240}\right)&\ex\left(-{169\over 240}\right)&\ex\left(-{196\over 240}\right)\end{pmatrix}.
\end{split}
\end{gather}
As each entry of $S(M_3)$ is nonzero, it follows that this manifold has six non-abelian flat connections with CS invariants given by the entries of $T(M_3)$. Furthermore, we see that four of them are real and two of them are complex, as $$(60-1)-(60-31)-(60-41)-(60-49)=(60-13)-(60-37)-(60-43)-(60-53)=0.$$ The complex flat connections have ${\rm CS}=-{1\over 240}$ and ${\rm CS}=-{169\over 240}$.

\subsubsection{D-type manifolds}

   In this section we will consider negative-definite plumbing diagrams whose graph takes the shape of a $D_{k+3}$, $k\geq 1$ Dynkin diagram. The simplest plumbing for such a graph assigns a weight of ``-2" to all nodes, as pictured below:
\be
\begin{array}{ccccccc}
& \overset{\displaystyle{-2}}{\bullet} &&&&& \\
& \vline &&& &&\\
\overset{\displaystyle{-2}}{\bullet}
\frac{\phantom{xxx}}{\phantom{xxx}}
& \underset{\displaystyle{-2}}{\bullet} &
\frac{\phantom{xxx}}{\phantom{xxx}}
\underbrace{\overset{\displaystyle{-2}}{\bullet}\frac{\phantom{xxx}}{\phantom{xxx}}\ldots\frac{\phantom{xxx}}{\phantom{xxx}}\overset{\displaystyle{-2}}{\bullet}}_\text{$k$ nodes}
\end{array}
\ee

This describes a Seifert manifold with three singular fibers and Seifert invariants  $M\left (-2; {1\over 2},{1\over 2}, {k\over k+1}\right).$ This manifold can also be represented as an intersection of a $D_{k+3}$ singularity with a unit sphere in $\mathbb C^3$: 
\be
M\left (-2; {1\over 2},{1\over 2}, {k\over k+1}\right):= S^5 \cap \{ (x,y,z) \in \CC^3 \; \vert \; x^{k} + xy^2 + z^{2} = 0 \}.
\label{Dsing}
\ee
The connection between WRT invariants and false theta functions for these manifolds was considered in \cite{hikami1}; here we analyze them from the point of view of resurgence and (irreducible) Weil representations.
When $k$ is odd, $H_1(M_3)= \ZZ_2\oplus \ZZ_2$ and when $k$ is even, $H_1(M_3) =\ZZ_4$.
In both cases,  the relevant $\SL(2,\ZZ)$ representation is $m+K= k+1$, with $m=k+1$ and $K=\{1\}$ the trivial group. This is an irrep whenever $m$ is a prime to some power; i.e. $m= p^N$. As we will see in \S \ref{sec:optimal}, this includes optimal examples for $m=2,3,4,5,6,7,8,9,10,12,13,16,18,25.$
Interestingly, we will see the phenomenon of center symmetry which played a role in some of our previous examples also reappears here.

We  consider the case of even and odd $k$ separately:
\begin{itemize}
\item $k$ odd:

There are four ${\widehat Z}_b(q)$, none of which are related by Weyl symmetry. With some choice of basis for $H_1(M_3)$, we have
\begin{gather}
\begin{split}
q^{m^2+1\over 4m}\widehat Z_0(q)&=2q^{1/4m}-\Psi^m_1(\t)\\
q^{m^2+1\over 4m}\widehat Z_1(q)&=-\Psi^m_{m-1}(\t)\\
q^{m^2+1\over 4m}\widehat Z_2(q)&=-\Psi^m_{m-1}(\t)\\
q^{m^2+1\over 4m}\widehat Z_3(q)&=-\Psi^m_1(\t).
\end{split}
\end{gather}
Corresponding to the above are the following $(k+3)$-dimensional vectors in $2 {\rm Coker}( M)+ \delta$, up to multiplication by $(2M)^{-1}$:
\be
\begin{gathered}
(\tfrac{1}{2},\cdots,\tfrac{1}{2}) 
\\
(\tfrac{1}{2},0,\tfrac{1}{2},0,\cdots,\tfrac{1}{2},0) 
 \\
(\tfrac{1}{2},0,0,\tfrac{1}{2},\tfrac{1}{2}, \cdots, \tfrac{1}{2}) 
 \\
(\tfrac{1}{2},\tfrac{1}{2},0,0,\tfrac{1}{2},0,\tfrac{1}{2},0,\cdots,\tfrac{1}{2},0) 
\end{gathered}
\ee
where ``$\ldots$" signifies a repetition of ${1\over 2}$ in the first and third lines, and a repetition of $({1\over 2},0)$ in the second and fourth lines.

The center symmetry acts on the above vectors by adding:
\be
\begin{gathered}
(0,0,\tfrac{1}{2},\tfrac{1}{2},0,\tfrac{1}{2},0,\tfrac{1}{2},\cdots,0,\tfrac{1}{2}) 
\\
(0,\tfrac{1}{2},\tfrac{1}{2},0, \cdots 0) 
\end{gathered}
\ee
From this action we can infer that it is possible to fold the homological blocks by taking the linear combinations $\widehat{Z}'_0(q)=\widehat{Z}_0(q) + \widehat{Z}_3(q)$ and $\widehat{Z}'_1(q)=\widehat{Z}_1(q) + \widehat{Z}_2(q)$, and that the center symmetry group is $\mathbb{Z}_2 \oplus \mathbb{Z}_2$.

The Chern-Simons invariants of the abelian connections are
\be
{\rm CS}(a)=\{0, {m+2\over 4}, {m+2\over 4}, 1\}.
\ee
Note that for $m/2$ odd this is just ${\rm CS}(a)=\{0, 0,0,0\} \pmod \ZZ$ and for $m/2$ even this is ${\rm CS}(a)=\{0, {1\over 2},{1\over 2},0\} \pmod \ZZ$.

Associated to these abelian connections are a set of $k+3$-dimensional vectors $a \in {\rm Coker}(M)$ which we can take to be (up to multiplication by $M^{-1}$) :
\be
\begin{gathered}
(0,\cdots,0) 
 \\
(0,\tfrac{1}{2}, \cdots, 0,\tfrac{1}{2}) 
\\
(0,\tfrac{1}{2},\tfrac{1}{2},0,\cdots,0) 
\\
(0,0,\tfrac{1}{2},\tfrac{1}{2},0,\tfrac{1}{2},\cdots,0,\tfrac{1}{2}) 
\end{gathered}
\ee
where now ``$\ldots$" signifies a repetition of $0$ in the first and third lines, and a repetition of $(0,{1\over 2})$ in the second and fourth lines.
The center symmetry acts on these vectors by adding:
\be
\begin{gathered}
(0,0,\tfrac{1}{2},\tfrac{1}{2},0,\tfrac{1}{2},\cdots,0,\tfrac{1}{2}) 
\\
(0,\tfrac{1}{2},\tfrac{1}{2},0, \cdots 0) 
\end{gathered}
\ee
From which we can infer that the corresponding ${\rm CS}(a)$ should be grouped as $\{0,1\}$ and $\{\tfrac{m+2}{4},\tfrac{m+2}{4}\}$.

The matrix $S^{(A)}$ is
\begin{equation}
S^{(A)} = \frac{1}{2}\begin{pmatrix} 1 & 1&1&1 \\  1 & 1&1&1\\ 1 & 1&1&1\\ 1 & 1&1&1\end{pmatrix}
\end{equation}
and the matrix $T^{(A)}$ is
\begin{equation}
T^{(A)} = \begin{pmatrix} 1 & 0&0&0 \\  0 & \pm1&0&0\\ 0 & 0&\pm1&0\\ 0 & 0&0&1\end{pmatrix}
\end{equation}
where the ``$+$" is for $m/2$ odd and the ``$-$" is for $m/2$ even. After folding by the center symmetry these matrices become
\begin{equation}
S'^{(A)} = \begin{pmatrix} 1 & 1 \\  1 & 1\end{pmatrix}
\end{equation}
and
\begin{equation}
T'^{(A)} = \begin{pmatrix} 1 & 0 \\  0 & \pm1\end{pmatrix}
\end{equation}

Upon reverse-engineering, we observe no complex flat connection but $m/2$ real non-abelian flat connections with Chern-Simons invariants $-\frac{r^2}{4m}$ for odd $r \in (0,m)$. They are quite degenerate, mostly due to the fact that Chern-Simons invariants are defined modulo one. Therefore, it can be delicate to distinguish the contributions of two non-abelian flat connections with the same Chern-Simons invariants to the asymptotic expansion of $Z_{CS}(M_3)$. Nevertheless, we can reverse-engineer the perturbatitve expansion without ambiguity:
$$\frac{\pi i }{2}k^{-3/2} + \frac{(m^2-2)\pi^2}{8m}k^{-5/2} + \cdots $$
which is identical for all four abelian flat connections due to the center symmetry.

\item $k$ even:

There are four $\widehat Z_b(q)$, two of which are related by Weyl symmetry.  With some choice of basis, after modding out by the Weyl action, we have
\begin{gather}
\begin{split}
q^{m^2+1\over 4m}\widehat Z_0(q)&=2q^{1/4m}-\Psi^m_1(\tau)\\
q^{m^2+1\over 4m}\widehat Z_1(q)&=-\Psi^m_{1}(\tau)\\
q^{m^2+1\over 4m}\widehat Z_2(q)&=-2\Psi^m_{m-1}(\tau).
\end{split}
\end{gather}

Corresponding to the above homological blocks are the following elements $b \in 2 {\rm Coker}( M)+ \delta$, up to multiplication by $(2M)^{-1}$:
\be
\begin{gathered}
(\tfrac{1}{2},\cdots,\tfrac{1}{2}) 
\\
(\tfrac{1}{2},0,0,\tfrac{1}{2},\cdots,\tfrac{1}{2}) 
\\
(0,\tfrac{1}{4},\tfrac{3}{4},\tfrac{1}{2},0,\cdots,\tfrac{1}{2},0) 
\end{gathered}
\ee
where ``$\ldots$" corresponds to repetition of ${1\over 2}$ in the first two lines, and repetition of $({1\over 2},0)$ in the third line.
The center symmetry acts on these vectors through the addition of
$$(0,\tfrac{1}{2},\tfrac{1}{2},0,\cdots,0).$$
From this action we deduce that the center symmetry group is $\mathbb{Z}_2$ for these cases and one can fold the homological blocks by this $\ZZ_2$ by the grouping $\widehat Z'_0(q)=\widehat{Z}_0(q) + \widehat{Z}_1(q)$ and $\widehat Z'_1(q)=\widehat{Z}_2(q)$.

The Chern-Simons invariants of the abelian connections are
\be
{\rm CS}(a)=\{0, 1, {m+2\over 4}\}.
\ee
Note that for $m= 1\mod 4$  this is just ${\rm CS}(a)=\{0, 0,{3\over 4}\} \pmod \ZZ$ and for $m=3 \mod 4$  this is $\text{CS}(a)=\{0, 0, {1\over 4}\} \pmod \ZZ$.

Corresponding to these abelian connections are a set of $k+3$-dimensional vectors $a \in {\rm Coker}(M)$ which we can take to be (up to multiplication by $M^{-1}$):
\be
\begin{gathered}
(0,\cdots,0) 
\\
(0,\tfrac{1}{2},\tfrac{1}{2},0, \cdots, 0) 
 \\
(\tfrac{1}{2},\tfrac{1}{4},\tfrac{3}{4},0,\tfrac{1}{2},\cdots,0,\tfrac{1}{2}) 
\end{gathered}
\ee
where now ``$\ldots$" corresponds to repetition of $0$ in the first two lines, and repetition of $(0,{1\over 2})$ in the third line.
The center symmetry acts on these vectors through the addition of the vector
$$(0,\tfrac{1}{2},\tfrac{1}{2},0, \cdots, 0), $$
and we deduce that upon modding out by the center symmetry group, ${\rm CS}(a)$ are grouped as $\{0,1\}$ and $\{\tfrac{m+2}{4}\}$.

The matrix $S^{(A)}$ is
\begin{equation}
S^{(A)} = \frac{1}{2}\begin{pmatrix} 1 & 1&1 \\  1 & 1&1\\ 2&2&-2\end{pmatrix}
\end{equation}
and the matrix $T^{(A)}$ is
\begin{equation}
T^{(A)} = \begin{pmatrix} 1 & 0&0 \\  0 & 1&0\\ 0 & 0&\ex({m+2\over 4})\end{pmatrix}
\end{equation}

After modding out the center symmetry we obtain the matrices
\begin{equation}
{S'}^{(A)} = \begin{pmatrix} 1 & 1 \\  1 & -1 \end{pmatrix}
\end{equation}
\begin{equation}
{T'}^{(A)} = \begin{pmatrix} 1 & 0 \\  0 & \ex({m+2\over 4}) \end{pmatrix}
\end{equation}
and folded homological blocks
\begin{gather}\begin{split}
q^{m^2+1\over 4m}\widehat{Z}'_0(q) &=q^{m^2+1\over 4m}\left( \widehat{Z}_0(q) + \widehat{Z}_1(q) \right)= 2q^{1/4m}-2\Psi^m_1(\tau) \\
q^{m^2+1\over 4m}\widehat{Z}'_1(q) &= q^{m^2+1\over 4m}\widehat{Z}_2(q) = -2\Psi^m_{m-1}(\tau).\end{split}
\end{gather}
Upon reverse-engineering, we observe no complex flat connection but $(m-1)/2$ real non-abelian flat connections with Chern-Simons invariants $-\frac{r^2}{4m}$ for odd $r \in (0,m)$. We also obtain the following perturbative expansions which are pairwise identical:
\begin{equation}
\begin{gathered}
\frac{\pi i}{2} k^{-3/2} + \frac{(m^2-2)\pi^2}{8m}k^{-5/2} + \cdots \\
\frac{\pi i}{2} k^{-3/2} + \frac{(m^2-2)\pi^2}{8m}k^{-5/2} + \cdots \\
\frac{2}{m}k^{-1/2} + \frac{i (m^2+2) \pi}{6m^2} k^{-3/2} + \cdots \\
\frac{2}{m}k^{-1/2} + \frac{i (m^2+2) \pi}{6m^2} k^{-3/2} + \cdots
\end{gathered}
\end{equation}
\end{itemize}

\section{Going to the other side}
\label{sec:otherside}

In this section we explore what happens to the $q$-series invariants $\widehat{Z}_a (M_3)$
when the orientation of the three-manifold $M_3$ is reversed.
As will be explained shortly, this operation is expected to have the effect of (formally) replacing $q\leftrightarrow q^{-1}$
and then re-expanding the result again as a $q$-series.
Luckily, precisely this question was independently asked by Rademacher \cite{rad_exp_zero} and his followers in the context of
(mock) modular objects and their extension from the upper half-plane (or, $|q|<1$) to the lower half-plane (respectively $|q|>1$), and has gained more attention since the introduction of the notion of quantum modular forms by Zagier \cite{MR2757599}. 

In particular, we mostly focus on the families of 3-manifolds whose $q$-series invariants $\widehat{Z}_a (M_3)$
are given by false theta functions discussed in \S\ref{sec:mod} and illustrated by examples in \S\ref{sec:examples}.
In these cases, we propose that the $q$-series invariants $\widehat{Z}_a (- M_3)$ of a manifold $- M_3$
are given by mock modular forms with shadows (see \S\ref{sec:QMF} for definitions)
associated to the false theta functions. We summarize the relation in Figure~\ref{diagram_mockquantumfalse}.

Our proposal is supported by the following three facts:
\vspace{-5pt}
\begin{itemize}
\item In some cases the false theta functions admit expressions as $q$-hypergeometric series, which converge not only inside but also outside the unit circle. In those cases one can establish that the expression outside the unit circle is given by a mock theta function.
\item The mock theta function and the corresponding false theta function have the same asymptotic expansions transseries structure near $x\in \QQ$ (cf. \eq{trans_false2}, \eq{quantum_mock}), up to 
$x\leftrightarrow -x$, 
precisely as $\widehat{Z}_a (M_3)$ and $\widehat{Z}_a (- M_3)$ should.
\item When the mock modular form can be expressed as a so-called Rademacher sum, one can prove in general that the same Rademacher sum, now performed in the lower rather than upper half-plane, yields precisely the corresponding Eichler integral. In other words, the Rademacher sum yields a function defined on both $\HH$ and $\HH^-$, where they coincide with the mock respectively false theta function.
\end{itemize}
After explaining the physics and topology motivation to go between the upper- and lower-half planes, we explain the above three points in \S\ref{subsec:qhyper}, \S\ref{sec:QMF} and \S\ref{subsec:Rad} respectively.

\subsection{The physics of the other side}
\label{subsec:otherphysics}

As we already mentioned earlier, around \eqref{IndexFactorization}, it is natural to compare
the appropriate extension, or ``leakage'',  of $\widehat{Z}_a (q)$ to $\text{Im} (\tau) < 0$
with the half-index \eqref{BlockD2S1} obtained by orientation reversal (parity) transformation
applied to the original 2d-3d system on $D^2\times_q S^1$.
The two are expected to be closely related, if not simply equal.

In other words, we wish to compare $\widehat{Z}_a (q^{-1})$, understood as a $q$-series expansion,
with the half-index of 2d-3d system where all Chern-Simons coefficients of the 3d theory have opposite signs
and where the 2d $\mathcal{N}=(0,2)$ boundary condition $\mathcal{B}_a$ is replaced by $\widetilde{\mathcal{B}}_a$:
\begin{equation}
\text{parity} : \quad \mathcal{B}_a \; \mapsto \; \widetilde{\mathcal{B}}_a
\label{BBparity}
\end{equation}
The resulting $q$-series --- which, abusing notations, we denote $\widehat{Z}_a (q^{-1})$ --- together
with the original homological blocks $\widehat{Z}_a (q)$ are expected to combine into another $q$-series \eqref{IndexFactorization}
which does not depend on the choice of boundary conditions.
Namely, \eqref{IndexFactorization} gives the superconformal index $\mathcal{I} (q)$
of the 3d $\mathcal{N}=2$ theory which, moreover, can be computed independently, by other means.
In the context of 3d-3d correspondence, to which we turn momentarily,
it is believed that the integer coefficients of the $q$-series $\mathcal{I} (q)$
count normal surfaces in the 3-manifold \cite{Garoufalidis:2016ckn}.

When applied to 3d $\N=2$ theories $T[M_3]$,
this parity reversal is equivalent to changing the orientation of the 3-manifold,
{\it i.e.} replacing $M_3$ by $- M_3$.
Therefore, formally, we expect
\begin{equation}
\widehat{Z}_a(- M_3, q^{-1})
\;\;\stackrel{\text{re-expand}}{=\joinrel=\joinrel=}\;\;
\widehat{Z}_a(M_3, q)
\label{ZZformal}
\end{equation}
While in what follows we present further physics arguments for this relation,
the challenge is to turn them into a concrete computational algorithm.
Note, based on our experience in \S\ref{sec:examples},
we do not expect this algorithm to be simple.
For example, the orientation reveral turns negative-definite
plumbings (for which $\widehat{Z}_a(q)$ can be systematically computed in full generality)
into positive-definite ones (for which no general algorithms were available until now).

{}From the viewpoint of WRT invariants or quantum Chern-Simons theory,
the behavior \eqref{ZZformal} is rather clear, and therefore one might say that $q$-series
invariants $\widehat{Z}_a (M_3)$ simply inherit it through the relation~\eqref{WRTviaSZ}.
Indeed, the Chern-Simons partition function of $M_3$ is defined as a (formal) infinite-dimensional integral
\be
Z_{CS}(M_3,k) \; = \; \int \, e^{ i k \text{CS} (A)} \,\mathcal{D} A
\label{ZCSpathint}
\ee
over the space of gauge connections $A$. At least formally, from this expression
it follows that an orientation reversal $M_3 \to - M_3$ is equivalent to changing the sign $k \to - k$.
If we now recall the standard relation between $k$, $\hbar$, and $q$, {\it cf.} \eqref{WRTviaSZ},
\be
q \; = \; e^{\hbar} \; = \; e^{2 \pi i \tau} \; = \; e^{2\pi i /k}
\ee
then we conclude that $M_3 \to - M_3$ should be equivalent to $q \to q^{-1}$.

While this argument, based on Feynman path integral, may sound a little formal, it is easy to see that it should
hold to all orders in perturbation theory by expanding \eqref{ZCSpathint} into Feynman diagrams around a given flat
connection $\alpha \in \mathcal{M}_{\text{flat}} (G_{\CC},M_3)$.
In Chern-Simons theory with complex gauge group $G_{\CC}$
such perturbative expansion is carried out explicitly {\it e.g.} in \cite{Dimofte:2009yn},
and the coefficient of each term in the $\hbar$-expansion is given by a finite-dimensional integral.
The result is an asymptotic expansion, {\it cf.} \eqref{eqn:borel1},
\be
Z_{\rm pert}^{(\alpha)} (M_3, \hbar) \; = \; \sum_{n} a_n \hbar^n
\label{Zperta}
\ee
which generalizes the Ohtsuki series of $M_3$ (the latter corresponds to $\alpha = 0$.)
Even in complex Chern-Simons theory, where $q$ and $\hbar$ are complex variables,
the perturbative expansion has a symmetry
$Z_{\rm pert}^{(\alpha)} (M_3, \hbar) = Z_{\rm pert}^{(\alpha)} (- M_3, - \hbar)$,
{\it cf.} \cite[sec.2.3]{Dimofte:2009yn}, so that
\be
Z_{\rm pert}^{(\alpha)} (- M_3, \hbar) \; = \; \sum_{n} (-1)^n a_n \hbar^n.
\label{Zpertb}
\ee
Since the $q$-series invariants ${Z}_a (M_3)$ and ${Z}_a (- M_3)$
are obtained by Borel resummation of \eqref{Zperta} and \eqref{Zpertb} for abelian $\alpha = a$,
they too are expected to enjoy the property \eqref{ZZformal}.

Even though \eqref{Zperta} and \eqref{Zpertb} look very similar and appear on the same footing,
in practice, so far it was much easier to compute only one of the $q$-series invariants,
$\widehat{Z}_a (M_3)$ or $\widehat{Z}_a (- M_3)$, while the other remained elusive.
This asymmetry between $M_3$ and $- M_3$ may seem surprising from the topology viewpoint.
However, from the viewpoint of resurgent analysis, it is relatively well known that
among two asymptotic expansions, \eqref{Zperta} and \eqref{Zpertb},
usually one may have a relatively simple Borel resummation,
whereas the other one can be much more complicated~\cite{Marino:2012zq}.
Similarly, from the viewpoint of their modular behavior, which will occupy the rest of this section,
the two sides also usually play rather different asymmetric role.

In particular, in the rest of this section we use a variety of methods and recent developments in number theory
to answer a question in topology: Given $\widehat{Z}_b(M_3)$ and, possibly, some basic topological invariants of $M_3$,
can one determine $\widehat{Z}_b(- M_3)$?

\subsection{Examples: $q$-hypergeometric series}
\label{subsec:qhyper}

In this subsection, we give examples illustrating how certain false theta functions, which play the role of homological blocks for certain three-manifolds (cf. \S\ref{sec:examples}), can be defined in the other side of the plane using their expressions as $q$-hypergeometric series.
Surprisingly, on the other side of the plane they turn out to coincide with some of Ramanujan's famous mock theta functions. This establishes in a very direct way the connection between mock modular forms, false theta functions, and three-manifolds, at least for these examples.
After reviewing the examples, we will also describe the ambiguities when extending a function to the lower-half plane via $q$-hypergeometric series.

\subsubsection{Example: $M(-2;\frac{1}{2},\frac{1}{3},\frac{1}{2})$ and the order three mock theta function $f$}

In \eq{eqn:zhat6+2} we have seen that the false theta function $\Psi^{6+2}_1(\tau)$ coincides, up to additive and multiplicative simple $q$ factors, with the homological blocks $\widehat{Z}_{(1,-1,-1,-1)}(q)$ and $\widehat{Z}_{(3,-3,-1,-3)}(q)$, for the three-manifold $M(-2;\frac{1}{2},\frac{1}{3},\frac{1}{2})$.

To write down the relevant $q$-hypergeometric series, we use the {\em q-Pochhammer symbol}
\be\label{def:qPochhammer}
(a;x)_n := \prod_{k=0}^{n-1} (1-a x^{k})
\ee
satisfying
\be\label{qhyperid}
(a;q^{-1})_n = (-1)^n a^n q^{-{n(n-1)\over 2}} (a^{-1};q)_n.
\ee

 Note that this false theta function admits the expression\cite{BrFoRh}
  \be\label{6p21}
  \Psi^{6+2}_1(\t) =  \psi^{6+2}_1(q),~~ {\psi}^{6+2}_1(q)=  {q^{1\over 24}\over 2}\left(1- \sum_{n\geq 1}{(-1)^nq^{n(n-1)\over 2} \over (-q;q)_n}\right),
  \ee
  for $|q|<1 \Leftrightarrow \t\in \HH$.
 Moreover, the series ${\psi}^{6+2}_1$ converges both for $|q|>1$ and $|q|<1$.
 Using \eq{qhyperid} one obtains
 \be
{\psi}^{6+2}_1(q^{-1}) = {q^{-{1\over 24}}\over 2}\left(1- \sum_{n\geq 1}{(-1)^nq^{n} \over (-q;q)_n}\right).
  \ee
 It turns out that this is a mock modular form, related to the celebrated order three mock theta function $f(q)$  as
 \be\label{6p22}
2 q^{1\over 24} {\psi}^{6+2}_1(q^{-1})  = f(q) =  1+q-2q^2+3q^3+O(q^4) .
 \ee
 As we will see in \S\ref{sec:optimal}, it belongs to a family of special vector-valued mock modular forms $h^{m+K}=(h^{m+K}_r)$; in the notation of \S\ref{sec:optimal} the relation is simply
  \be\label{6p22opt}
 {\psi}^{6+2}_1(q^{-1})  = -{1\over 2} h^{6+2}_1(\t) . 
 \ee

From the argument in \S\ref{subsec:otherphysics}, we hence propose that the mock theta function $f(q)$ plays a role as the homological block for the three-manifold that is related to $M(-2;\frac{1}{2},\frac{1}{3},\frac{1}{2})$ via an orientation reversal. As we will discuss shortly, this example also illustrates the intrinsic ambiguity of the $q$-hypergeometric approach (see \eq{ambiguity_example}).

\subsubsection{Example: $M(-2;\frac{1}{2},\frac{1}{2},\frac{3}{5})$ and the order ten mock theta function $X$}

As we will see in \S\ref{sec:optimal}, the  false theta function
\be
\Psi^{10+2}_1(\t) = \Psi_{10,1}(\t)-\Psi_{10,9}(\t) = q^{1/40}(1-q^2+q^3-q^9+O(q^{10}))
\ee
plays the role of homological blocks for $M_3=M(-2;\frac{1}{2},\frac{1}{2},\frac{3}{5})$.

Note that this false theta function admits the expression
  \be
  \Psi^{10+2}_1(\t) =  \psi^{10+2}_1(q),~~ {\psi}^{10+2}_1(q)=  {q^{1\over 40}}  \sum_{n\geq 0}{(-1)^nq^{n(n+1)} \over (-q;q)_{2n}}
  \ee
  for $|q|<1$.
Similar to $  \psi^{6+2}_1$, the series ${\psi}^{10+2}_1$ converges both for $|q|>1$ and $|q|<1$ and has the following relation to the optimal mock Jacobi form (\S\ref{sec:optimal}) and order 10 mock theta function $X$:
\begin{gather}
\begin{split}
 \psi^{10+2}_1(q^{-1}) &= {q^{-\frac{1}{40}}} \sum_{n\geq 0}{(-1)^nq^{n^2} \over (-q;q)_{2n}}= -h^{10+2}_1(\tau) \\ &=  {q^{-\frac{1}{40}}}X(q)  = {q^{-\frac{1}{40}}} \left(1-q+q^2+O(q^4)\right)
 \end{split}
\end{gather}

\subsubsection{Example: $\Sigma(2,3,5)$ and the order five mock theta function $\chi_0$}

As we have discussed in \S\ref{subsec:brieskorn}, the  false theta function
\be
\Psi^{30+6,10,15}_1(\t) =( \Psi_{30,1}+ \Psi_{30,11}+ \Psi_{30,19}+ \Psi_{30,29})(\t) = q^{1/120}(1+q+q^3+q^7+O(q^{8}))
\ee
plays the role of homological blocks for the homology sphere $\Sigma(2,3,5)$.

Note that this false theta function admits the expression
\be
\Psi^{30+6,10,15}_1(\t) = \psi^{30+6,10,15}_1(q) ,~~  \psi^{30+6,10,15}_1(q)=  {q^{1\over 120}}  \left( 2- \sum_{n\geq 0}{(-1)^n q^{n(3n-1)\over 2} \over (q^{n+1};q)_{n}} \right)
  \ee
   for $|q|<1$.
As before, the series ${\psi}^{30+6,10,15}_1$ converges both for $|q|>1$ and $|q|<1$ and has the following relation to the optimal mock Jacobi form and order 5 mock theta function:
\begin{gather}
\begin{split}
 \psi^{30+6,10,15}_1(q^{-1}) &= {q^{-\frac{1}{120}}}\left( 2- \sum_{n\geq 0}{ q^{n} \over (q^{n+1};q)_{n}} \right) = -h^{30+6,10,15}_1(\tau) \\ &=  {q^{-\frac{1}{120}}}(2-\chi_0(q))  = -{q^{-\frac{1}{120}}} \left(-1+q+q^2+2q^3+O(q^4)\right)
 \end{split}
\end{gather}

It is for this case that the relation between false theta functions and WRT invariants was first discussed by Lawrence and Zagier in \cite{lawrence1999modular}, where they also noted the relation to the mock theta function $\chi_0$.

\subsubsection{Example: $\Sigma(2,3,7)$ and the order seven mock theta function $F_0$}

As we have discussed in \S\ref{subsec:brieskorn}, the  false theta function
\be
\Psi^{42+6,14,21}_1(\t) =( \Psi_{42,1}- \Psi_{42,13}- \Psi_{42,29}+ \Psi_{42,41})(\t) = q^{1/168}(1-q-q^5+O(q^{10}))
\ee
plays the role of homological blocks for the homology sphere $\Sigma(2,3,7)$.

Note that this false theta function admits the expression
  \be
\Psi^{42+6,14,21}_1(\t) = \psi^{42+6,14,21}_1(q) ,~~  \psi^{42+6,14,21}_1(q)=  {q^{1\over 168}}   \sum_{n\geq 0}{(-1)^n q^{n(n+1)\over 2} \over (q^{n+1};q)_{n}} 
  \ee
   for $|q|<1$.
As before, the series ${\psi}^{42+6,14,21}_1$ converges both for $|q|>1$ and $|q|<1$ and has the following relation to the optimal mock Jacobi form and order 7 mock theta function:
\begin{gather}
\begin{split}
 \psi^{42+6,14,21}_1(q^{-1}) &= {q^{-\frac{1}{168}}} \sum_{n\geq 0}{ q^{n^2} \over (q^{n+1};q)_{n}} = -H^{42+6,14,21}_1(\tau) \\ &=  {q^{-\frac{1}{168}}}\,F_0(q)  = -{q^{-\frac{1}{168}}} \left(1+q+q^3+q^4+O(q^5)\right)
 \end{split}
\end{gather}
{}From the argument in \S\ref{subsec:otherphysics}, we hence propose that the mock theta function is the homological block of the three-manifold obtained from $\Sigma(2,3,7)$ via the orientation reversal.

\subsubsection{The ambiguity}

The above examples illustrate an explicit relation between false and mock theta functions when going between the upper- and the lower-half plane. One should however be cautious about the applicability and the ambiguity of the treatment.

First, this treatment depends on the existence of an expression of false/mock theta function as a $q$-hypergeometric series. In many  cases interesting for us, such an expression is not available. Moreover, sometimes more than one such expressions exist and they might have different extension outside the unit disk.
Such examples abound. See for instance \cite{MR3141412} where the Rogers-Fine false theta functions are extended to the other side in a specific way which leads to mock forms that sometimes differ from what other methods discussed in \S\ref{sec:QMF}-\ref{subsec:Rad} give. 
We will now explain one explicit example in details to illustrate the ambiguities.

The relation between the order three mock theta function $f(q)$ and other mock theta function, inside and outside the unit disc, has been  studied in details in  \cite{BrFoRh}, which we follow here.
First we have seen in \eq{6p21}-\eq{6p22}  that the hypergeometric series ${\psi}^{6+2}_1$ satisfies
\be
{\psi}^{6+2}_1(q) ={\Psi}^{6+2}_1(q) ~{\rm and}~   {\psi}^{6+2}_1(q^{-1})  = {q^{-{1\over 24}}\over 2} f(q)
\ee
for $|q|<1$.

Now, define other two hypergeometric series
\begin{gather}
\begin{split}
{\psi}'(q) &={q^{1/24}\over 2}(1+ \sum_{n\geq 1}{q^n \over (-q;q)_n^2})\\
{\psi}''(q) &=q^{1/24} \sum_{n\geq0}{q^n \over (-q^2;q^2)_n}.
\end{split}
\end{gather}
One can easily check that they too are defined both inside and outside the unit disk.
It turns out that they are related to ${\psi}^{6+2}_1(q) $ in a very interesting way. To describe the relation, we need to introduce two more functions.
The first is a modular form given by
\be
T(\t) := {\eta^7(2\t)\over \eta^3(\t)\eta^3(4\t)} ,
\ee
and the second is the ratio of a false theta function and a modular form
\be
S(\t) := {1\over \eta^2(\t)} \Psi_{2,1}(\t) .
\ee
These functions are related via
\begin{gather}\label{ambiguity_example}\begin{split}
{\psi}^{6+2}_1(q) &= {\psi}'(q)  + {1\over 2} S(\t) ={\psi}''(q) \\
{\psi}^{6+2}_1(q^{-1}) &= {\psi}'(q^{-1})  ={\psi}''(q^{-1}) - {1\over 2}T(\t+1/2) \\
\end{split}
\end{gather}
for $|q|<1 \Leftrightarrow \tau\in \HH$.
In other words, the two $q$-hypergeometric series which are the same in the upper-half plane might extend to different functions in the lower-half plane, and vice versa.
In the following two subsections we will see a more systematic way of describing and understanding the relation between mock and false theta functions, as well as quantum modular forms.

\subsection{False, mock, and quantum}
\label{sec:QMF}

In \S\ref{sec:mod} and \S\ref{sec:examples} we have seen the role of false theta functions in describing the homological blocks associated to
certain three-manifolds.
In the previous subsections we have seen hints that, when considering the superconformal indices by venturing to the lower-half plane, mock theta functions are likely to play an important role for the related three-manifolds.
In fact, despite their very different appearances and modular behaviors, false and mock theta functions both share the structure of the so-called quantum modular forms \cite{MR2757599}. See also the Ch 21 of \cite{MR3729259} for a recent account.

We propose that the (strong) quantum modularity of the false and mock theta functions is in fact what makes them relevant for three-manifolds and homological blocks. Moreover, we propose that going to the other side of the plane in the current context turns a false theta into a mock theta, such that the false--mock pair corresponds to the same quantum modular form.

To explain these ideas, we will start by recalling the definitions of mock modular forms and quantum modular forms.
The definition for quantum modular forms is  purposely a little vague in order to encompass the different types of examples with slightly different properties \cite{MR2757599}. It states:
\begin{defn}
\label{def:qmf}
\cite{MR2757599}
A {\em quantum modular} form of weight $k$ and multiplier $\chi$ on $\Gamma$ is a function $Q$ on $\QQ$ such that for every $\gamma\in \Gamma$ the function $p_\gamma: \QQ\backslash\{
\gamma^{-1}\infty\}\to \CC$, defined by
\be \label{def:period}p_\gamma (x):=Q(x) - Q\lvert_{k,\chi} \gamma(x)\ee (the ``period function'') has some property of continuity or analyticity for every $\gamma\in \Gamma$. Moreover, we say that $Q$ is
 {\em strong quantum modular} if it has formal power series attached to each rational number so that \eq{def:period} holds as an identity between countable collations of formal power series.
\end{defn}

In the above we have used the {\em slash operator} for weight $k$ and multiplier $\chi$ on $\Gamma$, acting on the space of holomorphic functions on the upper-half plane and defined as
\be\label{def:slash}
f(\tau)\lvert_{k,\chi} \gamma = f\left({a\t+b\over c\t+d}\right)  \chi(\gamma) (c\t+d)^{-k}
\ee
where we wrote $\gamma=\left( \begin{smallmatrix} a&b\\c&d\end{smallmatrix}\right)\in \Gamma$.

In fact, the Eichler integrals we encountered in \S\ref{subsec:partial} are examples of quantum modular forms.
To explain this, define the non-holomorphic Eichler integral $\tilde g^{\ast}:  \HH^- \to \CC$
\be
\tilde g^{\ast} (z) := C \int_{\bar z}^{i\infty} g(z' ) (z' -z)^{w-2} dz'
\ee
of a weight $w$ cusp form $g$ with multiplier $\chi$, where the
constant $C$ is the same as in the definition of the Eichler integral \eq{def:Eichler_int1}.
For the purpose of the present article, we can restrict to the  cusp forms $g$ with real coefficients, namely $\overline{g(-\bar \t)} = g(\t).$
Note that $\tilde g^{\ast} (z) $ has nice transformation property while $\tilde g(\t) $ has nice Fourier expansions.
In \cite{lawrence1999modular,BR14} it was shown that $\tilde g^{\ast}$ and the Eichler integral \eq{def:eichler} $\tilde g$ agree to infinite order at any $x\in \QQ$, in the sense that
\be\label{asymp_false}
\tilde g(x+it) \sim \sum_{n\geq 0} \a_n t^n  ~{\rm and}~ \tilde g^{\ast} (x-it) \sim \sum_{n\geq 0} \a_n (-t)^n \ee for $t>0$.
See Figure \ref{eich_quantum} for an illustration.

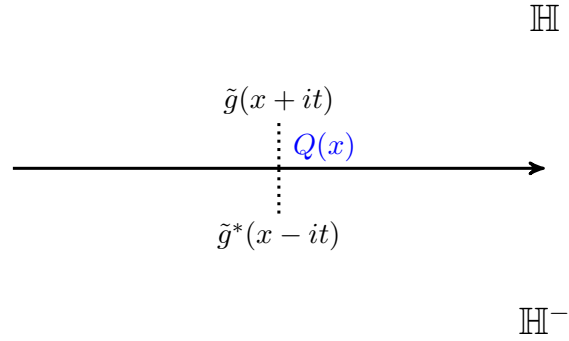
\begin{figure}[h]
\begin{center}
\begin{tikzpicture}[node distance=4cm]
\node (H) at (7, 2) {\Large $\HH$};
\node (Hm) at (7, -2) {\Large $\HH^-$};
\node (up) at (3.5, 0.9) {$\tilde g (x+it)$};
\node (down) at (3.5, -0.9) {$\tilde g^\ast (x-it)$};
\node (Q) at (4.1, 0.28) {\textcolor{blue}{$Q(x)$}};
\draw[->,line width=.4mm] (0, 0) -- (7,0) ;
\draw[dotted,line width=.4mm]  (3.5, .6) --  (3.5, 0);
\draw[dotted,line width=.4mm]  (3.5, -0.6) --  (3.5, 0);
\end{tikzpicture}
\caption{\label{eich_quantum} The upper- and lower-half planes and quantum modular forms.}
\end{center}
\end{figure}

Furthermore, it is easy to see that $\tilde g^{\ast}$ is nearly modular of weight $2-w$ in $\HH^-$, and the discrepancy is given precisely by the period function:
\be
\tilde g^{\ast}(z) - \tilde g^{\ast}\lvert_{2-w,\chi} \gamma (z) =C  \int_{\gamma^{-1}(i\infty)}^{i\infty} g(z' ) (z'-z)^{w-2} dw.
\ee
Combining the above two facts we are immediately led to the conclusion that {\bf $\til g$ is a quantum modular form of weight $2-w$} and multiplier system $\chi$. In the notation of Definition \ref{def:qmf}, the period function corresponding to $\tilde g$ is given by
\be\label{periodEich}
p_\gamma(x) = C  \int_{\gamma^{-1}(i\infty)}^{i\infty} g(z' ) (z'-x)^{w-2} dw
\ee
and is a smooth function on $\RR$ except for $x=\gamma^{-1}(i\infty)$ and has an analytic extension to $\{u+iv\lvert \, ~u>0\, ~{\rm or}\, v>0\}$ (cf. Lemma 3.3 in \cite{CLR}).
In particular, the false theta functions $\Psi_{m,r}$, arising from taking the cusp form $g$ to be given by the weight 3/2 unary theta functions $\theta^1_{m,r}$, are quantum modular forms of weight 1/2.

Soon we will see that mock modular forms produce examples of quantum modular forms.
In his last letter to Hardy in 1920, Ramanujan constructed 17 examples of what he called mock theta functions and claimed that they have a few striking properties regarding their behavior near the roots of unity. Ramanujan did not give a definition for mock theta functions, but stated that they should be a $q$-series that converges for $|q|<1$ that have the following properties
\begin{enumerate}\label{prop:mock1}
\item infinitely many roots of unity are exponential singularities,
\item for every root of unity $\xi$ there is a modular form $f_\xi(q)$ such that the difference $f-f_\xi$ is
bounded as $q\to \xi$ radially,\label{prop:mock3}
\item $f$ is not the sum of two functions, one of which is a modular form and the other a function
which is bounded radially toward all roots of unity.
\end{enumerate}

The long search for a definition of mock modular forms, which would place mock theta functions in the context of modular forms,  ended with the PhD thesis of Zwegers \cite{Zwegers}, where he gave mock modular forms a definition, which basically states that they can be viewed as the holomorphic part of certain harmonic Maass forms.
Moreover, the other, non-holomorphic, part of the harmonic Maass form is given by a modular form, called the {\em shadow} of the mock modular form. Since we have a specific application in mind and in order to simplify the discussion, in the following definition we restrict to mock modular forms whose shadows are cusp forms. The generalization is standard and straightforward.

\begin{defn}\label{def:mock}
We say that a holomorphic function $f$ on $\HH$ is a {\em mock modular form} of weight $k$ and multiplier $\chi$ on $\Gamma$, if and only if it exists a weight $2-k$ cusp form $g$ on $\Gamma$
such that the non-holomorphic {\em completion} of $f$, defined as
$$\hat f (\t) = f(\tau) - g^\ast (\tau) $$
satisfies $\hat f = \hat f\lvert_{k,\chi} \gamma$ for every $\gamma\in \Gamma$.
In the above, we defined the non-holomorphic Eichler integral
\be
g^\ast(\t) 
:=
 { C\int_{-\bar \t}^{i\infty} (\t' + \t)^{-k} \overline{g(-\bar\t')} \, d\t' }
\ee
for $\t\in \HH$.
\end{defn}
Note that there is no canonical normalization of the shadow and  we  choose ours to simplify the comparison between mock modular forms and Eichler integrals.
For convenience, we will denote by ${\mathbb M}_{k,\chi}(\Gamma)$, $M^{!}_{k,\chi}(\Gamma)$, $S_{k,\chi}(\Gamma)$, $Q_{k,\chi}(\Gamma)$ the spaces of mock modular, weakly holomorphic modular,  cusp and quantum modular forms respectively, of weight $k\in {1\over 2} \ZZ$ and multiplier $\chi$ for the group $\Gamma < \SL(2,\RR)$.  In the present article we will mainly encounter the cases $\Gamma=\SL(2,\ZZ)$ and  $\Gamma=\Gamma_0(N)$, the congruence subgroup of $\SL(2,\ZZ)$ with the congruence condition $N|c$.
We will also define the {\em shadow map} $\xi: {\mathbb M}_{k,\chi}(\Gamma) \to S_{2-k,\bar\chi}(\Gamma)$ by letting $\xi(f) = g$ in the notation of Definition \ref{def:mock}.

In what follows we will see a relation between the above modern definition of mock modular forms and the characterizations  Ramanujan gave in his letter, and how mock modular forms lead to quantum modular forms in a way that is closely related to the case of Eichler integrals discussed above. We will follow the work by Choi--Lim--Rhoades \cite{CLR} quite closely in this part of the discussion.  

To show that mock theta functions do have the above-mentioned properties that Ramanujan claimed, the following  was proven recently. 
\begin{theorem}\label{thm_mock_Ram} \cite{MR3065809, CLR}
If $f\in{\mathbb M}_{k,\chi}(\Gamma_0(N))$  such that it has non-vanishing shadow, and $\Gamma_0(N)$ has $t$ inequivalent cusps, $\{q_1,\dots, q_t\}\subset \QQ\cup \{i\infty\}$. Then\vspace{-5pt}
\begin{enumerate}
\item The function $f(\t)$ has exponential singularities at infinitely many rational numbers, \vspace{-7pt}
\item for every $G\in M^{!}_{k,\chi}(\Gamma_0(N))$, $f-G$ has exponential singularities at infinitely many rational numbers, \vspace{-7pt}
\item there is a collection $\{G_j\}_{j=1}^{t}$ of weakly holomorphic modular forms such that $f-G_j$ is bounded towards all cusps equivalent to $q_j$.
\end{enumerate}
\end{theorem}

A famous example of the above is the third order mock theta function of Ramanujan that we have encountered in
\eq{6p22}-\eq{6p22opt}. Ramanujan's observation, written in terms of the mock modular form $h^{6+2}_1(\tau)$, states  that
\be\label{6p2oddroot}
\lim_{\tau \to \zeta} h^{6+2}_1(\tau) = O(1)
\ee
for all roots of unity $\ex(\zeta)$ of odd order (such as $q\to 1$), and
\be\label{6p2evenroot}
\lim_{\tau\to \zeta} (h^{6+2}_1 + (-1)^k b(\tau)) = O(1)
\ee
for all order $2k$ roots of unity $\ex(\zeta)$, with the modular form subtraction  given by $b(\tau) = \frac{\eta^3(\t)}{\eta^2(2\t)}$.

Moreover, after the modular subtraction the asymptotic expansion of the mock modular form near  a specific cusp is the same (up to a minus sign) as that of the modular correction:
\begin{lemma}\label{lem:asymp}
In the notation of Theorem \ref{thm_mock_Ram}, we have the following equality among asymptotic series:
\be\label{asymp_mock}
(f-G_x)(x+it) \sim \sum_{n\geq 0}\beta_n t^n ~{\rm and}~g^\ast(x+it)\sim \sum_{n\geq 0}\beta_n t^n.
\ee
\end{lemma}
\begin{proof}
The equality among the limiting values is shown in the Lemma 3.1 of \cite{CLR} using  the fact that $\hat f-G_x$ is a harmonic Maass form and expand it near the cusp $\tau \to x$. The same method gives the above equality among the asymptotic series.
\end{proof}

Given a choice of $\{G_j\}_{j=1}^{t}$, one define $Q_f: \QQ\to \CC$ by setting
\[Q_f (x):=\lim_{t\to 0^+} (f-G_x)(x+it),\]
where we write $G_x = G_j$ when $x$ is equivalent to $q_j$ under the action of $\Gamma=\Gamma_0(N).$
The Lemma \ref{lem:asymp}, the analyticity property of the period function \eq{periodEich} associated to $\tilde g^\ast$ and the fact that $g^\ast(\t) = \tilde g^\ast(-\t)$ (in the cases we care about where $\overline{g(-\bar \t)} = g(\t)$) immediately shows that the mock modular form gives rise to a (strong) quantum modular form $Q_f$.

Note that the choice of the {\em modular subtraction} $\{G_j\}_{j=1}^{t}$ with which to carve out the singularities of the mock modular forms is {not unique}. At present, a satisfactory systematic study of the possibilities and their properties  is not yet available.
For a family of mock modular forms, namely those with known expressions in terms of the so-called
 universal mock modular forms $g_2$ and $g_3$, specific choices are given explicitly in \cite{MR3589294,MR3392042}.
Given this lack of uniqueness of the modular subtractions, it is important to note that
 the limiting value and the asymptotic expansion \eq{asymp_mock} is {independent} of the choices of the modular form $G_j$ as long as they do subtract the singularity.

To sum up, given a cusp form $g\in S_{2-k,\bar\chi}(\Gamma)$, if $f$ is a mock modular form $f\in{\mathbb M}_{k,\chi}(\Gamma)$ with shadow $\xi(f)=g$ and $\tilde g$ is its Eichler integral, then $f$ and $\tilde g$ have the same limiting value at $x\in \QQ$ in the sense that
\be\label{quantum_mock}
\lim_{t\to 0^+} (f-G_x)(x+it) = \lim_{t\to 0^+} \tilde g^\ast(-x+it).
\ee
Moreover, they also have the same asymptotic series; in terms of the asymptotic series \eq{asymp_false} and \eq{asymp_mock} we have $\a_x(n) =(-1)^n \beta_{-x}(n)$ and we have
\be
 (f-G_x)(-x+it) \sim  \sum_{n\geq 0} \alpha_{x}(n) (-t)^n ~{\rm and}~\tilde g(x+it)\sim \sum_{n\geq 0}
 \alpha_{x}(n) t^n.
 \ee
In particular, at cusp 0 we have the ``same'' asymptotic series, approaching from the upper- and lower-half plane, in the sense that:
\be \label{expansionzero}
 (f-G_0)(it) \sim  \sum_{n\geq 0} \alpha_{0}(n) (-t)^n ~{\rm and}~\tilde g(it)\sim \sum_{n\geq 0}
 \alpha_{0}(n) t^n .
\ee

Note that the relations  \eq{asymp_mock} among  the asymptotic expansion relations are precisely what we need to make contact with the homological blocks of the three-manifold: the former states that the limiting value at $\tau \to {1\over k}$ respectively $-{1\over k}$ coincide which is what we need to obtain the expected relations among the WRT invariants of $M_3$ and $-{M_3}$, and the latter gives the expected relation among Ohtsuki series.

We summarize the relation between these objects in Figure \ref{diagram_mockquantumfalse}.  Note that the $q\leftrightarrow q^{-1}$ line between mock and Eichler integral of its shadow is in dashed line, since the $q\leftrightarrow q^{-1}$ procedure is non-unique in both directions. This is clear from the fact that the asymptotic expansion only depends on the shadow of the mock modular form, and hence the map from mock to quantum modular forms is in fact a  linear injective map
\be\label{map:mockquantum}
\mu : {\mathbb M}_{k,\chi}(\Gamma)/M^{!}_{k,\chi}(\Gamma) \to Q_{k,\chi}(\Gamma),\ee given by $\mu(f) = Q_f$. Relatedly, it is insensitive to the choice of the  modular subtraction $\{G_j\}_{j=1}^{t}$.

Finally, from the discussions in \S\ref{sec:mod}--\ref{sec:resurgence} it is easy to see that in our context we need a vector-valued version of the above discussion, for the full modular group $\SL(2,\ZZ)$. It should be straightforward to generalize the existing discussion to the vector-valued situation and we leave the details for future work.

\tikzstyle{block} = [rectangle, draw,
    text width=9em, text centered, minimum height=4em]
    \tikzstyle{block2} = [rectangle,
    text width=15em, text centered, minimum height=4em]
\tikzstyle{line} = [draw, -latex']
\tikzstyle{smallblock} = [rectangle, draw,
    text width=6em, text centered, minimum height=4em]
\tikzstyle{cloud} = [draw, ellipse,fill=red!20, node distance=3cm,
    minimum height=2em]

\begin{figure}[h!]
\begin{tikzpicture}[node distance = 4cm, auto]
    \node [block] (mock) {{\small Mock Modular Form}\\ $f\in {\mathbb M}_{k,\chi}$};
    \node [smallblock, right of=mock, node distance=5cm] (shadow) {{\small Shadow}\\ $g\in S_{2-k,\bar \chi}$};
    \node [block, right of=shadow] (nonhol) {{\small Non-hol. Eichler Int.}\\ $\tilde g^\ast(z),\,z\in {\mathbb H}^-$};
    \node [block, below of=nonhol, node distance=3cm] (eichler) {{\small Eichler Int. (False $\theta$)}\\ $\tilde g(\t),\,\t\in {\mathbb H}^+$};
    \node [block, above of=nonhol, node distance=3cm] (shad) {{\small Modular Correction}\\{\small $ g^\ast(\t) ,\,\t\in {\mathbb H}^+$}};
    \node [block2, above of = shad,node distance=1.5cm] (tit) {{ \bf Quantum Modular Forms}};
    \path [line] (mock) --node {\footnotesize shadow map} (shadow);
    \path [line] (shadow) -- (eichler);
    \path [line] (shadow) -- (shad);
    \path [line] (shadow) -- (nonhol);
    \draw[<->]  (shad)  to node {$\footnotesize z=-\t$}  (nonhol);
      \draw[<->]  (nonhol)  to node {\footnotesize \begin{tabular}{c} same asymp. \\ \eq{asymp_false}\end{tabular}}  (eichler);
     \draw[<->, dashed]   (mock) --node [swap]{{$\footnotesize q\;\leftrightarrow\; q^{-1}$}
} (eichler);
\end{tikzpicture}
\caption{\small
The relation between the different modular objects involved. The dashed line is to denote that the relation is non-unique in both directions.}
\label{diagram_mockquantumfalse}
\end{figure}
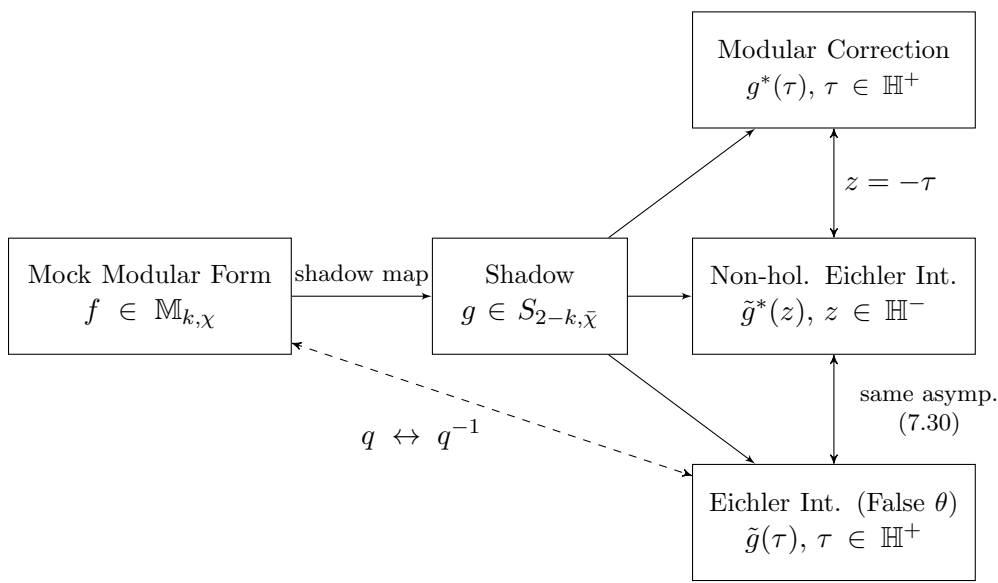

\subsection{Rademacher sums}
\label{subsec:Rad}

Apart from the $q$-hypergeometric relation discussed in \S\ref{subsec:qhyper}, another way to explicitly see how mock theta functions become false theta functions when going between upper- and lower-half planes, is via the method of Rademacher sums.
For the families of examples we are interested in in this paper, this approach is arguably  more systematic than that of the $q$-hypergeometric series.
We will explain this further in \S\ref{sec:optimal}.

As Poincar\'e pointed out, a simple way to construct modular forms is simply by averaging a quantity over its images under the modular group \cite{Poi_FtnMdlFtnFuc}.
Taking this quantity to be a monomial $q^\mu$, and for a given (compatible) multiplier system $\chi$ of weight $k$ for a group $\Gamma<\SL(2,\RR)$ that is commensurable with $\SL(2,\ZZ)$, we define the {\em Poincar\'e sum:}
\be
\label{poincaresum}
P^{[\mu]}_{\Gamma,k,\chi}(\t) := \sum_{\gamma\in \Gamma_\infty\backslash \Gamma} q^\mu \lvert_{k,\chi} \gamma,
\ee
where $\Gamma_\infty$ is the subgroup of $\Gamma$ that preserves the cusp $\{i\infty\}$, and is in general generated as $\Gamma_\infty = \langle T^h, -{\bf 1}\rangle$. The unique such positive integer $h$ is called the width of the cusp $i\infty$ of the group $\Gamma$.
A choice of $\mu$ is compatible if and only if (cf. \eq{consistency_rad})
\(
q^\mu \lvert_{k,\chi} \gamma= q^\mu\) for all  $\gamma\in \Gamma_\infty$.
We are mainly interested in the special case  $\Gamma = \SL(2,\ZZ)$. In this case $\Gamma_\infty=\langle T, -{\bf 1}\rangle$ and the choice of $\mu$ is consistent if and only if $\chi(T) \ex(\mu)=1$.
For $k>2$ the sum \eq{poincaresum} converges absolutely and $P^{[\mu]}_{\Gamma,k,\chi}$ is indeed a holomorphic function on $\HH$ which is moreover a modular form of weight $k$ and multiplier $\chi$ by construction.

For $k\leq 2$, which is the range of interest for us, the sum is no longer absolutely convergent and one needs to regularise the Poincar\'e sum, which leads to what is often known as the {\em Rademacher sums}.  See \cite{Cheng:2012qc} for a review.
First, the sum can no longer be taken over the full coset $\Gamma_\infty\backslash \Gamma$ and we consider instead the following subset
\begin{gather}
\Gamma_{K,K^2}=\left\{\begin{pmatrix}a&b\\c&d\end{pmatrix}\in\Gamma\Big\lvert |c|<K,\,|d|<K^2\right\}.
\end{gather}
Such an adjustment  of the range of sums is sufficient for $k=2$, but for $k<2$ we also need to introduce an additional regularisation factor
$$
\Re^{[\mu]}_k(\gamma, \tau):= \frac{\bar{\gamma}(1-w, 2\pi i n (\gamma \tau - \gamma \infty))}{\Gamma(1-w)}
$$
where $\bar{\gamma}$ denotes the lower incomplete gamma function
\be
\bar{\gamma}(s,x) = \int_0^x t^{s-1} e^{-t} dt.
\ee
Using the above, we define the Rademacher sum, associated to the data $(\Gamma, k,\chi,\mu)$ determining the group, the weight, the multiplier and the seed (or polar part when $\mu<0$) of the sum, to be\footnote{In the special case of $\chi(T)=1$, which we do not  encounter in the present work, an additional constant should be added to this sum. The same comment also applies to \eq{eqn:series:RadFou}. See \cite{Cheng:2012qc} for details. In what follows we will not consider such cases.}
\be
 \label{eqn:rad}
{R}^{[\mu]}_{\Gamma, k,\chi}(\t):=  \lim_{K\to \infty} \sum\limits_{\gamma \in \Gamma_{\infty}\backslash \Gamma_{K,K^2}} \Re^{[\mu]}_k(\gamma, \tau)  \left(q^\mu \lvert_{k,\chi} \gamma\right). \quad
\ee
Niebur proved that in the case of negative weight the above construction  gives rise to a conditionally convergent series, which he referred to as automorphic integral \cite{Niebur1}.
Clearly, after regularization there is no guarantee that the sum will still be a modular form.
It turns out that in general the Rademacher sum \eq{eqn:rad}  is a  mock modular form  with  a shadow given by a cusp form. Moreover, the shadow of the Rademacher sum  is itself a Rademacher sum:
\be\label{shadow_rad}
\xi\big({R}^{[\mu]}_{\Gamma, k,\chi}\big) = (-\mu)^{1-k} {R}^{[\mu]}_{\Gamma, 2-k,\bar\chi},
\ee
now with the dual weight and the conjugate multiplier system (cf. Definition \ref{def:mock}).

This technique was extended to differenet weights, multiplier systems and modular groups by  \cite{Knopp1, Knopp2, Niebur1, DunFre} and later to weight $1/2$ mock modular forms in \cite{Pribit,Cheng2011}. Further developments in the context of harmonic Maass forms are reported in \cite{{BO06},{BO10}}.
As such, Rademacher sum construction can be viewed as a useful tool to construct mock modular forms.
It can also be generalised to the vector-valued cases, which are relevant for the application discussed in the present paper, as was done in \cite{Whalen}.

After massaging the sum in \eq{eqn:rad}, one can recast the Rademacher sum as a $q$-series
\begin{gather}\label{eqn:series:RadFou}
R^{[\m]}_{\Gamma,k,\chi}(\t)=q^{\m}+\sum_{\substack{h(\nu-\mu)\in\ZZ\\\nu> 0}}c_{\Gamma,k,\chi}(\m,\nu)q^{\nu}
\end{gather}
 and obtain an explicit expression for its Fourier coefficients $c_{\Gamma,k,\chi}(\m,\nu)$, sometimes referred to as the Rademacher series.

To write down this expression, we define the subset
\begin{gather}\label{eqn:sums:GKdef}
\Gamma^{\times}_{K} = \left\{\begin{pmatrix}a&b\\c&d\end{pmatrix}\in\Gamma\mid 0<|c|<K\right\},
\end{gather}
and the functions
\begin{gather}
	K_{\gamma,\chi}(\m,\n)=\label{eqn:series:Kdef}
	\ex\left( \m\frac{a}{c}\right)\ex\left( \n\frac{d}{c}\right)\chi(\gamma),\\
	B_{\gamma,k}(\m,\n)=\label{eqn:series:Bdef}
	\begin{cases}
	\ex\left(-\frac{k}{4}\right)\sum_{ n\geq 0}\left(\frac{2\pi}{c}\right)^{2n+k}
	\frac{(-\m)^{n}}{n!}\frac{\n^{k+n-1}}{\Gamma(k+n)},&k\geq 1,\\
	\ex\left(-\frac{k}{4}\right)\sum_{n\geq 0}\left(\frac{2\pi}{c}\right)^{2n+2-k}
	\frac{(-\m)^{n+1-k}}{\Gamma(n+2-k)}\frac{\n^{n}}{n!},&k\leq 1.
	\end{cases}	
\end{gather}
In terms of these we have
\begin{gather}\label{eqn:series:cdef}
c_{\Gamma,k,\chi}(\m,\n)
=\frac{1}{h}\lim_{K\to \inf}\sum_{\gamma\in\Gamma_{\inf}\backslash\Gamma^{\times}_K/\Gamma_{\inf}}
K_{\gamma,\chi}(\m,\n)B_{\gamma,k}(\m,\n)
\end{gather}
defined for
\be\label{consistency_rad}
(\m,\n) \in \frac{1}{h}\ZZ\times\frac{1}{h}\ZZ-\left(\frac{\bar\nu}{h},\frac{\bar\nu}{h}\right)\ee
where $\chi(T^h) = \ex(\bar\nu)$. For later use we choose the branch $0<\bar\nu<1$.

For the case $\Gamma=\SL(2,\ZZ)$ and $\m<0$ we have
\be\label{radseries_bessel}
c_{\Gamma,k,\chi}(\m,\n)
	=\sum\limits_{c>0}\, \mathfrak{s}_{\Gamma,\chi}(\mu,\nu,c)\, \frac{2\pi}{c}\biggl( -\frac{\nu}{\mu}\biggr)^{\frac{k-1}{2}}I_{|1-k|}\Bigl(\frac{4\pi }{c}\sqrt{-\mu\nu}\Bigr)
\ee
where
\(\mathfrak{s}_{\Gamma,\chi}(\mu,\nu,c)\) is the Kloosterman sum
\be
\mathfrak{s}_{\Gamma,\chi}(\mu,\nu,c) = \sum_{\gamma \in \Gamma_{\infty}\backslash \Gamma / \Gamma_{\infty}} \hspace{-1.9mm}\ex\Bigl(\m\frac{a}{c} +\n\frac{d}{c}\Bigr)\chi(\gamma)
 \ee
where the sum is over the $\gamma = (\begin{smallmatrix}* &*\\c&*\end{smallmatrix})$ and we write the double coset representative as $\gamma = (\begin{smallmatrix} a & b\\c&d\end{smallmatrix})$. One can easily check that
the compatibility condition guarantees that the summand is independent of the choice of representative. More explicitly,  one can write the above as
\be
\mathfrak{s}_{\Gamma,\chi}(\mu,\nu,c) = \sum_{\substack{0\leq d < c\\(c,d)=1}} \ex\Bigl(\m\frac{a}{c} +\n\frac{d}{c}\Bigr)\chi((\begin{smallmatrix} a & b\\c&d\end{smallmatrix})),
\ee
where in each term in the sum we choose any $(a,b)$ such that $(\begin{smallmatrix} a & b\\c&d\end{smallmatrix})\in \SL(2,\ZZ)$, and the summand is again independent of the choice.

The first way to see the relation between Eichler integral of the shadow and the Rademacher sum performed in the lower-half plane, is by noting the following relations among the Rademacher series \cite{Cheng2011}.
The so-called  {\em Eichler duality} states
\begin{gather}\label{eqn:series:Eicdual}
	-\overline{c_{\Gamma,k,\chi}(-\m,-\n)\m^{1-k}}=c_{\Gamma,2-k,\bar\chi}(\m,\n)\n^{1-k}.
\end{gather}
Together with the so-called {\em Zagier duality} relating the Rademacher sums of dual weights
\be
\label{eqn:series:dual}
c_{\Gamma,2-k,\bar\chi}(-\nu,-\mu) = c_{\Gamma,\chi,k}(\m,\n),
\ee
we have the following expression for the Fourier coefficients of the mock modular form $f= {R}^{[\mu]}_{\Gamma, k,\chi}$, its shadow $g= \xi({R}^{[\mu]}_{\Gamma, k,\chi})$ \eq{shadow_rad}, and the Eichler integral $\tilde g$:
\begin{align}\label{fFour}
f(\t) & = q^{\m}+\sum_{\substack{h(\nu-\mu)\in\ZZ\\\nu> 0}}C(\m,\nu)\,q^{\nu} \\
g(\t) &=(-\mu)^{1-k} \Big(q^{-\m}+\sum_{\substack{h(\nu'+\mu)\in\ZZ\\\nu> 0}} C(-\nu',\mu)\,q^{\nu'} \Big)   \\
\label{EichFour}
\tilde g (\t) &= q^{-\m}+\sum_{\substack{h(\nu'+\mu)\in\ZZ\\\nu> 0}} \overline{C(\mu,-\nu')}\,q^{\nu'}
\end{align}
 in terms of the Rademacher series $C(\mu,\nu):=c_{\Gamma,\chi,k}(\m,\n)$ \cite{Cheng:2012qc}.
Compare \eq{fFour} and \eq{EichFour}, and focus on the case where the mock modular form (such as mock theta functions) have real coefficients $C(\mu,\nu)=\overline{C(\mu,\nu)}$,
we see that $f$ and $\tilde g$ can be viewed as being related by $q\leftrightarrow q^{-1}$, as depicted in  Fig \ref{diagram_mockquantumfalse}.

The above relation between the mock modular form $f$ and the Eichler integral $\tilde g$ via $q\leftrightarrow q^{-1}$ can be seen even more explicitly by manipulating the Rademacher sum itself.
The question of how to extend the Rademacher series to the lower-half  plane was first discussed by Rademacher in \cite{RademFourier}.
This analysis was later reviewed and extended to the context of harmonic Maass forms and mock modular forms by Rhoades \cite{Rhoades}, in the special case of weight $k=1/2$. In what follows we will follow his treatment and show that one can define a function convergent both in the upper and the lower half-plane which coincides with the mock modular form in the upper half-plane and the Eichler integral of its shadow in the lower half-plane.
Given the convergence of the weight 1/2 Rademacher sums proven in \cite{Cheng:2011ay}, we have the following theorem, generalising the result of \cite{Rhoades}.

\begin{theorem}\label{defined_bothupdown}
Let $f(\t)$ be a mock modular form of weight 1/2 defined by the Rademacher sum $R^{[\mu]}_{\Gamma,1/2,\chi}(\t)$, for $\Gamma=\Gamma_0(N)$ with some positive integer $N$. Then there exists a function $F(\tau)$ on $\HH$ and $\HH^-$, satisfying
\be
F(\t) =\begin{cases}f(\t) &{\rm{when}}~ \t\in \HH \\  \til g(-\t) &{\rm{when}}~ \t\in \HH^-.\end{cases}
\ee
in the notation of \eq{fFour}--\eq{EichFour}.
\end{theorem}

\begin{proof}
To construct $F$, let us start with the Rademacher sum which defines the mock modular form $f$. For convenience we will focus on the case $\Gamma=\SL(2,\ZZ)$. The generalization to $\Gamma=\Gamma_0(N)$ with $N>1$ is straightforward.
Using  \eq{radseries_bessel} and the following integral expression for the Bessel function (cf. Lemma 3.1 of \cite{Rhoades})
\be\label{Bessel_integral}
t^{-1/4} I_{1/2}\big({4\p\over k}\sqrt{t}\big) = \oint_{|s|=\epsilon} {ds\over 2\pi i}\, e^{st}\sum_{m\geq 0} {({4\p\over k})^{2m+{1\over2}} \over \Gamma(m+{3\over 2})}{1\over s^{m+1}},
\ee
where $\epsilon$ may be taken to be arbitrarily small, we obtain the following expression for $f := {R}^{[\mu]}_{\SL(2,\ZZ), {1\over 2},\chi}$:
\begin{gather}
\begin{split}
f(\t) &=  q^\mu
+ {(-\m)^{-1/2} \over 2}  \sum_{\substack{\n-\m\in \ZZ\\\n >0}}  \sum_{c>0}\sum_{0\leq d < c} q^\nu \ex\big(\nu{d\over c}+\m \frac{a}{c}\big) \,\chi((\begin{smallmatrix} a & b\\c&d\end{smallmatrix}))\\& \qquad\qquad\times \oint_{|s|=\epsilon} {ds\over 2\pi i}\, e^{-s\mu\nu}\sum_{m\geq 0} {({4\p\over c})^{2m+{3\over2}} \over \Gamma(m+{3\over 2})}{1\over s^{m+1}}.
\end{split}
\end{gather}
The proof of the convergence of the above sum is the same as in \cite{Cheng:2011ay}. Since the sum over $m$ is absolutely convergent, we can switch the order of the integral and the sum and obtain the succinct expression
\be
f(\t)  = q^\mu  + \oint_{|s|=\epsilon} {ds\over 2\pi i}  \sum_{c>0}\sum_{0\leq d < c} f_{c,d;+}(\t,s) G_{c,d}(s)
\ee
where
\be
G_{c,d}(s) = {(-\m)^{-1/2} \over 2}  \ex\big(\m \frac{a}{c}\big) \,\chi((\begin{smallmatrix} a & b\\c&d\end{smallmatrix})) \sum_{m\geq 0} {({4\p\over c})^{2m+{3\over2}} \over \Gamma(m+{3\over 2})}{1\over s^{m+1}}
\ee
is a $\tau$-independent factor
and
\begin{align}
f_{c,d;+}(\t,s) =\sum_{\substack{\n-\m\in \ZZ\\\n >0}}   q^\nu \ex\big(\nu{d\over c}\big) e^{-s\mu\nu}
\end{align}
captures the summation over $\nu$.
Note that $f_{c,d;+}$ is a geometric series. Consequently, if we define
\be
f_{c,d}(\t,s)  = {q^{\bar \nu} \ex\big(\bar\nu{d\over c}\big)e^{-s\mu\bar\nu}\over {1-q \ex\big(\bar{d\over c}\big)e^{-s\mu}} }
\ee
where $\bar\nu$ is as in \eq{consistency_rad},
we have
\be
f_{c,d}(\t,s)  = \begin{cases} f_{c,d;+}(\t,s) ,& \mid q e^{ -\m s} \mid <1 \\f_{c,d;-}(\t,s) ,& \mid q e^{ -\m s} \mid >1 \end{cases}
\ee
where
\be
f_{c,d;-}(\t,s) =-\sum_{\substack{\n'+\m\in \ZZ\\\n' >0}}   q^{-\nu'} \ex\big(-\nu'{d\over c}\big) e^{s\mu\nu'} .
\ee

Using the above, one finally shows that
\be
F(\t)  := q^\mu  + \oint_{|s|=\epsilon} {ds\over 2\pi i}  \sum_{c>0}\sum_{0\leq d < c} f_{c,d}(\t,s) G_{c,d}(s)
\ee
converges both for $\t\in \HH$ and $\t\in \HH^-$ (cf. Theorem 1.1 of \cite{Rhoades}). Moreover, plugging in $f_{c,d}(\t,s) =f_{c,d;-}(\t,s)$ in the lower-half plane and again using the integral expression for the Bessel function \eq{Bessel_integral}, we obtain the key statement of Theorem \ref{defined_bothupdown}.
\end{proof}

The content of the above manipulation is technically equivalent to the relations \eq{eqn:series:Eicdual} and \eq{eqn:series:dual} among the Rademacher series, but further highlights the fact that Rademacher sums lead to a natural definition of functions defined both on the upper- and lower-half plane.

We finish this subsection with some remarks.
\begin{itemize}
\item Given a false theta function, the Rademacher sum formalism does not determine a unique mock modular form as its companion in the other side of the plane. This is because the shadow map has a large kernel: the addition of a modular form to a mock modular form does not change its shadow. Since the Eichler integral, arising as Rademacher sums performed on the other side of the plane, only depends on the shadow of the mock modular form, Rademacher sums with the same shadow are extended to the same function on the other side of the plane. In other words, there can be many different ways to write a false theta function as Rademacher sums, corresponding to distinct mock modular forms with the same shadow. A closely related fact is that they also give rise to the same quantum modular form, as we have discussed in
\eq{map:mockquantum}.

\item For the main part of the paper, including all the examples we discuss in \S\ref{sec:examples}, we are interested in the special cases where the weight of the mock modular form is $k={1\over 2}$ and the group is $\Gamma=\SL(2,\ZZ)$. Moreover, the multiplier is that obtained from Weil representations discussed in \S\ref{sec:WeilRep}.
In this family of cases, the mock modular forms can be conveniently described in terms of mock Jacobi forms and the results of \cite{Sko_Thesis} imply that these vector-valued mock modular forms enjoy the property that they are uniquely determined by their polar part, i.e. their behaviour near the cusp $\tau \to i\infty$ (cf. \cite{omjt}).
Moreover, generically the Rademacher sums in this context give rise to $q$-series with transcendental coefficients which cannot be relevant as quantum invariants since the coefficients are supposed to count BPS states.  This shows that despite ambiguity one should be able to use physical and topological criteria to search for the relevant mock modular forms.

\item
The Rademacher sums discussed here have a natural interpretation in the physical setup \eqref{BlockD2S1}
illustrated in Figure~\ref{fig:halfindex}.
Indeed, recall that $\tau ={1\over 2\pi i} \log q$ is the complex structure of the boundary torus $T^2 \cong \partial \big( D^2\times_q S^1 \big)$.
However, unlike other physical systems where the full modular group $\Gamma = \SL(2,\ZZ)$ acts on $\tau$,
in our setup 2d $\N=(0,2)$ boundary theory enjoys $\SL(2,\ZZ)$ symmetry,
whereas 3d $\N=2$ theory is only invariant under the subgroup $\Gamma_{\infty}$
because one of the 1-cycles of $T^2 \cong \partial \big( D^2\times_q S^1 \big)$
is contractible in the combined 3d-2d systems. This explains the origin of the coset $\Gamma_\infty\backslash \Gamma$.

Moreover, the leading term in the sum (with $c=0$) that corresponds to the contribution of the cusp $i\infty$
can be interpreted as the partition function of a 1d effective quantum mechanics obtained from our 3d-2d system
in the limit $\tau \to i \infty$, in which $T^2$ is effectively stretched to a product of a ``long'' circle
and a ``short'' circle, with ratio of radia $\text{Im} \tau$.
The contribution of the other terms, with $c \ne 0$, then can be understood as the sum over KK modes along the ``short'' circle.

\item
For 39 of such cases there are ``optimal" natural choices for their polar parts \cite{omjt} and they appear prominently also in the context of 3-manifold.
As a result, they serve as examples of how Rademacher sums give rise to nice $q$-series both in the upper- and lower-half planes and will be discussed separately in \S\ref{sec:optimal}.
\end{itemize}

\subsection{The ``optimal" examples}
\label{sec:optimal}

In the previous subsections we have discussed how mock modular forms and Eichler integrals are related via $q\leftrightarrow q^{-1}$, from the point of view of $q$-hypergeometric functions, quantum modular forms, and Rademacher sums, respectively.
In this subsection we will give explicit examples of such mock--false pairs with the following desirable properties, as alluded to at the end of the previous subsection:
\begin{itemize}
\item They can be obtained as Rademacher sums in a particularly simple way, making them a perfect illustration of the principle explained in \S \ref{subsec:Rad}.
\item They appear in the three-manifold context, as illustrated in Table \ref{tab:opt1}--\ref{tab:opt2}.
\item As we mentioned before, the modular subtractions of generic mock modular forms are not unique and a totally systematic treatment is not yet available.
From the relation to the perturbative Chern--Simons, or  the Ohtsuki series, we are particularly interested in mock modular forms which are finite at $q\to 1$ (cf. \eq{expansionzero}). Due to the simple structure of the poles of the functions in our example,   it is possible to
show that some of them are finite in the limit $q\to 1$ and hence  have the same asymptotic expansions at the cusp $\tau \to 0$ on the nose. As a result, these mock modular forms are readily candidates for the quantum invariants of  $-{M_3}$. 
\end{itemize}

\begin{table}[h!]
\begin{center}\scalebox{0.85}{
{\renewcommand{\arraystretch}{1.2}
\begin{tabular}{ccccc}\toprule
$m+K$&$\sigma^{m+K}$&$M_3$& $H_1(M_3)$ & $r \in \sigma^{m+K}$\\\midrule
2&$\{1\}$&$M(-2;1/2,1/2,1/2)$&$\ZZ_2\oplus \ZZ_2$ &$r=1$\\\hline
3&$\{1,2\}$&$M(-2;1/2,1/2,2/3)$&$\ZZ_4$ &$r=1,2$\\\hline
4&$\{1,2,3\}$&$M(-2;1/2,1/2,3/4)$&$\ZZ_2\oplus \ZZ_2$ &$r=1,3$\\\hline
5&$\{1,2,3,4\}$&$M(-2;1/2,1/2,4/5)$&$\ZZ_4$ &$r=1,4$\\\hline
6&$\{1,\ldots,5\}$&$M(-2;1/2,1/2,5/6)$&$\ZZ_2\oplus \ZZ_2$ &$r=1,5$\\\hline
6+3&$\{1,3\}$&$M(-2;1/2,2/3,2/3)$&$\ZZ_3$ &$r=1,3$\\\hline
7&$\{1,\ldots,6\}$&$M(-2;1/2,1/2,6/7)$&$\ZZ_4$ &$r=1,6$\\\hline
8&$\{1,\ldots,7\}$&$M(-2;1/2,1/2,7/8)$&$\ZZ_2\oplus \ZZ_2$ &$r=1,7$\\\hline
9&$\{1,\ldots, 8\}$&$M(-2;1/2,1/2,8/9)$&$\ZZ_4$ &$r=1,8$\\\hline
10&$\{1,\ldots,9\}$&$M(-2;1/2,1/2,9/10)$&$\ZZ_2\oplus \ZZ_2$ &$r=1,9$\\\hline
\multirow{2}*{10+5}&\multirow{2}*{$\{1,3,5\}$}&$M(-1;1/2,1/5,1/5)$&$\ZZ_3$ &$r=1,5$\\
&&$M(-4;1/2,1/2,1/2)$& $\ZZ_2 \oplus \ZZ_2 \oplus \ZZ_5$ &$r=1,3,5$\\\hline
12&$\{1,\ldots,11\}$&$M(-2;1/2,1/2,11/12)$&$\ZZ_2\oplus \ZZ_2$ &$r=1,11$\\\hline
12+4&$\{1,4,5\}$&$M(-1;1/2,2/3,3/4)$&$\ZZ_2$ &$r=1,5$\\\hline
13&$\{1,\ldots,12\}$&$M(-2;1/2,1/2,12/13)$&$\ZZ_4$ &$r=1,12$\\\hline
\multirow{2}*{14+7}&\multirow{2}*{$\{1,3,5,7\}$}&$M(-5;1/2,1/2,1/2)$& $\ZZ_2 \oplus \ZZ_2 \oplus \ZZ_7$ & $r=1,3,5,7$\\
&&$M(-1;1/2,1/7,2/7)$&$\ZZ_7$ &$r=3,7$\\\hline
16&$\{1,\ldots,15\}$&$M(-2;1/2,1/2,15/16)$&$\ZZ_2\oplus \ZZ_2$ &$r=1,15$\\\hline
18&$\{1,\ldots,17\}$&$M(-2;1/2,1/2,17/18)$&$\ZZ_2\oplus \ZZ_2$ &$r=1,17$\\\hline
\multirow{3}*{18+9}&\multirow{3}*{$\{1,3,5,7\}$}&$M(-1;1/2,1/3,1/9)$&$\ZZ_3$ &$r=1,5$\\
&&$M(-2;1/2,1/3,2/3)$& $\ZZ_9$ &$r=1,3,5,7$\\
&&$M(-6;1/2,1/2,1/2)$& $\ZZ_2 \oplus \ZZ_2 \oplus \ZZ_9$ &$r=1,3,5,7$\\\hline
25&$\{1,\ldots,24\}$&$M(-2;1/2,1/2,24/25)$&$\ZZ_4$ &$r=1,24$\\\hline
\multirow{2}*{22+11}&\multirow{2}*{$\{1,3,5,7,9,11\}$}&$M(-7;1/2,1/2,1/2)$& $\ZZ_2 \oplus \ZZ_2 \oplus \ZZ_{11}$ &$r=1,3,5,7,9,11$\\
&&$M(-1;1/2,1/11,4/11)$&$\ZZ_{11}$ &$r=7,11$\\\hline
30+6,10,15&$\{1,7\}$&$\Sigma(2,3,5)$& 0& $r=1$\\\hline
\multirow{2}*{30+15}&\multirow{2}*{$\{1,3,\ldots,15\}$}&$M(-9;1/2,1/2,1/2)$& $\ZZ_2\oplus\ZZ_2\oplus\ZZ_{15}$ &$r=1,3,\ldots,15$\\
&&$M(-1;1/2,2/5,1/15)$&$\ZZ_{5}$ &$r=7,11$\\\hline
46+23&$\{1,3,\ldots,23\}$&$M(-13;1/2,1/2,1/2)$& $\ZZ_2 \oplus \ZZ_2 \oplus \ZZ_{23}$ &$r=1,3,\ldots,23$\\
\bottomrule
\end{tabular}}}\caption{\label{tab:opt1}Optimal mock Jacobi thetas  of Niemeier type and examples of the relevant 3-manifolds. }
\end{center}
\end{table}

\begin{table}[h!]
\begin{center}\scalebox{0.85}{
{\renewcommand{\arraystretch}{1.2}
\begin{tabular}{ccccc}\toprule
$m+K$&$\sigma^{m+K}$&$M_3$& $H_1(M_3)$ & $r \in \sigma^{m+K}$\\\midrule
6+2&$\{1,2,4\}$&$M(-2;1/2,1/2,1/3)$&$\ZZ_8$ &$r=1,2,4$\\\hline
10+2&$\{1,2,3,4,6,8\}$&$M(-2;1/2,1/2,3/5)$&$\ZZ_8$ &$r=1,4,6$\\\hline
\multirow{2}*{12+3}&\multirow{2}*{$\{1,2,3,5,6,9\}$}&$M(-1;1/3,1/3,1/4)$&$\ZZ_3$ &$r=1,9$\\
&&$M(-2;1/2,1/2,1/4)$& $\ZZ_2 \oplus \ZZ_2 \oplus \ZZ_3$ &$r=1,3,5,9$\\\hline
\multirow{3}*{15+5}&\multirow{3}*{$\{1,2,4,5,7,10\}$}&$M(-1;1/2,1/3,1/10)$& $\ZZ_4$& $r=1,4$\\
&&$M(-1;1/3,1/5,2/5)$&$\ZZ_5$ &$r=4,10$\\
&&$M(-3;1/2,1/2,1/3)$& $\ZZ_{20}$ &$r=1,2,4,5,10$\\\hline
18+2&$\{1,\ldots,8,10,12,14,16\}$&$M(-2;1/2,1/2,7/9)$&$\ZZ_8$ &$r=1,8,10$\\\hline
20+4&$\{1,3,4,7,8,11\}$&$M(-1;1/2,1/4,1/5)$&$\ZZ_2$ &$r=1,11$\\\hline
21+3&$\{1,\ldots,6,8,9,11,12,15,18\}$&$M(-2;1/2,1/2,4/7)$ & $\ZZ_8$ &$r=1,6,8,15$\\\hline
24+8&$\{1,2,5,7,8,13\}$&$M(-1;1/2,1/3,1/8)$&$\ZZ_2$ &$r=1,7$\\\hline
\multirow{2}*{28+7}&$\{1,2,3,5,6,7,9,$
&\multirow{2}*{$M(-1;1/4,1/7,4/7)$}&\multirow{2}*{$\ZZ_7$} &\multirow{2}*{$r=13,21$}\\
&$10,13,14,17,21\}$&&\\\hline
30+3,5,15&$\{1,3,5,7,9,15\}$&& &\\\hline
\multirow{2}*{33+11}&$\{1,2,4,5,7,8,10,$
&$M(-5;1/2,1/2,1/3)$& $\ZZ_{44}$ &$r={\rm all}$\\
&$11,13,16,19,22\}$&$M(-1;1/3,1/11,6/11)$&$\ZZ_{11}$ &$r=16,22$\\\hline
\multirow{2}*{36+4}&$\{1,3,4,5,7,8,11,$
&& &\\&$12,15,16,19,23\}$\\\hline
42+6,14,21&$\{1,5,11\}$&$\Sigma(2,3,7)$& 0& $r=1$\\\hline
60+12,15,20&$\{1,2,7,11,13,14\}$&$\Sigma(3,4,5)$& 0& $r=13$\\\hline
70+10,14,35&$\{1,3,9,11,13,23\}$&$\Sigma(2,5,7)$& 0& $r=11$\\\hline
78+6,26,39&$\{1,5,7,11,17,23\}$&$\Sigma(2,3,13)$& 0& $r=7$\\\bottomrule\end{tabular}}}
\end{center}\caption{\label{tab:opt2}Optimal mock Jacobi thetas of non-Niemeier type and examples of the relevant 3-manifolds.}
\end{table}

This class of 39 examples is studied and classified in \cite{omjt} as the only {\em optimal mock Jacobi forms of weight one} with non-transcendental coefficients.
To explain what they are, recall that so far we always encounter  false theta functions that are Eichler integrals of weight 3/2 unary theta functions transforming according to Weil representations (cf. \S\ref{sec:WeilRep} and \S\ref{sec:examples}).
Following the quantum modular form analysis in \S\ref{sec:QMF},  on the other side of the plane they
  correspond to weight 1/2 mock modular forms which are vector-valued and transforming according to the dual Weil representations.
A succinct way to say this is they are mock Jacobi forms of weight one. (Everywhere in the present paper (mock) Jacobi forms refer to those transforming under the whole modular group $\SL(2,\ZZ)$, and not just   some proper subgroup of it.)  In other words, suppose the homological blocks of a three-manifold $M_3$ are given in terms of false theta functions of index $m$ (cf. \eq{def:partial1} and \eq{fouriercoeff_foldedpartial}), then from the analysis of the previous subsections we expect a certain index $m$ mock Jacobi form to be relevant for $-{M_3}$.

Given a  vector-valued mock modular form $h=(h_r)$, $r=1,\dots, m-1$, with completion $\hat h=(\hat h_r)$ (cf. Definition \ref{def:mock}), we say that its combination with the index $m$ theta functions  (cf. \eq{def:usualtheta})
  \be\label{def:mockJac}
  \psi (\t,z)= \sum_{r=1,\dots ,m-1} h_r(\t)  \left(\theta_{m,r}-\theta_{m,-r}\right) (\t,z)
  \ee
is a {\em mock Jacobi form} of index $m$ and weight one if its non-holomorphic completion
 \be
  \hat \psi (\t,z)= \sum_{r=1,\dots ,m-1}{ \hat h}_r(\t)  \left(\theta_{m,r}-\theta_{m,-r}\right) (\t,z)
  \ee
  transforms as a usual Jacobi form (of index $m$ and weight one). We refer to, for instance \cite{eichler_zagier} and \cite{Dabholkar:2012nd,omjt}, for background on  Jacobi forms and mock Jacobi forms repsectively. Note that the opposite sign in the theta function factor reflects the anti-invariance of Jacobi forms of odd weights under $z\leftrightarrow -z$.

  From a number theory point of view,  weight one mock Jacobi forms are rather special.
First, in a sense we will make precise shortly, almost all of them have transcendental coefficients \cite{MR2726107} and are therefore not related to any counting problem in topology and physics.

  Second, as we have seen in the previous subsections, an important property
of (mock) modular forms is their behavior at the cusps $\QQ\cup \{i\infty\}$. At weight one mock Jacobi form is in fact uniquely determined by its  poles \cite{Sko_Thesis,Dabholkar:2012nd}.  The {\em optimal} choice of the poles, for a given index $m$, is given by
  \be\label{def:optpoles}
  q^{1\over 4m}h_r=  O(1).
  \ee

In \cite{omjt} it was shown that the space of weight one mock Jacobi forms, of {any} index $m\in \ZZ_{>0}$ and having  1. the optimal poles \eq{def:optpoles}
and 2. non-transcendental coefficients, is surprisingly finite-dimensional (34-dimensional to be precise). Moreover, there are 39 special vectors in this 34-dimensional space, distinguished by their symmetries, which span the space. (Five of them are not linearly independent of the rest.)
They are labelled by the same pair $(m,K)$ that we used in \S \ref{sec:WeilRep} to define sub-representations of Weil representations.  Let's denote the corresponding mock forms by
\be\label{not:optjac} \psi^{m+K} = \sum_{r=1,\dots ,m-1} h_r^{m+K}(\t)  (\theta_{m,r}-\theta_{m,-r}).\ee Then the group $K$ dictates the symmetry of $\Psi^{m+K}$ that it is invariant under $\theta_{m,r}\mapsto \theta_{m,ra(n)}$ for every $n\in K$ (cf. \eq{Om_action} -- \eq{def:an}).   In particular, since $a(m)=-1$ we will never have a non-vanishing $\psi^{m+K}$ unless $m\not \in K$.
In fact, quite remarkably they are in one-to-one correspondence with the 39 pairs $(m,K)$ with $m\not \in K$ which define discrete subgroups $\Gamma^{m+K}$ of $\SL(2,\RR)$ with the property that $\Gamma^{m+K}\backslash \HH$ is a genus zero Riemann surface (minus finitely many points).  We refer to \cite{omjt} and \cite{Anagiannis:2018jqf} for the details.

From the above classification, we obtain 39 distinguished mock Jacobi forms $ \psi^{m+K}$, or equivalently 39 vector-valued mock modular forms $h^{m+K}=(h_r^{m+K})$, with independent components given by $r\in \sigma^{m+K}$ (cf. \S\ref{sec:WeilRep}). They have three further striking number theoretic properties \cite{omjt} of great importance to the problems at hand\footnote{Another noteworthy property, though not directly related to the present application, is the fact that all Ramanujan's mock theta functions (up to modular forms) can be expressed in terms of these 39 mock Jacobi forms.}: 
\vspace{-5pt}
\begin{enumerate}
\item Integral coefficients;\vspace{-6pt}
\item Rademacher sums;\vspace{-6pt}
\item Theta function shadows.
\end{enumerate}
\vspace{-5pt}
Although we only demanded the coefficients to be non-transcendental, with a suitable normalization they are in fact all integral! The first dozens of coefficients can be found in \cite{omjt} and \cite{MUM}.
The properties of the coefficients further divide the 39 cases into two groups: the forms in the first group, called the {\em Niemeier type}, have {\em nonnegative}  coefficients of $ h_r^{m+K}$ for all non-polar terms in the $q$-expansion, and are in one-to-one correspondence with the 23 Niemeier lattices and play the role of the graded dimensions of the finite group modules for umbral moonshine  \cite{UM,MUM}. The second group contains the other 16 cases which have both positive and negative Fourier coefficients. The corresponding $m+K$ of the two groups are tabulated in Table \ref{tab:opt1} and \ref{tab:opt2} respectively.
Notice that in the first group, not all of them correspond to irreducible Weil representations.

The second property states that they can be constructed as vector-valued Rade-macher sums, whose simpler, single-valued version we have reviewed in \S\ref{subsec:Rad} as a way to interpolate between upper- and lower-half plane. This means that the discussion \S\ref{subsec:Rad} is appliable for these functions, and they are related to the Eichler integral of their shadows, via the $q\leftrightarrow q^{-1}$ transformation  discussed in \S\ref{sec:QMF} and \S\ref{subsec:Rad}. The good news for us is then that for all these 39 mock Jacobi forms, the shadows are given by weight 3/2 unary theta functions of the form \eq{unary theta}, and their Eichler integrals are precisely the false theta functions $\Psi^{m+K}_r$ that we encounter. 
The precise form of the shadows can be found in \cite{omjt} and \cite{MUM}.
From the dominant role of the Eichler integrals $\Psi^{m+K}_r$ of them \eq{fouriercoeff_foldedpartial} in the homological blocks for Seifert manifolds with three singular fibres, as demonstrated in \S\ref{sec:examples}, we expect these 39 examples to be relevant for the same manifolds with reversed orientation.
Indeed, for almost all of them we can easily find three-manifolds for which the homological blocks are given by the corresponding $\Psi^{m+K}_r$ for some $r$. We   tabulate some of them in Table \ref{tab:opt1}--\ref{tab:opt2}.

Finally, we comment that
\be\label{list:finiteqto1}
\lim_{\tau \to 0}h^{m+K}_r(\tau) =O(1)
\ee
for the following $(m,K)$ and $r$:
\begin{itemize}
\item 6+2 , $r=1$
\item 10+2 , $r=1,3$
\item 18+2 , $r=1,3,5,7$,
\end{itemize}
which can easily be verified from the known behaviour  \eq{def:optpoles} of $h^{m+K}_r$ near $\tau \to i\infty$ and by computing the $S$-matrix \eq{Smatrix}.
For instance,
$h^{6+2}_1(\tau) = - q^{-{1\over 24}} f(q)$ is given by the order three mock theta function of Ramanujan $f(q)$, which as we have commented in \eq{6p2oddroot}--\eq{6p2evenroot} has a finite value at $q\to 1$.
As mentioned in the beginning of the subsection, this fact gives these mock modular forms the distinguished status that their asymptotic expansion near $\tau \to 0$ coincides with the corresponding Ohtsuki series on the nose.

\section{Beyond false}
\label{sec:4sing}

In this section, we study Seifert manifolds with four singular fibers.
It turns out that structure of the homological blocks are very analogous as the cases with three singular fibres. 
The novelty is that they have the following ``building blocks'', playing a similar role as the false theta functions $\Psi^{m+K}_r$ in the previous cases, are a mix between Eichler integrals of weight 1/2 and weight 3/2 theta functions (cf. \eq{weight1/2eich})
\be
\begin{gathered}\label{def:4singbuilding}
B_{m,r}(\tau) \equiv \frac{1}{2m}\bigg[ \Phi_{m,r}(\tau) - r \Psi_{m,r}(\tau) \bigg]\\
B^{m+K}_r(\tau) =2^{|K|-1} \sum_{r' \xmod 2m} P^{m+K}_{rr'}B_{m,r'}(\tau). \\
\end{gathered}
\ee
We provide a non-spherical example $M(-2;\frac{1}{2},\frac{2}{3},\frac{2}{5},\frac{2}{5})$ and compute its asymptotic expansion by exploiting its modular-like properties.

\subsubsection{The building blocks}

We proceed by analogy with 3-fiber examples to identify the ``building blocks.'' For Brieskorn homology spheres $\Sigma(p_1,p_2,p_3)$, their WRT invariants decompose into false theta functions \cite{hikami2011decomposition, hikami1}:
\be\label{eqn:WRT3sing}Z_{CS} (\Sigma(p_1,p_2,p_3)) = \frac{q^{-\phi/4}}{i\sqrt{8k}} \bigg[ \sum_{r=1}^{m-1} \chi^{(1,1,1)}_{2m}(r) \Psi_{m,r}(\tau) + H\Big(-1+\sum_{j=1}^3 \frac{1}{p_j}\Big)q^{1/120}\bigg],\ee
where $m = \prod_j p_j$, and $H$ is the heaviside step-function. The $2m$-periodic function $\chi^{\vec{l}}_{2m}(r)$ is defined from $n$-dimensional vectors $\vec{l} = (l_1, \cdots, l_n)$ and $\vec{p} = (p_1, \cdots, p_n)$ satisfying $0 < l_j < p_j$: 
$$\chi^{\vec{l}}_{2m}(r) = \begin{cases} -\displaystyle \prod_{j=1}^{n} \epsilon_j & \text{if} \ \ r \equiv m\bigg(1+\sum_j \frac{\epsilon_j l_j}{p_j} \bigg) \xmod 2m, \ \ \text{where} \ \ \epsilon_j = \pm1 \\ 0 & \text{otherwise.} \end{cases}$$
Thus, $\chi^{(1,1,1)}_{2m}(r)$ in \eq{eqn:WRT3sing} is given by $n=3$, $\vec{l} = (1,1,1)$, and $\vec{p} = (p_1,p_2,p_3)$. One can observe that partial theta functions play the role of basic building blocks for the WRT invariants of Brieskorn homology spheres, with the latter determined by $\chi^{(1,1,1)}_{2m}(r)$.

It was shown in \cite{Hikami4sing1, Hikami4sing2},  for four-singularly fibered Seifert homology spheres $\Sigma(p_1,p_2,p_3,p_4)$ the quantity $Z_{CS} (M_3)$ can be expressed in terms of partial theta functions and a weight $1/2$ Eichler integral. Then, one can similarly extract basic building blocks by pulling out $\chi_{2m}^{(p_1-1,1,1,1)}(r)$\footnote{Spherical Seifert manifolds are uniquely determined by the orders of their singular fibers. In our convention, their Euler characteristic is $-1/\prod_{j}p_j$: following \cite{NR}.}
\begin{multline}
Z_{CS} (\Sigma(p_1,p_2,p_3,p_4)) = \\
\frac{q^{-\phi/4}}{i\sqrt{8k}} \bigg[ \sum_{r = 1}^{m-1} \chi_{2m}^{(p_1-1,1,1,1)}(r) \frac{1}{2m}\Big( \Phi_{m,r}(\tau) -r \Psi_{m,r}(\tau) \Big) \\
+ H\Big(-1+\sum_j \frac{1}{p_j}\Big) \Psi_{m,(2m - \sum_{j} m/pj)}(\tau) \bigg],
\label{eqn:4singBuilding}
\end{multline}
where $\Phi_{m,r}(\tau)$ are the weight $1/2$ Eichler integrals \eq{def:eichler} of the weight 1/2 theta functions  (cf. \eq{unary theta})
$$\theta^0_{m,r}(\tau):= \theta_{m,r}(\t,z) \lvert_{z=0}\; =\displaystyle\sum_{\substack{\ell\in\ZZ\\\ell=r\xmod 2m}} \, q^{\ell^2/4m}.$$ 
Explicitly, we have 
\begin{gather}\label{weight1/2eich}
\begin{split}
\Phi_{m,r}(\tau) = \sum_{n \geq 0} n \psi'^{(r)}_{2m}(n) q^{n^2/4m}, \\
\text{where} \quad \psi'^{(r)}_{2m}(n) = \begin{cases} 1  & \text{if} \quad n \equiv \pm r \mod 2m \\ 0 & \text{otherwise.} \end{cases}.
\end{split}
\end{gather}
Similar to \eq{def:coeff_partial1}, in terms of projectors \eq{def1:proj} we simply have $\psi'^{(r)}_{2m}(n) =2 (P_m^+(m))_{r,n}$.

The expression \eq{eqn:4singBuilding} can be best understood in comparison with three-singular fiber cases:
\begin{equation}
\begin{aligned}
\chi_{2m}^{(p_1-1,1,1,1)}(r) &\longleftrightarrow \chi^{(1,1,1)}_{2m}(r)      \\
\frac{1}{2m}\Big( \Phi_{m,r}(\tau) - r \Psi_{m,r}(\tau) \Big) &\longleftrightarrow \Psi_{m,r}(\tau) \\
H\Big(-1+\sum_j \frac{1}{p_j}\Big) \Psi_{m,(2m - \sum_{j} m/pj)}(\tau)  &\longleftrightarrow H\Big(-1+\sum_j \frac{1}{p_j}\Big) q^{1/120}
\end{aligned}
\end{equation}

Consequently, we propose that WRT invariants of Seifert manifolds with four singular fibers (not necessarily integral homology spheres) decompose into the following building blocks:
$$B_{m,r}(\tau) \equiv \frac{1}{2m}\bigg[ \Phi_{m,r}(\tau) - r \Psi_{m,r}(\tau) \bigg].$$
Note that while $\Psi_{m,r} = -\Psi_{m,-r}$, we have $B_{m,r} = B_{m,-r}$. As a result, while $m\not\in K$ for the pair $m+K$ relevant for the examples in \S\ref{sec:examples}, in this case we must have $m\in K$.
The modular-like property of $B_{m,r}$ follows from those of $\Psi_{m,r}$ and $\Phi_{m,r}$:
\be \label{eqn:psiModularTransform}
\begin{gathered}
\frac{1}{\sqrt{k}}\Psi_{m,r}(1/k) + \frac{1}{\sqrt{i}}\sum_{r'=1}^{m-1}\sqrt{\frac{2}{m}} \sin \frac{rr' \pi}{m}  \Psi_{m,r'}(-k)  = \sum_{n \geq 0}\frac{c_n}{n!} \bigg( \frac{\pi i }{2m} \bigg)^n k^{-n-\tfrac{1}{2}}, \\
\Psi_{m,r}(-k) = \left( 1- \frac{r}{m} \right) e^{-2 \pi i k r^2/4m}, \quad \text{where} \quad \frac{\sinh (m-r)z }{\sinh mz } = \sum_{n = 0}^{\infty} \frac{c_n}{2n!} z^{2n}.
\end{gathered}
\ee
\begin{multline}
\frac{1}{\sqrt{k}}\Phi_{m,r}(1/k) + \frac{k}{\sqrt{i^3}}\sum_{r'=1}^{m-1}\sqrt{\frac{2}{m}} \frac{r'(m-r')}{m} \cos \frac{rr' \pi}{m} e^{- 2 \pi i k \frac{(r')^2}{4 m}} \\
= \frac{m k}{\pi i} \sum_{n=0}^{\infty} \frac{c'_n}{n!} \bigg( \frac{\pi i }{2m} \bigg)^n k^{-n-\tfrac{1}{2}}, \quad \text{where} \quad \frac{\partial}{\partial r} \frac{\sinh (m-r)z }{\sinh mz } = \sum_{n = 0}^{\infty} \frac{c'_n}{2n!} z^{2n}. \label{eqn:psiprime}
\end{multline}
The above property can be employed to compute the transseries expression of $Z_{CS} (M_3)$. 

Note that the relation between characters of singlet $(1,p)$ logarithmic vertex algebras and homological blocks of Seifert manifolds with three singular fibers persists to the present case. One observes a close relation between characters of singlet $(p_+,p_-)$ vertex algebras \cite{BriKasMil,MR3203897} and the homological blocks of Seifert manifolds with four singular fibers studied in this section, strengthening the observed connection between logarithmic algebras.

It would be also interesting to explicitly construct building blocks for generic number of fibers, and we leave it for future works.

\subsubsection{Example: $M(-2;\frac{1}{2},\frac{2}{3},\frac{2}{5},\frac{2}{5})$}

We apply the modularity dictionary to homological blocks of a non-spherical, four-singularly fibered Seifert manifold. The Seifert manifold $M(-2;\frac{1}{2},\frac{2}{3},\frac{2}{5},\frac{2}{5})$ has the following plumbing graph:

\begin{equation}
\begin{array}{cccc}
& \overset{\displaystyle{-2}}{\bullet} & & \\
& \vline & & \\
& \overset{\displaystyle{-2}}{\bullet} & & \\
& \vline & & \\
\overset{\displaystyle{-2}}{\bullet}
\frac{\phantom{xxx}}{\phantom{xxx}}
& \underset{\displaystyle{-2}}{\bullet} &
\frac{\phantom{xxx}}{\phantom{xxx}}
\overset{\displaystyle{-3}}{\bullet} &
\frac{\phantom{xxx}}{\phantom{xxx}}
\overset{\displaystyle{-2}}{\bullet} \\
& \vline & & \\
& \underset{\displaystyle{-3}}{\bullet} & & \\
& \vline & & \\
& \underset{\displaystyle{-2}}{\bullet} & & \\
\end{array}
\end{equation}
As before, we compute homological blocks and $S^{(A)}$:
\be
\begin{gathered}
\text{CS}(a) = \begin{pmatrix} 1 & \frac{1}{5} & \frac{9}{5}\end{pmatrix}, \quad S^{(A)} = \frac{1}{\sqrt{5}}
\begin{pmatrix}
1 & 1 & 1 \\
2 & \frac{-1-\sqrt{5}}{2} & \frac{-1+\sqrt{5}}{2} \\
2 & \frac{-1+\sqrt{5}}{2} & \frac{-1-\sqrt{5}}{2}
\end{pmatrix}, \\
\widehat{Z}_0(q) = q^{7/2}(1-q^{11}+q^{14}-q^{19}-q^{33}+q^{40}-q^{45}+2q^{53}+q^{74}+\cdots), \\
\widehat{Z}_1(q) = 0, \\
\widehat{Z}_2(q) = 2q^{93/10}(-1+q^{15}+q^{25}-q^{50}-2q^{65}+2q^{120}-2q^{165}-3q^{190}+\cdots) .
\end{gathered}
\ee
By the prescription of modularity dictionary, we can easily see that $m =30$. Then, homological blocks correspond to the Weil representation $\sigma = 30+5,6,30$, whose projector is explicitly given by
$$P^{30+5,6,30} = P^{+}_{30}(5) P^{+}_{30}(6) P^{-}_{30}(15),$$
which leads to, using \eq{def:4singbuilding}, 
\be \begin{aligned} 
B^{30+5,6,30}_7(\tau) &= \left(B_{30,7}-B_{30,13} + B_{30,17} - B_{30,23}\right)(\tau), \\
B^{30+5,6,30}_5(\tau) &= \left(B_{30,5}-B_{30,23}\right)(\tau).
\end{aligned}
\ee
In terms of these, the homological blocks read
\begin{equation}
\begin{aligned}
\widehat{Z}_0(q) &= q^{-109/120}\left( \Psi_{30,23}(\tau) - B^{30+5,6,30}_7(\tau) \right)  \\
\widehat{Z}_1(q) &= 0 \\
\widehat{Z}_2(q) &= 2q^{-109/120}B^{30+5,6,30}_5(\tau).
\end{aligned}
\label{eqn:Zhat_decomp}
\end{equation}
By the modular-like properties of the building blocks, we obtain the transseries summarized in Table~\ref{table:2355}, where an overall factor of $-i q^{-109/120}/2\sqrt{2}$ is omitted as usual.

\begin{table}[h]\begin{center}\scalebox{0.9}{
\begin{tabular}{c c c c}
CS action & stabilizer & type & transseries \vspace{3pt}
\\ \toprule
$0$ & $SU(2)$ & central & $e^{2 \pi i k \cdot 0} \bigg( \frac{4 \pi i}{5}k^{-3/2} + \mathcal{O}(k^{-5/2}) \bigg) $ \\[1.5ex]
$\frac{1}{5}$ & $U(1)$ & abelian & $e^{2 \pi i k \frac{1}{5}} \bigg( \frac{5-\sqrt{5}}{6} k^{-1/2} + \mathcal{O}(k^{-3/2}) \bigg) $ \\[1.5ex]
$\frac{9}{5}$ & $U(1)$ & abelian & $e^{2 \pi i k \frac{9}{5}} \bigg( \frac{5+\sqrt{5}}{6} k^{-1/2} + \mathcal{O}(k^{-3/2}) \bigg) $ \\[1.5ex]
$-\frac{1}{120}$ & $\pm1$ & non-abelian, real & $e^{-2 \pi i k \frac{1}{120}} \frac{e^{3 \pi i/4}}{150\sqrt{15}}(-25+5\sqrt{5}+24\sqrt{3}\cos\frac{\pi}{10})$ \\[1.5ex]
$-\frac{4}{120}$ & $\pm1$ & non-abelian, real & $-e^{-2 \pi i k \frac{4}{120}} \frac{e^{3 \pi i/4}}{4}(1+\sqrt{5})$ \\[1.5ex]
$-\frac{16}{120}$ & $\pm1$ & non-abelian, real & $e^{-2 \pi i k \frac{16}{120}} \frac{e^{3 \pi i/4}}{4}(1-\sqrt{5})$ \\[1.5ex]
$-\frac{25}{120}$ & $\pm1$ & non-abelian, real & $e^{-2 \pi i k \frac{25}{120}} \frac{4 e^{3 \pi i/4}}{5} \sqrt{1-\frac{2}{\sqrt{5}}} $ \\[1.5ex]
$-\frac{40}{120}$ & $\pm1$ & non-abelian, real & $e^{-2 \pi i k \frac{40}{120}} e^{3 \pi i/4}$\\[1.5ex]
$-\frac{49}{120}$ & $\pm1$ & non-abelian, real & $-e^{-2 \pi i k \frac{49}{120}} \frac{e^{3 \pi i/4}}{30 \sqrt{15}}(25+5\sqrt{5}+24 \sqrt{3}\cos\frac{3\pi}{10})$ \\[1.5ex]
$-\frac{73}{120}$ & $\pm1$ & non-abelian, real & $-e^{-2 \pi i k \frac{73}{120}} \frac{8 e^{3 \pi i/4}}{5\sqrt{5}}\cos\frac{3\pi}{10}$ \\[1.5ex]
$-\frac{76}{120}$ & $\pm1$ & non-abelian, real & $e^{-2 \pi i k \frac{76}{120}} \frac{e^{3 \pi i/4}}{4}(1-\sqrt{5})$ \\[1.5ex]
$-\frac{81}{120}$ & $\pm1$ & non-abelian, real & $e^{-2 \pi i k \frac{81}{120}} \frac{e^{3 \pi i/4}}{3\sqrt{3}}(1-\sqrt{5})$ \\[1.5ex]
$-\frac{97}{120}$ & $\pm1$ & non-abelian, real & $e^{-2 \pi i k \frac{97}{120}} \frac{8 e^{3 \pi i/4}}{5\sqrt{5}} \cos\frac{\pi}{10}$ \\[1.5ex]
$-\frac{105}{120}$ & $\pm1$ & non-abelian, real & $e^{-2 \pi i k \frac{105}{120}} \frac{4e^{3 \pi i/4}}{3\sqrt{3}}$ \\[1.5ex]
$-\frac{9}{120}$ & $\pm1$ & non-abelian, complex & $0$ \\[1.5ex]
$-\frac{64}{120}$ & $\pm1$ & non-abelian, complex & $0$ \\[1.5ex]
$-\frac{100}{120}$ & $\pm1$ & non-abelian, complex & $0$\vspace{2pt}\\ \bottomrule
\end{tabular}}
\caption{Transseries and classification of flat connections on $M(-2;\tfrac{1}{2},\tfrac{2}{3},\tfrac{2}{5},\tfrac{2}{5})$.}
\label{table:2355}
\end{center}
\end{table}

\section{Discussions and open questions}

In this paper we discussed the following surprising features of half-indices of certain ${\cal N}=2$ 3d supersymmetric quantum field theories, which are also
the homological blocks \cite{GPPV} of a family of three-manifolds.
In the first part of the paper we discussed three different $\SL(2,\ZZ)$ (projective) representations  we encountered in the problem and make use of them to compute topologically and physically interesting quantities.
Second we propose the relevance of the  false--mock pair for our problem, making use of their relation to quantum modular forms.

We will end the main part of the paper with a list of open questions and future directions:
\begin{itemize}

\item 
Though the relevance of the false--mock pair is manifest, there are still a few important puzzles remaining. Just purely from the number theory point of view there are two ambiguities in identifying the correct $\widehat Z_a$ on the mock side:
\begin{enumerate}\vspace{-5pt}
\item 
As summarised in
Figure \ref{diagram_mockquantumfalse}, the Eichler integral associated to a mock function only depends on its shadow and is therefore insensitive to the addition of a purely modular form. Relatedly, at the end of \S\ref{subsec:qhyper}, we have seen that from the $q$-hypergeometric perspective there are various ambiguities when going between upper- and lower-half planes: two $q$-hypergeometric series can define the same function on one side and different functions on the other.
This is related to the so-called ``expansion of zero'' described by Rademacher\cite{rad_exp_zero}. In our context of weight one Jacobi forms, the ambiguity is equivalent to the ambiguity of specifying the poles of the function (see \cite{Sko_Thesis} and \cite{omjt}).
\item Moreover, once a mock modular form $f$ is chosen, one still need to choose the modular subtraction $G_x$ at the cusp $x$ (cf. \eq{asymp_mock}). While for comparison with perturbative  Chern--Simons one only needs the subtraction for the cusp $\tau \to 0$, in order to literally compare with all the WRT invariants one would need the subtraction at all roots of unity of the form $\tau \to {1\over k}$. (Of course, this is the description of the ambiguity when considering the mock form $f$ as a single-valued mock modular form for a subgroup $\Gamma < \SL(2,\ZZ)$. When described in terms of vector-valued mock modular forms for $\SL(2,\ZZ)$, the corresponding ambiguity is that of specifying the  modular subtractions for all components of the vector-valued function.)
\end{enumerate}
Recall that, while one does not necessarily need to care about these ambiguities for the pupose of reproducing the perturbative Chern--Simons data, the actual $q$-series are physically very  meaningful! 
Clearly, in order to make general predictions for the homological blocks for general three-manifolds and in order to better understand the general modularity structure of 3d ${\cal N}=2$ theories, it is crucial to better understand the above ambiguities, and hopefully to find sufficiently powerful criteria from physics and topology to eliminate the ambiguities.
Also from the context of this work, the relation to Habiro ring \cite{MR2105705} appears to be an important lead.

\item
It would be interesting to compute, {\it e.g.} via resurgence \cite{Gukov:2016njj},
the $q$-series invariants $\widehat{Z}_a (M_3)$ for hyperbolic 3-manifolds and
test the conjecture in \S\ref{sec:logCFT}, namely whether in such cases $\widehat{Z}_a (M_3)$
are related to characters of logarithmic vertex algebras which are not $C_2$-cofinite.

\item
In this work we note the important role played by the Weil representations, labelled by a pair $m$ and $K\subset \Ex_m$, in our problem. While $m$ can phenomelogically be determined by the topological data \eq{eqn:mformula}, we do not know what the explicit relation between $K$ and 3-manifold topology is. 
More conceptually, it would be great to understand the origin of these Weil representations in our topological/physical problem. 
Furthermore, 
in \S\ref{sec:mod} we discussed three different $S$-matrices and three different ``$\SL(2,\ZZ)$ representations'' that play an important role in our story. 
We do expect these representations to be inter-related and one obvious open problem is to understand how exactly they are related.

\item In our story, mock modular forms and interesting modular structures
emerge from the physics of 3d $\N=2$ theories and BPS states.
In particular, our key players $\widehat{Z}_a (q)$ are ``counting'' BPS states.
It would be interesting to find relations (dualities) to other physics
problems where similar modular structures appeared,
{\it e.g.} \cite{Troost:2010ud,Alim:2010cf,Cheng:2014owa,Dabholkar:2012nd,Alexandrov:2018lgp}.

\item 
In this work we have mainly focused on examples of Seifert manifolds with $\ell$ singular fibers, where $\ell=3$ and the homological blocks are given by false thetas. These type of functions also appear as characters of modules of singlet $(1,p)$ logarithmic vertex algebras. In \S\ref{sec:4sing} we briefly discussed the case of $\ell=4$, where the homological blocks are composed of building blocks which contain false theta functions corresponding to a mix of quantum modular forms of weight 3/2 and weight 1/2. Moreover, the relation to logarithmic vertex algebras persists and this time the blocks appear to be related to characters of $(p_+,p_-)$ singlet vertex algebras. However, there is clearly much more to explore. One interesting question is to identify the building blocks of homological blocks for Seifert manifolds with general $\ell$, and the potential relation to logarithmic vertex algebras. Next, it would be intestering to explore the modularity structure of homological blocks for plumbed 3-manifolds which are non-Seifert. Third, it is very conceivable that there exists a nice relation between higher rank invariants and higher-depth quantum modular forms \cite{BriKasMil}, generalizing the $\SL(2,\CC)$ story which is the focus of the present paper. Finally, one may also investigate the modular-like properties of the half-indices arising from 3d theories that are not coming from three-manifolds. 

\item The examples given in \S\ref{sec:optimal} suggest the relevance of the 39 optimal mock Jacobi theta functions in our topological and physical problems. At the same time, these mock functions also play the role of the graded dimensions of finite group representations in the context of the still mysterious umbral moonshine \cite{UM,MUM}. One natural question is whether there is a relation between our setup and the moonshine finite groups. 

\end{itemize}

\section*{Acknowledgements}
We thank
D.~Adamovic, T.~Creutzig, T.~Dimofte, J. Duncan, P.~Etingof,
B.~Feigin, D.~Kazhdan, S.~L\"obrich, C.~Manolescu, D.~Pei, P.~Putrov,
L.~Rolen, C.~Schweigert, C.~Vafa and D.~Zagier
for helpful discussions.
The work of M.C. is supported by ERC starting grant H2020 ERC StG \#640159.
S.C. and S.G. are supported by the Walter Burke Institute for Theoretical Physics,
by the U.S. Department of Energy, Office of Science, Office of High Energy Physics, under Award No.\ DE{-}SC0011632,
and by the National Science Foundation under Grant No.~NSF DMS 1664240.
The work of S.C. is also supported in part by Samsung Scholarship.
The work of S.G. is also supported in part by Laboratory of Mirror Symmetry NRU HSE, RF Government grant, ag. No.~14.641.31.0001.
The work of F.F. is supported in part by the MIUR-SIR grant RBSI1471GJ ``Quantum Field Theories at Strong Coupling: Exact Computations and Applications''.
S.H. is supported by the National Science and Engineering Council of Canada,
an FRQNT new university researchers start-up grant, and the Canada Research Chairs program. M.C. would also like to thank LPTHE, Jussieu Paris, for hospitality during the final stage of this work.

\appendix

\section{Invariance of convergence under 3d Kirby moves}
\label{appendix:3dkirby}
In this section, we prove that 3d Kirby moves (Figure~\ref{fig:kirby}) preserve the domain of convergence of homological blocks.
\begin{figure}[htb]
\centering
\includegraphics{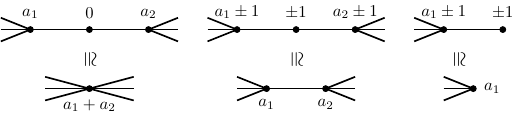}
\caption{3d Kirby moves for plumbed manifolds. The resulting plumbed manifolds $M$ and $M'$ are homeomorphic.}
\label{fig:kirby}
\end{figure}

Consider the bottom left graph of Figure~\ref{fig:kirby}. We may choose the basis in which the adjacency matrix $M$ has framing coefficient $a_1+a_2$ in the $(i,i)$-th coordinate.
By Lemma~\ref{thm:1-boundy}, homological blocks of the manifold plumbed along the bottom left graph would have the following asymptotic behavior for the $q$-exponents:
\begin{equation}
q^{-\frac{(\ell,M^{-1}\ell)}{4}} = q^{ -\frac{(M^{-1})_{ii}\ell_i^2}{4}+O(1)}, \quad \text{as} \quad |\ell| \rightarrow \infty 
\label{eqn:liGrowth}
\end{equation}
for the relevant terms of the sum. 
Next, we consider the adjacency matrix $M'$ of the top left graph, which has the following $[i-1,i+1] \times [i-1,i+1]$-submatrix:
$$\begin{pmatrix} a_1 & 1 & 0 \\ 1 & 0 & 1 \\ 0 & 1 & a_2 \end{pmatrix} \subset M'.$$
A simple linear algebra shows that its inverse has the following $[i-1,i+1]\times[i-1,i+1]$-submatrix:
$$\begin{pmatrix} \frac{(M^{-1})_{ii}}{2} & \cdots & -\frac{(M^{-1})_{ii}}{2} \\ \cdots & \cdots & \cdots \\ -\frac{(M^{-1})_{ii}}{2} & \cdots & \frac{(M^{-1})_{ii}}{2} \end{pmatrix} \subset (M')^{-1}.$$
Therefore, the homological blocks associated to the top left plumbing graph would have the $q$-exponents with the following asymptotic behavior:
\begin{equation}
q^{-\frac{(\ell, (M')^{-1}\ell)}{4}}=q^{ -\frac{(M^{-1})_{ii}(\ell_{i-1}-\ell_{i+1})^2}{8}+O(1)}, \quad \text{as} \quad |\ell| \rightarrow \infty
\label{eqn:lsplitGrowth}
\end{equation}
for the relevant terms of the sum. In particular, the asymptotic behavior depends only on $(\ell_{i-1}-\ell_{i+1})$, which is playing the role of $\ell_i$ in \eq{eqn:liGrowth}. Furthermore, the asymptotic behaviors of both \eq{eqn:liGrowth} and \eq{eqn:lsplitGrowth} are proportional to $(M^{-1})_{ii}$. Thus, we may conclude that the first Kirby move preserves the domain of convergence of homological blocks. We can analogously work out the diagonal elements of $(M')^{-1}$ for the two remaining Kirby moves and observe the invariance.

\section{Further examples}
In this section, we provide further examples of Seifert manifolds whose homological blocks are given in terms of the false theta functions $\Psi^{m+K}_r$ and whose data about the flat connections can be inferred from the modularity dictionary as discussed in \S\ref{sec:resurgence}.
\begin{table}[h]\begin{center}\scalebox{0.9}{
\begin{tabular}{c c c c}
CS action & stabilizer & type & transseries \vspace{3pt}
\\ \toprule
$0$ & $SU(2)$ & central & $e^{2 \pi i k \cdot 0} \bigg( \frac{4 \pi i}{3\sqrt{3}}k^{-3/2} +\frac{19\pi^2}{9\sqrt{3}}k^{-5/2}+ \mathcal{O}(k^{-7/2}) \bigg)$  \\[1.5ex]
$\frac{1}{3}$ & $U(1)$ & abelian & $e^{2 \pi i k \frac{1}{3}} \bigg( \sqrt{3}k^{-1/2} + \frac{5 \pi i}{12\sqrt{3}} k^{-3/2} + \mathcal{O}(k^{-5/2}) \bigg)$ \\[1.5ex]
$-\frac{1}{24}$ & $\pm1$ & irreducible, real & $e^{-2\pi ik\frac{1}{24}} e^{\frac{3 \pi i }{4}} \cdot (-2)  $ \vspace{2pt}\\ \bottomrule
\end{tabular}}
\caption{Transseries and classification of flat connections on $M( -2;\tfrac{1}{2},\tfrac{2}{3},\tfrac{2}{3}) =  S^3_{-3}(\mathbf{3}^{\ell}_1)$ up to an overall factor of $-iq^{-25/24}/2\sqrt{2}$. The corresponding Weil representation is $m+K = 6+3$.}
\end{center}
\end{table}

\begin{table}[h]
\begin{center}\scalebox{0.9}{
\begin{tabular}{c c c c}
CS action & stabilizer & type & transseries \vspace{3pt}
\\ \toprule
$0$ & $SU(2)$ & central & $e^{2 \pi i k \cdot 0} \bigg( \pi i \sqrt{2} k^{-3/2} + \frac{259 \pi^2}{20\sqrt{2}}k^{-5/2} + \mathcal{O}(k^{-7/2}) \bigg)$  \\[1.5ex]
$\frac{1}{2}$ & $SU(2)$ & central & $e^{2 \pi i k \frac{1}{2}} \bigg( \pi i \sqrt{2} k^{-3/2} + \frac{259 \pi^2}{20\sqrt{2}}k^{-5/2} + \mathcal{O}(k^{-7/2}) \bigg)$ \\[1.5ex]
$-\frac{9}{80}$ & $\pm1$ & non-abelian, real & $e^{-2\pi ik\frac{9}{80}} e^{\frac{3 \pi i }{4}} \cdot  \big( \frac{6+ 2\sqrt{5}}{5} \big)^{\frac{1}{4}} $ \\[1.5ex]
$-\frac{49}{80}$ & $\pm1$ & non-abelian, real & $e^{-2\pi ik\frac{49}{80}}e^{\frac{3 \pi i }{4}} \cdot  \big( \frac{6+ 2\sqrt{5}}{5} \big)^{\frac{1}{4}}$ \\[1.5ex]
$-\frac{1}{80}$ & $\pm1$ & non-abelian, complex & 0 \\[1.5ex]
$-\frac{121}{80}$ & $\pm1$ & non-abelian, complex & 0 \vspace{2pt}\\ \bottomrule
\end{tabular}}
\caption{Transseries and classification of flat connections on $M(-1;\frac{1}{2},\frac{1}{4},\frac{1}{5}) = S^3_{+2}(\mathbf{4_1})$ up to an overall factor of $-iq^{19/80}/2\sqrt{2}$. The corresponding Weil representation is $m+K = 20+4$.}
\end{center}
\end{table}

\begin{table}[h]\begin{center}\scalebox{0.9}{
\begin{tabular}{c c c c}
CS action & stabilizer & type & transseries \vspace{3pt}
\\ \toprule
$0$ & $SU(2)$ & central & $e^{2 \pi i k \cdot 0} \bigg( \frac{4 \pi i}{3\sqrt{3}} k^{-3/2} + \frac{103 \pi^2}{18\sqrt{3}}k^{-5/2} + \mathcal{O}(k^{-7/2}) \bigg)$  \\[1.5ex]
$\frac{2}{3}$ & $U(1)$ & abelian & $e^{2 \pi i k \frac{2}{3}} \bigg( \frac{\sqrt{3}}{2}k^{-1/2} - \frac{7 \pi i}{48\sqrt{3}} k^{-3/2} + \mathcal{O}(k^{-5/2}) \bigg)$ \\[1.5ex]
$-\frac{4}{48}$ & $\pm1$ & non-abelian, real & $e^{-2\pi ik\frac{4}{48}} e^{\frac{3 \pi i }{4}} \cdot \sqrt{2} $ \\[1.5ex]
$-\frac{25}{48}$ & $\pm1$ & non-abelian, real & $e^{-2\pi ik\frac{25}{48}}e^{\frac{3 \pi i }{4}} \cdot 1$ \\[1.5ex]
$-\frac{1}{48}$ & $\pm1$ & non-abelian, complex & 0 \vspace{2pt}\\ \bottomrule
\end{tabular}}
\caption{Transseries and classification of flat connections on $M(-1;\frac{1}{3},\frac{1}{3},\frac{1}{4})$ up to an overall factor of $-iq^{-1/48}/2\sqrt{2}$. The corresponding Weil representation is $m+K = 12+3$.}
\end{center}
\end{table}

\begin{table}[h]\begin{center}\scalebox{0.9}{
\begin{tabular}{c c c c}
CS action & stabilizer & type & transseries \vspace{3pt}
\\ \toprule
$0$ & $SU(2)$ & central & $e^{2 \pi i k \cdot 0} \bigg(\frac{4 \pi i }{5\sqrt{5}} k^{-3/2} + \frac{67 \pi^2}{25\sqrt{5}}k^{-5/2}+ \mathcal{O}(k^{-7/2}) \bigg)$  \\[1.5ex]
$\frac{2}{5}$ & $U(1)$ & abelian & $e^{2 \pi i k \frac{2}{5}} \bigg(\frac{1+\sqrt{5}}{2}k^{-1/2} - \frac{(125+61\sqrt{5}) \pi i}{200}k^{-3/2}  + \mathcal{O}(k^{-5/2}) \bigg)$ \\[1.5ex]
$\frac{3}{5}$ & $U(1)$ & abelian & $e^{2 \pi i k \frac{3}{5}} \bigg( \frac{-1+\sqrt{5}}{2}k^{-1/2} +\frac{(125-61\sqrt{5}) \pi i}{200}k^{-3/2}   + \mathcal{O}(k^{-5/2}) \bigg)$ \\[1.5ex]
$-\frac{9}{40}$ & $\pm1$ & non-abelian, real & $e^{-2\pi ik\frac{9}{40}} e^{3\pi i/4}\cdot(1+\frac{1}{\sqrt{5}}) $ \\[1.5ex]
$-\frac{25}{40}$ & $\pm1$ & non-abelian, real & $e^{-2\pi ik\frac{25}{40}} e^{3\pi i/4}\cdot\frac{4}{\sqrt{5}} $ \\[1.5ex]
$-\frac{1}{40}$ & $\pm1$ & non-abelian, complex & $0$ \vspace{2pt}\\ \bottomrule
\end{tabular}}
\caption{Transseries and classification of flat connections on $M( -1;\tfrac{1}{2},\tfrac{1}{5},\tfrac{1}{5})$ up to an overall factor of $-iq^{-19/40}/2\sqrt{2}$. The corresponding Weil representation is $m+K = 10+5$.}\end{center}
\end{table}

\begin{table}[h]\begin{center}\scalebox{0.9}{
\begin{tabular}{c c c c}
CS action & stabilizer & type & transseries \vspace{3pt}
\\ \toprule
$0$ & $SU(2)$ & central & $e^{2 \pi i k \cdot 0} \bigg( \frac{4 \pi i}{27} k^{-3/2} + \mathcal{O}(k^{-5/2}) \bigg)$  \\[1.5ex]
$0$ & $U(1)$ & abelian & $e^{2 \pi i k \cdot 0} \bigg( -k^{-1/2} +  \mathcal{O}(k^{-3/2}) \bigg)$ \\[1.5ex]
$\frac{2}{9}$ & $U(1)$ & abelian & $e^{2 \pi i k \frac{2}{9}} \bigg( (\frac{4}{3}\cos\frac{\pi}{9}-\frac{2}{3})k^{-1/2} + \mathcal{O}(k^{-3/2}) \bigg)$ \\[1.5ex]
$\frac{5}{9}$ & $U(1)$ & abelian & $e^{2 \pi i k \frac{5}{9}} \bigg( -(\frac{4}{3}\sin\frac{\pi}{18}+\frac{2}{3})k^{-1/2} + \mathcal{O}(k^{-3/2}) \bigg)$ \\[1.5ex]
$\frac{8}{9}$ & $U(1)$ & abelian & $e^{2 \pi i k \frac{8}{9}} \bigg( -(\frac{4}{3}\cos\frac{2\pi}{18}+\frac{2}{3})k^{-1/2} + \mathcal{O}(k^{-3/2}) \bigg)$ \\[1.5ex]
$-\frac{9}{72}$ & $\pm1$ & non-abelian, real & $e^{-2\pi ik \frac{9}{72}} e^{3 \pi i/4}\cdot(-2)$\vspace{2pt}\\ \bottomrule
\end{tabular}}
\caption{Transseries and classification of flat connections on $M( -2;\tfrac{1}{2},\tfrac{1}{3},\tfrac{2}{3})$ up to an overall factor of $-iq^{-45/72}/2\sqrt{2}$. The corresponding Weil representation is $m+K = 18+9$.}
\end{center}
\end{table}

\newpage{\pagestyle{empty}\cleardoublepage}

\end{document}